\documentclass[aps,pra,10pt,letterpaper,twocolumn,tightenlines,superscriptaddress,notitlepage]{revtex4-2}
\usepackage{qcircuit}
\usepackage[dvips]{graphicx}
\usepackage{amsmath,amssymb,amsthm,mathrsfs,amsfonts,dsfont}
\usepackage{subfigure, epsfig}
\usepackage{braket}
\usepackage{bm}
\usepackage{enumerate}
\usepackage{xcolor}
\usepackage{comment}
\usepackage{scalefnt}
\usepackage[colorlinks = true, 
citecolor = blue,
linkcolor = blue,
urlcolor = blue]{hyperref}

\usepackage{pagecolor}
\usepackage[T1]{fontenc}
\usepackage{lmodern}
\usepackage[utf8]{inputenc}
\usepackage{microtype}

\usepackage{tikzscale}
\usepackage{pgfplots}
\pgfplotsset{compat=newest}
\usetikzlibrary{plotmarks}
\usetikzlibrary{arrows.meta}
\usetikzlibrary{external}
\usepgfplotslibrary{patchplots}
\usepackage{grffile}

\usepackage{algorithm}  
\usepackage{algpseudocode}  

\newcommand{\rV}{\rVert}
\newcommand{\lV}{\lVert}

\newcommand{\tV}{\vert\kern-0.25ex\vert\kern-0.25ex\vert}

\newcommand{\laa}{\langle\!\langle}
\newcommand{\raa}{\rangle\!\rangle}

\newcommand{\ct}{\ensuremath{^{\dagger}}}

\DeclareMathOperator{\tr}{Tr}
\DeclareMathOperator{\Tr}{Tr}

\definecolor{applegreen}{rgb}{0.55, 0.71, 0.0}
\newcommand{\wy}[1]{{\color{applegreen}\textbf{WY:} #1}}
\newcommand{\pz}[1]{{\color{red}\textbf{PZ:} #1}}

\newcommand{\comments}[1]{}

\newenvironment{lmbis}[1]
{\renewcommand{\thelemma}{\ref{#1}$'$}%
	\addtocounter{lemma}{-1}%
	\begin{lemma}}
	{\end{lemma}}

\graphicspath{{./figure/}}

\newcommand{\cD}{\ensuremath{\mathcal{D}}}
\newcommand{\cE}{\ensuremath{\mathcal{E}}}

\newcommand{\cG}{\ensuremath{\mathcal{G}}}
\newcommand{\cH}{\ensuremath{\mathcal{H}}}

\newcommand{\cL}{\ensuremath{\mathcal{L}}}
\newcommand{\cM}{\ensuremath{\mathcal{M}}}

\newcommand{\cO}{\ensuremath{\mathcal{O}}}
\newcommand{\cP}{\ensuremath{\mathcal{P}}}

\newcommand{\cR}{\ensuremath{\mathcal{R}}}

\newcommand{\cT}{\ensuremath{\mathcal{T}}}
\newcommand{\cU}{\ensuremath{\mathcal{U}}}
\newcommand{\cV}{\ensuremath{\mathcal{V}}}

\newcommand{\bbE}{\ensuremath{\mathbb{E}}}

\newcommand{\bbG}{\ensuremath{\mathbb{G}}}

\newcommand{\coleq}{\mathrel{\mathop:}\nobreak\mkern-1.2mu=}

\newcommand{\renyi}{R$\mathrm{\acute{e}}$nyi}
\newcommand{\mc}{\mathcal}

\newcommand{\mr}{\mathrm}
\newcommand{\mbb}{\mathbb}

\newcommand{\id}{\mathrm{id}}
\newcommand{\ketbra}[2]{\vert #1 \rangle \langle #2 \vert}
\newcommand{\lket}[1]{\vert #1 \rangle\!\rangle}
\newcommand{\lbra}[1]{\langle\!\langle #1 \vert}
\newcommand{\lbraket}[2]{\langle\!\langle #1 \vert #2 \rangle\!\rangle}
\newcommand{\lketbra}[2]{\vert #1 \rangle\!\rangle\langle\!\langle #2 \vert}

\definecolor{almond}{rgb}{0.98, 0.81, 0.69}

\newtheorem{theorem}{Theorem}
\newtheorem{proposition}{Proposition}

\newtheorem{lemma}{Lemma}
\newtheorem{assumption}{Assumptions}

\newtheorem{definition}{Definition}
\newtheorem{protocol}{Protocol}

\definecolor{mycolor1}{rgb}{0.98431,0.41569,0.29020}%
\definecolor{mycolor2}{rgb}{0.45490,0.66275,0.81176}%
\definecolor{mycolor3}{rgb}{0.87059,0.17647,0.14902}%
\definecolor{mycolor4}{rgb}{0.16863,0.54902,0.74510}%
\definecolor{mycolor5}{rgb}{0.64706,0.05882,0.08235}%
\definecolor{mycolor6}{rgb}{0.01569,0.35294,0.55294}%

\definecolor{divergent1}{rgb}{0.1680, 0.5117, 0.7266}%
\definecolor{divergent2}{rgb}{0.6680, 0.8633, 0.6406}%
\definecolor{divergent3}{rgb}{0.9883, 0.6797, 0.3789}%
\definecolor{divergent4}{rgb}{0.8398, 0.0977, 0.1094}%

\newcommand{\bvec}{\boldsymbol}

\begin{document}

\title{Robust shadow estimation}

\author{Senrui Chen}
\thanks{These authors contributed equally to this work.}
\affiliation{Center for Quantum Information, Institute for Interdisciplinary Information Sciences, Tsinghua University, Beijing 100084, China}
\affiliation{Department of Electronic Engineering, Tsinghua University, Beijing 100084, China}
\affiliation{Pritzker School of Molecular Engineering, The University of Chicago, Illinois 60637, USA}
\author{Wenjun Yu}
\thanks{These authors contributed equally to this work.}
\affiliation{Center for Quantum Information, Institute for Interdisciplinary Information Sciences, Tsinghua University, Beijing 100084, China}
\author{Pei Zeng}
\email{qubitpei@gmail.com}
\affiliation{Center for Quantum Information, Institute for Interdisciplinary Information Sciences, Tsinghua University, Beijing 100084, China}
\affiliation{Pritzker School of Molecular Engineering, The University of Chicago, Illinois 60637, USA}
\author{Steven T. Flammia}
\affiliation{AWS Center for Quantum Computing, Pasadena, CA 91125, USA}

\begin{abstract}
Efficiently estimating properties of large and strongly coupled quantum systems is a central focus in many-body physics and quantum information theory. While quantum computers promise speedups for many of these tasks, near-term devices are prone to noise that will generally reduce the accuracy of such estimates.
Here, we propose a sample-efficient and noise-resilient protocol for learning properties of quantum states building on the shadow estimation scheme [Huang et. al., Nature Physics 16, 1050–1057 (2020)]. 
By introducing an experimentally-friendly calibration procedure, our protocol can efficiently characterize and mitigate noises in the shadow estimation scheme, given only minimal assumptions on the experimental conditions. 
When the strength of noises can be bounded, our protocol approximately retains the same order of sample-efficiency as the standard shadow estimation scheme, while also possesses a provable noise resilience. 
We give rigorous bounds on the sample complexity of our protocol and demonstrate its performance with several numerical experiments, including estimations of quantum fidelity, correlation functions and energy expectations, etc., which highlight a wide spectrum of potential applications of our protocol on near-term devices.
\end{abstract}
\date{\today}
\maketitle


\section{Introduction} \label{sec:Intro}



We are in the process of building large-scale and controllable quantum systems. 
This not only provides new insights and tool kits for fundamental research in quantum many-body systems~\cite{amico2008Entanglement} and the quantum nature of spacetime~\cite{qi2018gravity}, but also yields fruitful applications in computing~\cite{shor1994algorithms,grover1997quantum,lloyd1996Simulators,nielsen2011Quantum}, communication~\cite{bennett84cryptography,ekert1991cryptography,bennett1993Teleporting}, and sensing~\cite{wineland1992squeezing,giovannetti2006Metrology}. 
Learning the properties, e.g., fidelity~\cite{da2011practical,Flammia2011direct}, entanglement~\cite{brandao2005quantifying,abanin2012measuring}, and energy~\cite{peruzzo2014variational} of generated quantum states is usually a major step in many quantum benchmarking protocols and quantum algorithms. 
Among various figure of merits, robustness and efficiency are two key factors to assess the practicality of any property learning protocol.

In the Noisy Intermediate-Scale Quantum (NISQ) era~\cite{preskill2018quantum}, quantum circuits inevitably suffer from noise. 
The robustness of a property learning protocol then refers to the ability to tolerate such noise. 
In a typical property estimation process, we generate several identical copies of the target quantum states, and then measure them using some devices which might be noisy and uncharacterized.
To verify the property estimates, one has to introduce new benchmarking devices, which (in the NISQ era) will also be noisy. 
Consequently, we will be trapped into a loop of benchmarking. 
To get rid of this, at least two approaches have been proposed: One is to introduce extra assumptions on the noise model, in which case we might be able to mitigate the error~\cite{temme2017error,endo2018practical,maciejewski2020mitigationofreadout,bravyi2020mitigating}, but such assumptions may not be verifiable. 
The other is to use device-independent protocols~\cite{mckague2012robust,brunner2014Bell,guhne2009entanglement} which do not have any assumptions on the devices, but such protocols are mostly designed for some specific property learning tasks (\textit{e.g.} entanglement detection), and their requirements on devices and computational/sample complexity can be too strict to produce anything informative in practice. 

Thus, while property learning and testing leads to large efficiency gains in sample and computational complexity, one must in general have a well-characterized device for these methods to be applicable. 
Quantum tomography~\cite{ariano2007optimal,guta2020fast} is a standard method to extract complete characterization information, but it requires exponentially many samples with respect to the number of qubits. 
Several efficient tomographic schemes were proposed based on some properties of the prepared states, such as the low rank property~\cite{gross2010tomography,Flammia2012compressed}, permutation symmetry~\cite{toth2010permutationally,moroder2012permutationally}, and the locality of Schmidt decomposition~\cite{cramer2010efficient,baumgratz2013scalable}. 
Nevertheless, such assumptions are restrictive and not applicable in many cases. 
Another line of research focuses on efficiently extracting partial information of a quantum state without any prior knowledge. 
An example is the quantum overlapping tomography~\cite{cotler2020quantum, bonet2019nearly} which can simultaneously estimate all $k$-qubit reduced density matrix of an arbitrary quantum state in a sample-efficient manner for small $k$.
The simplest version of this idea is to measure uniformly random Pauli strings~\cite{evans2019scalable}, which leads to a sample complexity of $O(k 3^k \log n)$ for estimating all $k$-body Pauli observables to fixed precision.
Machine-learning based approaches are also proposed~\cite{carrasquilla2019reconstructing} in this direction.

Recently, a new paradigm for efficient and universal quantum property estimation has been proposed named quantum shadow estimation.
Shadow estimation was first put forward in Ref.~\cite{aaronson2018shadow}. 
Roughly speaking, this scheme can simultaneously estimate the expectation values with respect to $N$ observables of an unknown $d$-dimensional quantum state with order $\log d\log N$ number of samples, which is usually more efficient than either conducting full tomography or measuring the $N$ observables one by one. 
Later on, a more experiment-friendly shadow estimation scheme was proposed~\cite{huang2020predicting}, which is able to estimate many useful properties of a quantum system with a small number of samples (see also~\cite{paini2021estimating}). 
This protocol is also proven to be worst-case sample-optimal in the sense that any other protocol that is able to accurately estimate any collection of arbitrary observables must consume a number of samples at least comparable to this one.
Although promising for a broad spectrum of applications, the shadow estimation scheme in Ref.~\cite{huang2020predicting} (as well as the random Pauli scheme from~\cite{evans2019scalable}) assume perfect implementation of a group of unitary gates as well as ideal projective measurement on the computational basis. 
It remains unclear how experimental noise can affect the performance of this scheme.

In this work, we reexamine the shadow estimation scheme and regard it as a twirling and retrieval procedure of the measurement channel. 
In this way, we extend shadow estimation to the case when the unitaries and measurements are noisy. 
With similar techniques used in the study of randomized benchmarking~\cite{Emerson2005,knill2008randomized,chow2009randomized,magesan2011scalable,helsen2019new}, we propose a modified shadow estimation strategy which is noise-resilient. 
When the noise in the unitary operations and measurements is small, the robust shadow tomography scheme is able to faithfully estimate the required properties with a small additional cost, subject only to the assumption that one can prepare the initial ground state $\ket{0}^{\otimes n}$ with high fidelity. 
The proposed scheme is both robust and efficient, and hence highly practical for property estimation of a quantum system.

\section{Preliminaries}

We first introduce the Pauli-transfer-matrix (PTM) representation (or Liouville representation) to simplify the notation. 
Note that all the linear operators $\mc{L}(\mc{H}_d)$ on the underlying $n$-qubit Hilbert space $\mc{H}_d$ with $d=2^n$ can be vectorized using the $n$-qubit (normalized) Pauli operator basis $\{\sigma_a\coleq P_a/\sqrt{d}\}_a$, where $P_a$ are the usual Pauli matrices. 
For a linear operator $A\in\mc{L}(\mc{H}_d)$, we define a column vector $\lket{A} \in \mc{H}_{d^2}$ with the $a$-th entry to be $\lket{A}_a = \tr(P_a A) / \sqrt{d}$. 
The inner product on the vector space $\mc{H}_{d^2}$ is defined by the \emph{Hilbert-Schmidt inner product} as $\lbraket{A}{B}\coleq\tr(A^\dag B)$.
The normalized Pauli basis $\{\sigma_a\}_a$ is then an orthonormal basis in $\mc{H}_{d^2}$.
Superoperators on $\mc{H}_d$ are linear maps taking operators to operators $\mc{L}(\mc{H}_d) \to \mc{L}(\mc{H}_d)$.
In the vector space $\mc{H}_{d^2}$, a superoperator $\mc{E}$ can be represented by a matrix in the Pauli basis, with the entries given by $\mc{E}_{ab}=\lbraket{\sigma_a}{\mc{E}(\sigma_b)}=\lbra{\sigma_a}\cE\lket{\sigma_b}$. 
Choosing the Pauli basis for the superoperator is sometimes called the Pauli transfer matrix. 
With a slight abuse of notation, we sometimes denote a superoperator and its PTM using the same notation.
A detailed introduction to the PTM is given in Appendix~\ref{Liouville}.

In this work, we focus on the task of estimating the expectation values $\{\tr(O_i\rho)\}_i$ of a set of observables $\{O_i\}_i$ on an underlying unknown quantum state $\rho$,
\begin{equation}
	\tr(O_i\rho) = \lbraket{O_i}{\rho}, \quad 1\leq i\leq N.
\end{equation}
When the number of observables $N$ is large, a direct exhaustive measurement of the (generally incompatible) observables $\{O_i\}$ on $\rho$ is expensive.
Besides, in many cases we may want to perform tomographic experiments on $\rho$ before deciding which observables $\{O_i\}$ should be estimated. 
To realize this, a natural idea is to insert an extra prepare-and-measure superoperator between $\lbra{O_i}$ and $\lket{\rho}$,
\begin{equation} \label{eq:insertmap}
	\lbraket{O_i}{\rho} \to \sum_{x} \lbraket{O_i}{A_x}\lbraket{E_x}{\rho}.
\end{equation}
In an experiment, we first apply a POVM measurement $\{E_x\}_x$ at $\rho$.
Then, conditioned on the outcome $x$, we calculate $\lbraket{O_i}{A_x}$ via classical post-processing. 
If we repeat this procedure, then the sample average over these experiments gives an estimator for $\lbraket{O_i}{\rho}$. 
As long as the inserted superoperator $\sum_x \lketbra{A_x}{E_x}$ equals to $\mc I$, this estimator will be unbiased. 

To construct a realization of such a superoperator, we consider the dephasing channel in the computational basis ($Z$-basis) $\mc{M}_Z \coleq \sum_{z} \lketbra{z}{z}$, where $\lket{z}$ is the vectorization of the $Z$-basis eigenstate $\ket{z}\bra{z}$, with $z\in\{0,1\}^{\otimes n}$. 
Expanding $\mc{M}_Z$ in the Pauli operator basis $\{\lket{\sigma_0}, \lket{\sigma_x}, \lket{\sigma_y}, \lket{\sigma_z}\}^{\otimes n}$, we have
\begin{equation}
\begin{aligned}
	\mc{M}_Z &= \left( \lketbra{\sigma_0}{\sigma_0} + \lketbra{\sigma_z}{\sigma_z} \right)^{\otimes n} \\
	&= \left[\text{diag}(1,0,0,1)\right]^{\otimes n},
\end{aligned}
\end{equation}
where $\text{diag}(a_1,a_2,\cdots)$ is a diagonal matrix with the diagonal elements $a_1, a_2, \cdots$. 
If $\mc{M}_Z$ were invertible, we could insert the superoperator $\mc{M}_Z^{-1} \mc{M}_Z = \sum_z \lketbra{\mc{M}_Z^{-1}(z)}{z} = \mc{I}$. 
However, $\mc{M}_Z$ is not invertible due to the lack of $X,~Y$-basis information in a $Z$-basis measurement. 
To make $\mc{M}_Z$ invertible, we can introduce an extra unitary twirling~\cite{huang2020predicting},
\begin{equation}\label{eq:general_noiseless}
	\mc M = \mathop{\mathbb{E}}_{U \in\mathbb{G}} \cU^\dagger \mc{M}_Z \cU.
\end{equation}
Here, $\mathbb{G}$ is a subset of the unitaries $\{U\}$ in $U(d)$ to be specified later, and $\mc{U}$ is the PTM representation of $U$. 

When $\mbb{G}$ forms a group, the PTMs $\{\mc{U}\}$ forms a representation of $\mbb{G}$. 
A direct application of Schur's Lemma~\cite{fulton2013representation} (see Appendix \ref{App:rep}) allows us to calculate the explicit form of $\mc{M}$,
\begin{equation}\label{linearcombine}
	\mc{M}=\sum_{\lambda\in R_{\mbb{G}}}\frac{\tr[\cM_Z \Pi_{\lambda}]}{\tr[\Pi_{\lambda}]}\Pi_{\lambda},
\end{equation}
where $R_{\mbb{G}}$ is the set of irreducible sub-representations of the group $\mbb{G}$, and $\Pi_{\lambda}$ is the corresponding projector onto the invariant subspace. 
Since the projectors are complete and orthogonal to each other, $\cM$ is invertible if and only if all the coefficients are non-zero. 
Therefore the twirling group $\mbb{G}$ needs to satisfy 
\begin{equation} \label{eq:nonzeroCoff}
	\tr[\mc{M}_Z \Pi_{\lambda}]\neq 0,\ \forall\,\lambda\in R_{\mbb G}.
\end{equation}
Once Eq.~\eqref{eq:nonzeroCoff} is satisfied, we can construct a shadow estimation protocol based on the equation
\begin{equation} \label{eq:shadowM}
	\lbraket{O_i}{\rho} = \mathop \mbb{E}_{U\in \mbb{G}} \sum_{z\in\{0,1\}^{\otimes n}} \lbra{O_i} \mc{M}^{-1} \mc{U}^\dag \lket{z}\lbra{z} \mc{U} \lket{\rho}.
\end{equation}
To implement shadow estimation, one can repeat the following experiment: generate a single copy of $\rho$, act via a randomly sampled unitary $U$, and then perform a $Z$-basis measurement to return an output bit string $b$. 
Then $\lbra{O_i}\mc{M}^{-1} \mc{U}^\dag \lket{b}$ is calculated on a classical computer. 
Thanks to this decoupled processing of $\rho$ with respect to $O_i$, the estimation of different observables can be done in parallel with a relatively small increase in sample complexity.

The quantum shadow estimation procedure can be summarized as in Algorithm~\ref{Proto:QShE}.
\begin{figure}
\begin{algorithm}[H]
		\caption{Shadow Estimation (\textbf{Shadow})~\cite{huang2020predicting}}
		\label{Proto:QShE}
		\begin{algorithmic}[1]
			\Require
			Unknown $n$-qubit quantum state $\rho$, observables $\{O_i\}_{i=1}^M$, $\mbb{G}\subseteq U(2^n)$, quantum channel $\mc M$, and $N,K\in\mbb N_+$.
			\Ensure
			A set of estimates $\{\hat o_i\}$ of $\{\Tr(\rho O_i)\}$.
			\State $R\coleq NK$.
			\For{$r= 1~\text{\textbf{to}}~R$}
			\State Prepare $\rho$, uniformly sample $U\in\mbb G$ and apply to $\rho$.
			\State Measure in the computational basis, outcome $\ket{b}$.
			\State $\hat o_i^{(r)}\coleq\lbra{O_i}\mc M^{-1}\cU^\dagger\lket b,~\forall\,i$.
			\EndFor
			\State $\hat o_i\coleq \text{\textbf{MedianOfMeans}}(\{\hat o_i^{(r)}\}_{r=1}^R,N,K) ,~\forall\,i$.
			\State \Return{$\{\hat{o}_i\}$}.
		\end{algorithmic}
\end{algorithm}
\end{figure}
We refer to $\hat o_i^{(r)}$ as the single-round estimator. 
The subroutine \textbf{MedianOfMeans} divides the $R=NK$ single-round estimators into $K$ groups, calculates the mean value of each group, and takes the median of these mean values as the final estimator. 
As a formula:
\begin{equation} \label{eq:meanmedianestimator}
	\begin{aligned}
	\bar{o}_i^{(k)} &\coleq \cfrac1N\sum_{r=(k-1)N+1}^{kN}\hat o^{(r)}_i,\quad k=1,2,...,K. \\
	\hat{o}_i &\coleq \mathrm{median}\left\{ \bar{o}_i^{(1)},~\bar{o}_i^{(2)},~...,~\bar{o}_i^{(K)} \right\}.
	\end{aligned}
\end{equation}
For the standard shadow estimation algorithm~\cite{huang2020predicting}, the input quantum channel $\mc M$ is decided by Eq.~\eqref{eq:general_noiseless}.

\section{Robust Shadow Estimation} \label{sec:robust_est}

In practice, the unitary operations
and measurements used in the standard shadow estimation algorithm will be noisy. 
We want to mitigate the effect of this noise on the output estimate of the shadow. 
Our strategy to do this is simple: we first learn the noise as a simple stochastic model and then compensate for these errors using robust classical post-processing. 

In general, noise in quantum devices is not stochastic, and coherent errors must be addressed. 
However, thanks to the unitary twirling in shadow estimation, the stochastic nature of the noise is inherent to the protocol itself. 
For example, any noise map that is twirled over a Clifford group that contains the Pauli group as a subgroup will reduce the noise to a purely stochastic Pauli channel~\cite{Knill_2005}. 
The complete characterization of such noise channels can be efficiently and accurately performed~\cite{Harper2020efficient,flammia2019efficient,harper2020fast}. 
It is then straightforward to compensate for such errors by modifying the classical post-processing, although a lengthy analysis is required to show the efficacy of this strategy. 

In order to pursue a rigorous analysis of this strategy, we make the following two assumptions on the noise in the experimental device implementing the shadow estimation. 
\begin{assumption}\label{asm:noise}(Simplifying noise assumptions)
	\begin{itemize}
	    \item[$\bf A1$] The noise in the circuit is gate-independent, time-stationary, Markovian noise. 
	    \item[$\bf A2$] The experimental device can generate the computational basis state $\ket{\bvec 0}\equiv\ket{0}^{\otimes n}$ with sufficiently high fidelity.
	\end{itemize}
\end{assumption}

Our first assumption is used throughout to ensure that there exists a completely positive trace-preserving (CPTP) map such that the noisy gate $\tilde{\mc{U}}$ can be decomposed into $\Lambda \mc{U}$, where $\mc{U}$ is the ideal gate while $\Lambda$ is the noise channel. 
The noise map $\Lambda$ is independent of the unitary $\mc{U}$ and the time $t$. 
It also implies that the noise map occurring in the measurement is fixed independent of time and hence can be absorbed into $\Lambda$.
We remark that assumption $\bf A1$ is widely used in the analysis of randomized benchmarking protocols. 
The gate-independent part of the assumption is especially appropriate when the experimental unitaries are single-qubit gates, but it has been shown that the effect of weak gate dependence (a form of non-Markovianity) generally leads to weak perturbations~\cite{wallman2018randomized,merkel2018randomized}. We also provide numerical evidences in Sec.~\ref{sec:gatedependent} showing that our scheme is still quite robust against realistic gate-dependent noise models in experiment.

For our second assumption \textbf{A2}, from Sec.~\ref{sec:robust_est} to Sec.~\ref{sec:local} we initially make the stronger assumption that the experimental device can prepare the $\ket{\bvec 0}$ state \textit{exactly}. 
In Sec.~\ref{sec:prepare_noise} we relax this to show that when $\ket{\bvec 0}$ is not precisely prepared, but is prepared with sufficiently high fidelity, our protocol still gives a good estimation. 
Fortunately, the computational basis state $\ket{\bvec 0}$ is relatively easy to generate faithfully in many experimental platforms. 

To see how unitary twirling helps to reduce the number of noise parameters, we calculate the noisy version of the random measurement channel $\widetilde{\cM}$, 
\begin{equation} \label{eq:noisytwirl}
\begin{aligned}
	\widetilde{\cM} &= \mathop{\mathbb{E}}_{U \in \mbb{G}} \mc{U}^\dag \mc{M}_Z \Lambda \mc{U} \\
	&= \sum_{\lambda\in R_{\mbb{G}}}\frac{\tr[\cM_Z \Lambda \Pi_{\lambda}]}{\tr[\Pi_{\lambda}]}\Pi_{\lambda} = \sum_{\lambda} f_\lambda \Pi_\lambda,
\end{aligned}
\end{equation}
where the $\{f_\lambda\}$ are expansion coefficients of the twirled channel. 
Note that the channel $\Lambda$ describes both the noise in the gate $\cU$ and in the measurement $\mc M_Z$, which is always possible under our assumption $\bf{A1}$.
The number of $\{f_\lambda\}$ is related to the number of irreducible representations in the PTM representation of the twirling group~$\mbb{G}$. 
Later we will show that the coefficients $\{f_\lambda\}$ can be estimated in parallel, similar to the normal shadow estimation procedure (referred to as the \textit{calibration procedure}). 

\begin{figure}[!htbp]
	\centering
	\includegraphics[width = 0.9\columnwidth]{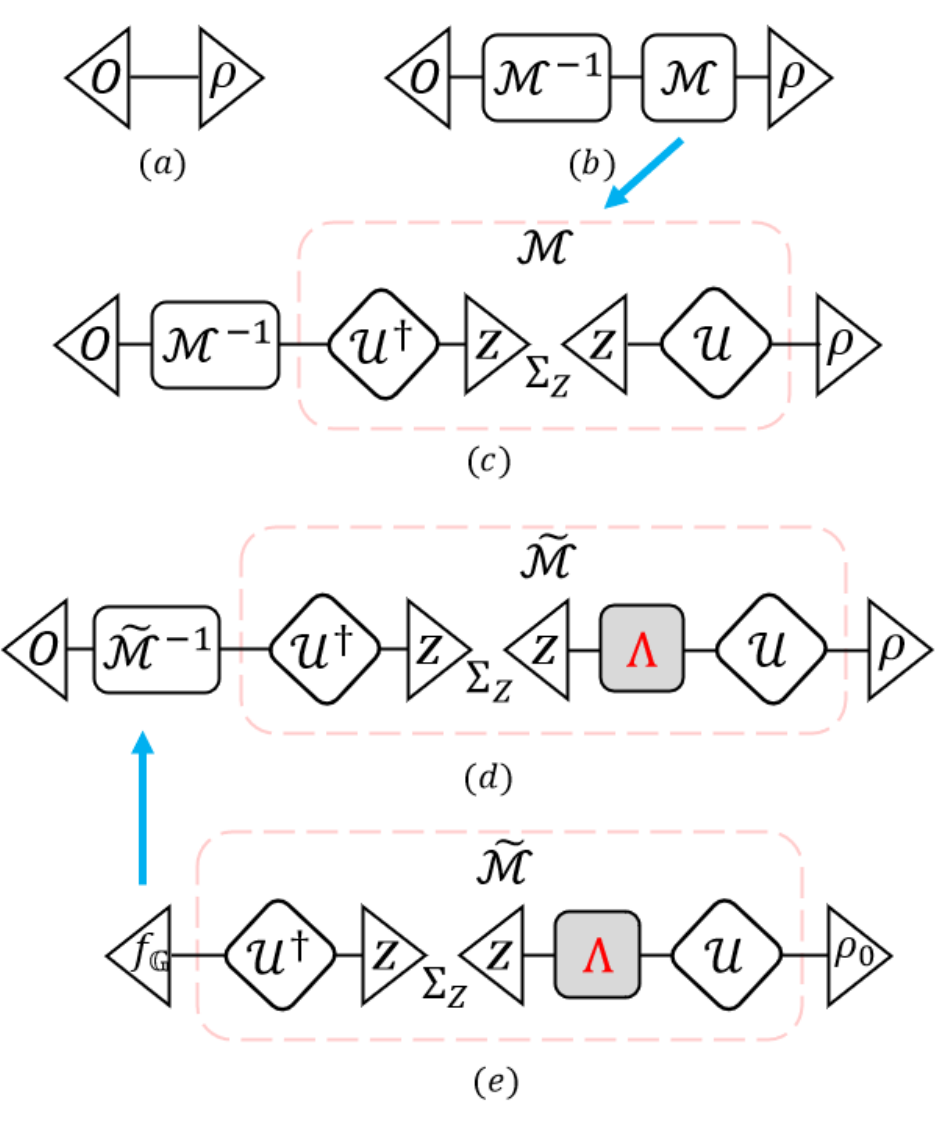}
	\caption{Diagram of the shadow estimation protocol.
	(a) We want to estimate the expectation value $\tr(O\rho)=\lbraket{O}{\rho}$ for a set of observables $\{O_i\}$ and an unknown state $\rho$. 
	(b) To do this, we insert a channel $\mc{M}$ and its corresponding inverse map $\mc{M}^{-1}$ in the middle, which will not change the expectation value. 
	(c) The channel $\mc{M}$ can be realized as a random unitary twirling $\mbb{E}_{U}\,\mc{U}^\dag \cdot \mc{U}$ acting on the $Z$-basis dephasing map $\mc{M}_Z = \sum_{z} \lketbra{z}{z}$. 
	(d) In practice, the implemented unitary $\mc{U}$ and the measurement $\lbra{z}$ are noisy, causing an extra uncharacterized noise channel $\Lambda$.
	(d) In practice, the unitary $\mc{U}$ and the measurement $\lbra{z}$ suffers from a noise channel $\Lambda$, causing an uncharacterized channel $\widetilde\cM$ that needs to be inverted. 
	(e) The calibration procedure of \textbf{RShadow}. By experimenting on some well-characterized state $\rho_0$, we can estimated the channel $\widetilde{\mc{M}}$ and its inverse, hence mitigate the noise in the shadow estimation procedure. Here $f_{\mbb G}$ is the \textbf{NoiseEst}${}_{\mbb G}$ subroutine described in Algorithm~\ref{Proto:robust}.}
	\label{fig:robustshadow}
\end{figure}

Based on the observations above, we propose our robust quantum shadow estimation (\textbf{\textrm{RShadow}}) protocol to faithfully estimate $\{\tr(O_i\rho)\}_i$ even with noise. 
The algorithm is depicted by Fig.~\ref{fig:robustshadow} and it works as follows. 
We first estimate the noise channel $\widetilde{\cM}$ of Eq.~\eqref{eq:noisytwirl} with the \emph{calibration procedure}, and then use the estimator $\widetilde{\cM}$ as the input parameter $\mc M$ of Algorithm~\ref{Proto:QShE} to predict any properties of interest (referred to as the \emph{estimation procedure}). 
The procedure is shown in Algorithm~\ref{Proto:robust}, where the subroutine 
$\text{\textbf{NoiseEst}}$ is decided by $\mbb G$ and is given later. 
\begin{figure}
\begin{algorithm}[H]
	\caption{Robust Shadow Estimation (\textbf{\textrm{RShadow}})}
	\label{Proto:robust}
	\begin{algorithmic}[1]
		\Require
		Unknown $n$-qubit quantum state $\rho$, observables $\{O_i\}_{i=1}^M$, $\mbb{G}\subseteq U(2^n)$ and $N_1,N_2,K_1,K_2\in\mbb N_+$.
		\Ensure
		A set of estimations $\{\hat o_i\}$ of $\{\Tr(\rho O_i)\}$.
		\State $R\coleq N_1K_1$.\Comment{\textit{Calibration}}
		\For{$r= 1~\text{\textbf{to}}~R$}
		\State Prepare $\ket{\bvec 0}$, sample (noisy) $U\in\mbb G$ and apply to $\ket{\bvec 0}$.
		\State Measure in the computational basis, return $\ket{b}$.
		\State $\hat f_\lambda^{(r)}\coleq\text{\textbf{NoiseEst}}_{\mbb G}(\lambda,U,b),~\forall~\lambda\in R_{\mbb G}$.
		\EndFor
		\State $\hat f_\lambda \coleq \text{\textbf{MedianOfMeans}}(\{\hat f_\lambda^{(r)}\}_{r=1}^R,N_1,K_1) ,~\forall~\lambda\in R_{\mbb G}$.\
		\State $\widehat{\cM}\coleq \sum_{\lambda\in R_{\mbb G}}\hat f_\lambda\Pi_\lambda$.
		
		\State $\{\hat{o}_i\} = \textbf{Shadow}(\rho,\{O_i\},\mbb{G},\widehat{\cM},N_2,K_2)$.\Comment{\textit{Estimation}}
		\State \Return{$\{\hat{o}_i\}$}.
	\end{algorithmic}
\end{algorithm}
\end{figure}

In the following discussion, we will focus on two specific groups $\mbb G$: the $n$-qubit Clifford group ${\sf Cl}(2^n)$ and the $n$-fold tensor product of the single-qubit Clifford group ${\sf Cl}_2^{\otimes n}$. 
We will give a specific construction of the \textbf{NoiseEst} subroutine and show the correctness and efficiency of our \textbf{\textbf{\textrm{RShadow}}} algorithm with these two groups.

\section{Robust shadow estimation using global Clifford group}
	
	

We first present a robust shadow estimation protocol using the $n$-qubit global Clifford group, ${\sf Cl}(2^n)$. 
The $n$-qubit Clifford group has many useful properties such as being a unitary 3-design~\cite{webb2015clifford,zhu2017multiqubit,kueng2015qubit}, which is widely used in many tasks of quantum information and quantum computation. 
It is a standard result that the $n$-qubit Clifford group has two irreducible representations in the Liouville representation whose projectors are given by $\lket {\sigma_{\bvec 0}}\lbra {\sigma_{\bvec 0}}$ and $I - \lket{\sigma_{\bvec 0}}\lbra{\sigma_{\bvec 0}}$. 
Assuming the $\widetilde\cM$ channel defined in Eq.~\eqref{eq:noisytwirl} is trace preserving, it can be written as 
\begin{equation}
	\widetilde\cM = \mathop{\mathbb{E}}_{U \sim {\sf Cl}(2^n)} \cU^\dagger \cM_Z\Lambda \cU = \lketbra{\sigma_{\bvec 0}}{\sigma_{\bvec 0}} + f(I-\lketbra{\sigma_{\bvec 0}}{\sigma_{\bvec 0}})
\end{equation}
for some $f\in\mbb R$, i.e.\ as a depolarizing channel. 
It is easy to obtain $f = 1/(2^n+1)$ for the noiseless case using Eq.~\eqref{eq:noisytwirl}.
The noise characterization subroutine with ${\sf Cl}(2^n)$ is defined as follows,
\begin{equation}
	\begin{aligned}
    \text{\textbf{NoiseEst}}_{{\sf Cl}(2^n)}(U,b)\coleq\frac{2^n\lbra b \cU \lket{\bvec 0}-1}{2^n-1},
	\end{aligned}
\end{equation}
where $\lket b$ is the Liouville representation of the computational basis state $\ket b\bra b$ and similar for $\lket{\bvec 0}$.

	Next, define the Z-basis average fidelity of a noise channel $\Lambda$ as $F_{Z}(\Lambda) = \frac{1}{2^n}\sum_{b\in\{0,1\}^n}\lbra b\Lambda\lket b$.
	The following theorem demonstrates the correctness and sample efficiency of our protocol. We remark that the validity of this theorem relies on Assumptions~\ref{asm:noise}.
	
	\begin{theorem}[Informal]\label{th:info1}
		For {\textbf{\textrm{RShadow}}} with $\mathbb{G} = {\sf Cl}(2^n)$, if the number of samples for the calibration procedure satisfies
		\begin{equation}
		R=\tilde{\mc O}(\varepsilon^{-2}F_{Z}^{-2}),
		\end{equation}
		where $F_{Z}\equiv F_{Z}(\Lambda)$ 
		and we assume $F_{Z} \gg 2^{-n}$, 
		then the subsequent estimation procedure with high probability satisfies
		\begin{align}
		\label{eq:unbiased}
		&\left| \mbb E(\hat o^{(r)}) - \Tr(O\rho)  \right| \le \varepsilon \|O\|_\infty,
		\end{align}
		for any observable $O$ and quantum state $\rho$, where $\hat o^{(r)}$ is the single-round estimator defined as in Algorithm~\ref{Proto:QShE}.
	\end{theorem}
	
	Here and throughout the paper, we use $\tilde{\mc O}$ to {represent the Big-O notation with poly-logarithmic factors suppressed}. 
	The more rigorous version of Theorem~\ref{th:info1} is Theorem~\ref{th:gl_stat} in Appendix~\ref{sec:app_gl}.
	We see that our protocol indeed eliminates the \emph{systematic error} of shadow estimation in a sample-efficient manner, since without the calibration step the empirical expectation value would converge to a value that conflated the noise map $\Lambda$ into the estimate, whereas $\Lambda$ does not appear in Eq.~\eqref{eq:unbiased}. 
	More specifically, if the $Z$-basis average fidelity of the noise channel $\Lambda$ is lower bounded by some constant (\textit{e.g.} constant-strength depolarizing noise), then the sample complexity of our calibration stage is approximately independent of the system size $n$. 
	
	A more realistic noise model to consider is that of local noise with fixed strength, where $\Lambda \coleq \bigotimes_{i=1}^n\Lambda_i$ and each single-qubit noise channel $\Lambda_i$ satisfies $F_{Z}(\Lambda_i)\ge 1-\xi$. 
	In that case, we have $F_{Z}(\Lambda)^{-2}\approx \exp(2n\xi)$ for small $\xi$, so we can efficiently deal with a system size $n$ that is comparable to $\xi^{-1}$. 
	
	Next, we consider the sample complexity of the estimation procedure. 
	Following a similar methodology of bounding the sample complexity in the noise-free standard shadow estimation scheme~\cite{huang2020predicting}, we bound the sample complexity of our \textbf{\textrm{RShadow}} estimation procedure as follows.
	
	\begin{theorem}[Informal]\label{th:gl_all}
	For \textbf{\textrm{RShadow}} with ${\sf Cl}(2^n)$, if the number of calibration samples $R_C$ and the number of estimation samples $R_E$ satisfies
	\begin{equation}
	\begin{aligned}
	R_C &= \tilde{\mc O}(\varepsilon_1^{-2}F_{Z}^{-2}),\\R_E &=\tilde{\mc O}(\varepsilon_2^{-2} F_{Z}^{-2}\log M),
	\end{aligned}
	\end{equation}
	respectively, then the protocol can estimate $M$ arbitrary linear functions $\Tr(O_1\rho),...\Tr(O_M\rho)$ such that $\max_i\Tr(O_i^2)\le 1$, up to accuracy $\varepsilon_1+\varepsilon_2$ with high success probability. 
	\end{theorem}

	The rigorous version of Theorem~\ref{th:gl_all} is Theorem~\ref{th:gl_allr} in Appendix~\ref{sec:app_gl}. 
	Compared with results in Ref.~\cite{huang2020predicting}, one can see that the \textbf{\textrm{RShadow}} scheme has nearly the same sample complexity order as the noise-free standard shadow estimation methods in a low-noise regime.
	
	
	Finally, we comment on the computational complexity of \textbf{\textrm{RShadow}}. 
	The computational complexity of our calibration procedure is favorable since the single-round fidelity estimator can be calculated efficiently with the Gottesman-Knill theorem~\cite{gottesmanPhDthesis,aaronson2004improved}. 
	However, a efficient computation using the Gottesman-Knill theorem for the estimation procedure would require the observable $O$ to have additional structure, such as being a stabilizer state or being a Pauli operator. 
	The standard shadow estimation scheme of Ref.~\cite{huang2020predicting} or the fast Pauli expectation estimation method of Ref.~\cite{evans2019scalable} also have such a requirement.

\section{Robust shadow estimation using local Clifford group} \label{sec:local}
	
Despite the useful properties the global Clifford group possesses, it is often challenging to implement the full $n$-qubit Clifford group under current experimental conditions. 
The local Clifford group ${\sf Cl}_2^{\otimes n}$, which is the $n$-fold tensor product of the single-qubit Clifford group, is an experimentally friendly alternative. 
We now present an robust shadow estimation protocol based on the local Clifford group which can efficiently calibrate and mitigate the error in estimating any \emph{local property}.
	
It is known that the $n$-qubit local Clifford group has $2^n$ irreducible representations~\cite{gambetta2012benchmarking}. 
Being twirled by the local Clifford group, the channel $\widetilde{\cM}$ becomes a Pauli channel that is symmetric among the $x, y, z$ indices, and the Pauli-Liouville representation is
\begin{equation}\label{eq:Mtilde}
	\widetilde{\cM} = \mathop{\mathbb{E}}_{U \sim {\sf Cl}_2^{\otimes n}} \cU^\dagger \cM_Z\Lambda \cU = \sum_{z\in\{0,1\}^n}f_z \Pi_z, 
\end{equation}
where $\Pi_z = \bigotimes_{i=1}^n \Pi_{z_i}$, 
$$
    \Pi_{z_i} = 
	\begin{cases}
	\ \lketbra{\sigma_0}{\sigma_0}, & z_i=0, \\ 
	\ I - \lketbra{\sigma_0}{\sigma_0}, & z_i=1,
	\end{cases}
$$
for $f_z\in\mbb R$ which is called the Pauli fidelity. 
Here, for any string $m\in\{0,1\}^n$, we define $\lket m$ to be the Liouville representation of the computational basis state $\ket m\bra m$, and define $P_m\coleq\bigotimes_{i=1}^{n}P_Z^{m_i}$ and $\sigma_m$ to be the corresponding normalized Pauli operators. 
In the noiseless case, one can obtain $f_z = 3^{-|z|}$ using Eq.~\eqref{eq:noisytwirl}, where $|z|$ is the number of $1$s in $z$. 
	
	
The noise characterization subroutine with ${\sf Cl}_2^{\otimes n}$ is defined as follows
\begin{equation}
	\text{\textbf{NoiseEst}}_{{\sf Cl}_2^{\otimes n}}(z,U,b)\coleq\lbra{b}\mc{U}\lket{P_z},~\forall z\in\{0,1\}^n.
\end{equation}

In the standard shadow estimation using ${\sf Cl}_2^{\otimes n}$~\cite{huang2020predicting} (and in the earlier work~\cite{evans2019scalable}), one can only efficiently estimate observables with small Pauli weight. 
An $n$-qubit observable $O$ is called $k$-local if it can be written as $O=\tilde O_S\otimes I_{[n]\backslash S}$ for some $k$-element index set $S\subset [n]$ and a $k$-qubit observable $\tilde O$. 
Similarly, our \textbf{\textrm{RShadow}} protocol with ${\sf Cl}_2^{\otimes n}$ is also designed for predicting $k$-local observables. 
The correctness and efficiency is given by the following theorem.

\begin{theorem}[Informal]\label{th:info2}
	For \textbf{\textrm{RShadow}} with ${\sf Cl}_2^{\otimes n}$, if the number of samples for the calibration procedure satisfies
	\begin{equation}
	    R=\tilde{\mc O}\left({3^k} {\varepsilon^{-2}F_{Z}^{-2}}\right),
	\end{equation}
		then the subsequent estimation procedure with high probability satisfies
		\begin{equation}
		\begin{aligned}
		\left| \mbb E(\hat o^{(r)}) - \Tr(O\rho)  \right| \le \varepsilon 2^k\|O\|_\infty,
		\end{aligned}
		\end{equation}
		for any $k$-local observable $O$ and quantum state $\rho$, where $\hat o^{(r)}$ is the single-round estimator defined as in Algorithm~\ref{Proto:QShE}.
%
%
	\end{theorem}
	
	The rigorous version of Theorem~\ref{th:info2} is Theorem~\ref{th:lc_stat} in Appendix~\ref{sec:app_lc}. 
	Indeed, this protocol can calibrate the shadow estimation process for all $k$-local observables using a small number of samples that only depends on $k$ (but basically not on the system size $n$). 
	Note that, Theorem~\ref{th:info2} holds for any gate-independent noise model, even for \emph{global unitary noise}.
	
	Now we investigate the sample complexity of the estimation procedure. 
	We are currently unable to bound the sample complexity against the most general noise channel, but we do have a bound for a local noise model, as shown in the following theorem:
	
	\begin{theorem}[Informal]\label{th:lc_all}
		For \textbf{\textrm{RShadow}} with ${\sf Cl}_2^{\otimes n}$, 
		suppose the noise is local, \textit{i.e.} $\Lambda\coleq\bigotimes_{i=1}^n\Lambda_i$, and satisfies $F_{Z}(\Lambda_i)\ge 1-\xi$ for all $i$ and some $\xi\ll\frac12$.
		If the number of calibration samples $R_C$ and the number of estimation samples $R_E$ satisfy
		\begin{equation}
		\begin{aligned}
		R_C &= \tilde{\mc O}(12^k e^{4k\xi}\varepsilon_1^{-2}),\\ R_E &=\tilde{\mc O}(4^ke^{4k\xi}\varepsilon_2^{-2}\log M),
		\end{aligned}
		\end{equation}
		respectively, 
		then the protocol can estimate $M$ arbitrary linear functions $\Tr(O_1\rho),...\Tr(O_M\rho)$ such that every $O_i$ is $k$-local and $\|O_i\|_\infty \le 1$, up to accuracy $\varepsilon_1+\varepsilon_2$ with high success probability.
	\end{theorem}
	
	The rigorous version of Theorem~\ref{th:lc_all} is Theorem~\ref{th:lc_allr} in Appendix~\ref{sec:app_lc}. 
	Again, we see \textbf{\textrm{RShadow}} using ${\sf Cl}_2^{\otimes n}$ has a sample complexity similar to the noiseless standard shadow estimation protocol when the noise is local and not too strong. 
	We also remark that, although we do not have a sample complexity bound against a more general noise model, our numerical results show that \textbf{\textrm{RShadow}} can still perform well in that case (see Appendix~\ref{sec:app_numer}). 
	Furthermore, in realistic experiments, one can monitor the standard deviation of estimators in real time, which means they can still suppress statistical fluctuations to an acceptable level even without a theoretical sample complexity bound.
	
	%
		
	Regarding the computational complexity, it is obviously impractical to calibrate all $2^n$ parameters $f_z$. 
	However, since we only care about $k$-local observables, only $\hat f_z^{(r)} $ such that $|z|\le k$ needs to be computed, the number of which is no greater than $n^k$. 
	Further note that $\hat f_z^{(r)}$ can be decomposed as $\prod_{i=1}^n \lbra{b_i}U_i\lket{P_Z^{z_i}}$, so all these $\hat f_z^{(r)}$ can be computed within $O(n^k)$ time using dynamic programming. 
	If there is extra structure of the observables to be predicted (\textit{e.g.} spatially local), the number of necessary $\hat f_z^{(r)}$ can be further reduced.
	In practice, one may store the raw data of the calibration procedure and see what observables are to be predicted, before deciding which set of $f_z$ need to be calculated.
	An example is given below in our numerical experiments.
	The computational complexity for the estimation procedure is therefore low when the observables are $k$-local for reasonably small $k$.

\section{Robustness against state preparation noise}\label{sec:prepare_noise}

In the last two sections, we prove the performance of the \textbf{\textrm{RShadow}} protocol based on the assumption of perfect $\ket{\bvec 0}$ preparation. 
Although $\ket{\bvec 0}$ is relatively easy to prepare on most current quantum computing platforms, state preparation (SP) noise is still inevitable. 
In this section, we show that the \textbf{\textrm{RShadow}} protocol is also robust against small SP noise in the following sense: when $\ket{\bvec 0}$ can be prepared with high fidelity during the calibration procedure, the estimators for the estimation procedure will not be too biased, and the sample complexity will not increase drastically.  
	
Formally, in a realistic calibration procedure, one prepares some $\rho_{\bvec 0}$ instead of $\ket{\bvec 0}\bra{\bvec 0}$ for each round.
We assume $\rho_{\bvec 0}$ is time-independent, which is reasonable if the experimental conditions do not change much during this process. 
We have the following theorems:
	
	
	\begin{theorem}\label{th:sp_g}
	For \textbf{\textrm{RShadow}} using $\sf{Cl}(2^n)$, if the state-preparation fidelity satisfies
	\begin{equation}
	    F(\ket{\bm 0}\bra{\bm 0},\rho_{\bm 0})\ge 1-\varepsilon_{\mr{SP}},
	\end{equation}
	then with 
	the same number of calibration samples as in Theorem~\ref{th:info1}, the subsequent estimation procedure with high probability satisfies
	\begin{equation}
    \left| \mbb E(\hat o^{(r)}) - \Tr(O\rho)  \right| \le (\varepsilon+2\varepsilon_\mr{SP}) \|O\|_\infty.
	\end{equation}
	up to the first order of $\varepsilon$ and $\varepsilon_\mr{SP}$.
	\end{theorem}
	
	\begin{theorem}\label{th:sp_l}
	For \textbf{\textrm{RShadow}} using $\sf{Cl}_2^{\otimes n}$, if the state is prepared as a {product state} $\rho_{\bm 0}=\bigotimes_{i=1}^n\rho_{0,i}$ and the single-qubit state-preparation fidelity satisfies
	\begin{equation}
	    F(\ket{0}\bra{0},\rho_{0,i})\ge 1-\xi_\mr{SP},\quad\forall i\in[n],
	\end{equation}
	then with 
	the same number of calibration samples as in Theorem~\ref{th:info2}, the subsequent estimation procedure with high probability satisfies
	\begin{equation}
    \left| \mbb E(\hat o^{(r)}) - \Tr(O\rho)  \right| \le (\varepsilon+2 k\xi_{SP})2^k \|O\|_\infty.
	\end{equation}
	up to the first order of $\varepsilon$ and $k\xi_{SP}$. $k$ is the locality of observable $O$.
	\end{theorem}
	
	The proof is given in Appendix~\ref{sec:app_spn}.
	The above two theorems show that the effect of state-preparation noise can indeed be bounded for \textbf{\textrm{RShadow}}. They also enable an experimentalist to decide a practical sample number according to how well his device can prepare $\ketbra{\bvec 0}{\bvec 0}$.
	
	
\section{Numerical results}\label{sec:numer}
	
	
	
	
	
	Here, we design several numerical experiments to demonstrate the practicality of the robust shadow estimation (\textbf{\textrm{RShadow}}) protocol. 
	We first benchmark the robustness of the \textbf{\textrm{RShadow}} protocol under various types of noise models in the task of estimating the fidelity of the GHZ state. 
	After that, we show the application of \textbf{\textrm{RShadow}} in estimating the $2$-point correlation as well as the energy of the ground state of the anti-ferromagnetic transverse-field Ising model (TFIM). 
	These tasks frequently appear in the field of quantum computational chemistry~\cite{mcardle2020quantum}. 
	In all the numerical experiments, we assume that the states to be tested are perfectly prepared while the shadow estimation circuits are noisy. 
	We compare the performance of \textbf{\textrm{RShadow}} protocol with the standard quantum shadow estimation scheme (standard \textbf{\textrm{Shadow}})~\cite{huang2020predicting} in all the tasks. 
	Our numerical simulation makes use of Qiskit~\cite{qiskit}, an open-source python-based quantum information toolkit. 
	
	For all the plots in this section, the error bars represent the standard deviation of the estimation procedure (which means we ran the calibration procedure of \textbf{RShadow} only once for each data point), and are calculated via the empirical bootstrapping method~\cite{efron1992bootstrap}, where we randomly samples the same size of data points \emph{with replacement} from the original data and calculate the estimator as a bootstrap sample. 
	Repeat this for $B=200$ times, and take the standard deviation among these bootstrap samples as an approximation to the standard deviation of our \textbf{RShadow} estimator. 
	
	In the first experiment, we numerically prepare a $10$-qubit GHZ state, and use the shadow estimation protocol to estimate its fidelity with the ideal GHZ state. 
	Each protocol use $R=10^5$ ($N=10^4,~K=10$) samples for the estimation stage, while our \textbf{\textrm{RShadow}} uses an extra $R=10^5$ ($N=10^4,~K=10$) samples for its calibration stage. 
	We simulate the following three noise model: depolarizing, amplitude damping, and measurement bit-flip, each with several different levels of strength. 
	The random circuits are set to be global Clifford gates. 
	Fig.~\ref{fig:global_noise} shows the results. 
	One can see that, for all these three noise models, when the noise level increases, the standard shadow estimation deviates from the true value, while the robust shadow estimation remains faithful.
	\begin{figure}[!htb]
		\centering
		\includegraphics[width = \columnwidth]{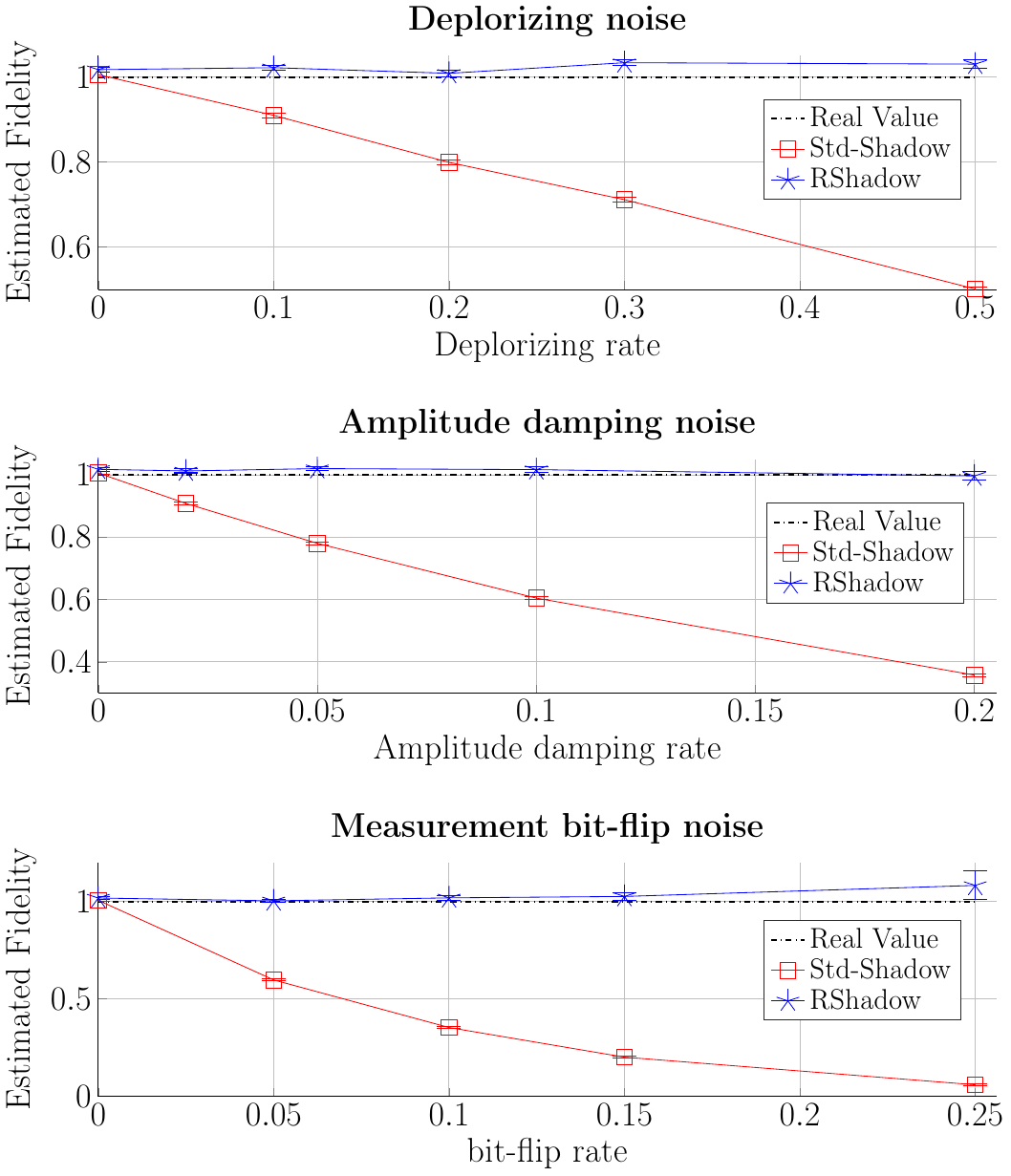}
		\caption{Comparison of the GHZ-state fidelity estimation using standard \textbf{\textrm{Shadow}} and the \textbf{\textrm{RShadow}} with respect to different noise models and noise levels. 
		The black dashed line represents the true value. 
		The red crosses and the blue stars represent the estimated values by the standard \textbf{\textrm{Shadow}} and \textbf{\textrm{RShadow}}, respectively.}
		\label{fig:global_noise}
	\end{figure}
	
	Note that there exist some numerical results from our robust procedure exceeding the ground truth in the above figure.
	That is due to the nature of the shadow protocol to eliminate the effect of noisy fidelity parameters $\{\hat{f}_\lambda\}$.
	To eliminate this fidelity parameters and extract the estimation of a desired observable, the protocol will use a ratio estimator, which tends to have results with systematically biased errors.
	Moreover, the statistical fluctuation will affect the estimation of our procedure.
	Fortunately, Theorem~\ref{th:gl_all} and \ref{th:lc_all} allow us to bound the size of fluctuation errors along with the systematic biases.
	From practical consideration, the estimation results of the observables $\{\hat{o}^{(r)}\}$ are allowed to be truncated given some physical ranges from prior knowledge, and this can help to improve the accuracy and circumvent some non-physical estimation. A caution is that we cannot

	On the same task of estimating the GHZ-state fidelity, we further test the performance of our \textbf{\textrm{RShadow}} method when the size of system increasing from $4$ qubits to $12$ qubits. 
	During the measurement procedure, we set a noise model where all the qubits undergo a local $X$-rotation $U_{X}(\theta) = e^{-i\theta X}$. 
	We remark that such kind of coherent noise can not be modeled as a classical error occurred in the measurement results. 
	We fix the number of trials to be $R=10^5$ ($N=2500, K=40$) for both the calibration and estimation stages. 
	Meanwhile, we choose the rotation angle to be $\theta=\frac{\pi}{25}$, $\frac{2\pi}{25}$, and $\frac{3\pi}{25}$. 
	In Fig.~\ref{fig:global_size}, we compare the fidelity estimation result of standard \textbf{\textrm{Shadow}} and \textbf{\textrm{RShadow}}. 
	When local noises occur, the performance of standard \textbf{\textrm{Shadow}} decreases when the system size increases. 
	In contrast, the estimation of \textbf{\textrm{RShadow}} is still accurate. 
	This highlights the necessity of noise suppression especially when the system size gets larger.


\begin{figure}[!htb]
\centering

\includegraphics[width = \columnwidth]{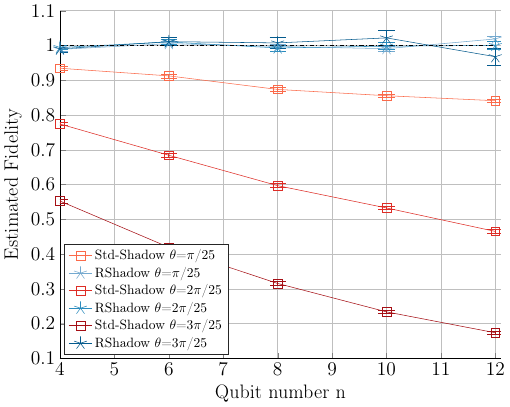}
\caption{Comparison of the GHZ fidelity estimation using standard \textbf{\textrm{Shadow}} and the \textbf{\textrm{RShadow}} with respect to different qubit numbers $n$. 
Here, we assume that all the qubits will experience a local $X$-rotation error $U_{X}(\theta) = e^{-i\theta X}$ with $\theta=\pi/25$, $2\pi/25$, and $3\pi/25$. 
In the experiment, we set the number of trials $R=10^5$ ($N=2500,K=40$) for both calibration and estimation stages.}
\label{fig:global_size}
\end{figure}

	
	
	\medskip
	The next experiment is designed for shadow estimation with local Clifford group. 
	We estimate the $2$-point $ZZ$-correlation functions and energy expectation of the ground state of an anti-ferromagnetic transverse field Ising model (TFIM) in one dimension with open boundary, whose Hamiltonian is $H = J\sum_iZ_iZ_{i+1} + h\sum_i X_i$ and we focus on the case $J=h=1$. 
	The ground state is approximated using density matrix re-normalization group method, represented by a matrix-product state (MPS). 
	Codes from~\cite{carrasquilla2019reconstructing} are modified here to sample random Pauli measurements on the MPS.
	We compare the performance of \textbf{\textrm{RShadow}} and the standard shadow estimation~\cite{huang2020predicting} scheme in the presence of measurement bit-flip noise, which means each qubit measurement outcome has an independent probability $p$ to be flipped. 
	Our \textbf{\textrm{RShadow}} uses $R=500000$ ($N=20000,~K=25$) calibration samples and $R=500000$ ($N=10000,~K=50$) estimation samples, while standard shadow estimation uses $R=500000$ ($N=10000,~K=50$) samples. 

	We first generate a 50-spin TFIM ground state, and estimate the $ZZ$-correlation functions between the leftmost spin and all other spins $\langle Z_0Z_i\rangle$, where the bit flip probability is set to be $5\%$. 
	Fig.~\ref{fig:TFIM_Corr} shows the estimation values and absolute errors of both \textbf{\textrm{RShadow}} and standard \textbf{\textrm{Shadow}}. 
	It can been seen that \textbf{\textrm{RShadow}} in general gives a much more precise estimation than standard \textbf{\textrm{Shadow}}.
	\begin{figure}[!htbp]
		\centering
	   \includegraphics[width = \columnwidth]{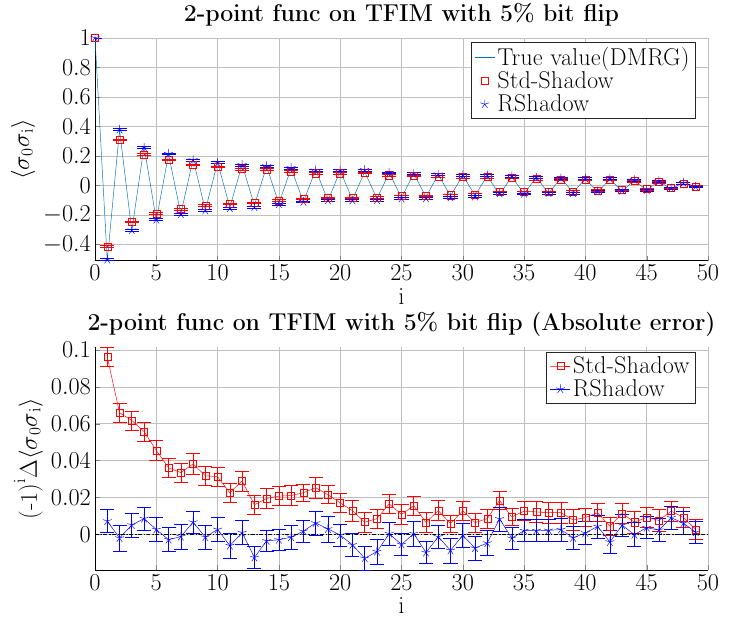}
	    \caption{2-point correlation function estimation on 50-spin 1-D TFIM ground state.}
		\label{fig:TFIM_Corr}
	\end{figure}

	We then estimate the energy expectation. 
	In Fig.~\ref{fig:TFIM_E} we plot the energy estimation results on a 50-spin TFIM ground state under three different noise models (similar as above numerical experiments of GHZ fidelity estimation). 
	One can see that the estimation error of standard \textbf{\textrm{Shadow}} increases when the noise level increases, while \textbf{\textrm{RShadow}} remains giving precise results. 
	\begin{figure}[!htbp]
		\centering
		\includegraphics[width = \columnwidth]{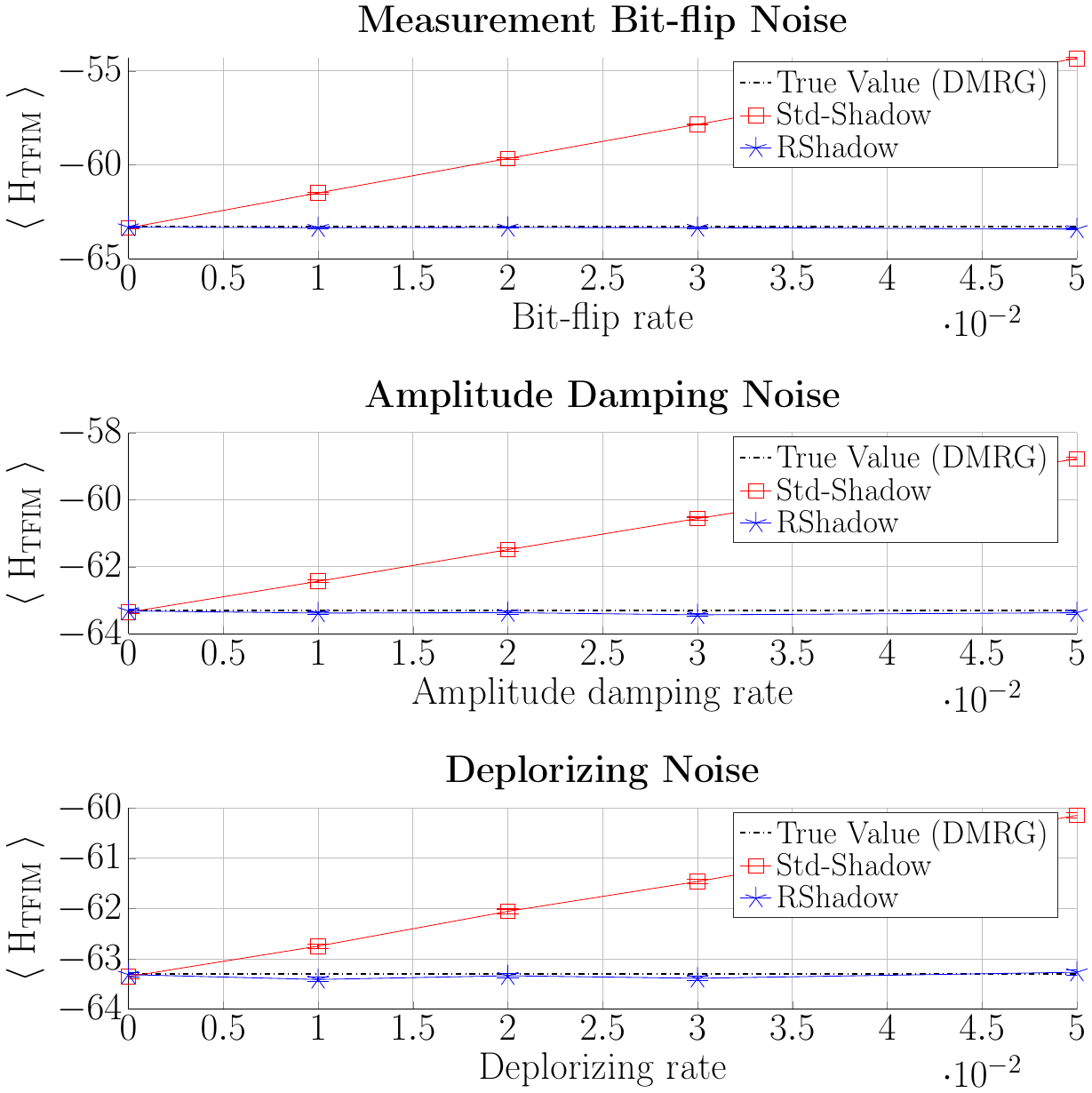}
		\caption{Energy expectation estimation on 50-spin 1-D TFIM ground state.}
		\label{fig:TFIM_E}
	\end{figure}
	Then we fix the noise model to be $5\%$ measurement bit flip and conduct estimation on different sizes of systems. 
	In Fig.~\ref{fig:TFIM_E_SD} we plot the absolute estimation error. 
	This error increases when the system size grows for the standard \textbf{\textrm{Shadow}}, but it remains close to zero for \textbf{\textrm{RShadow}} scheme. 
	This provides a strong reason why the \textbf{\textrm{RShadow}} scheme should be applied as the size of quantum system becomes increasingly large.
	\begin{figure}[!htbp]
		\includegraphics[width = \columnwidth]{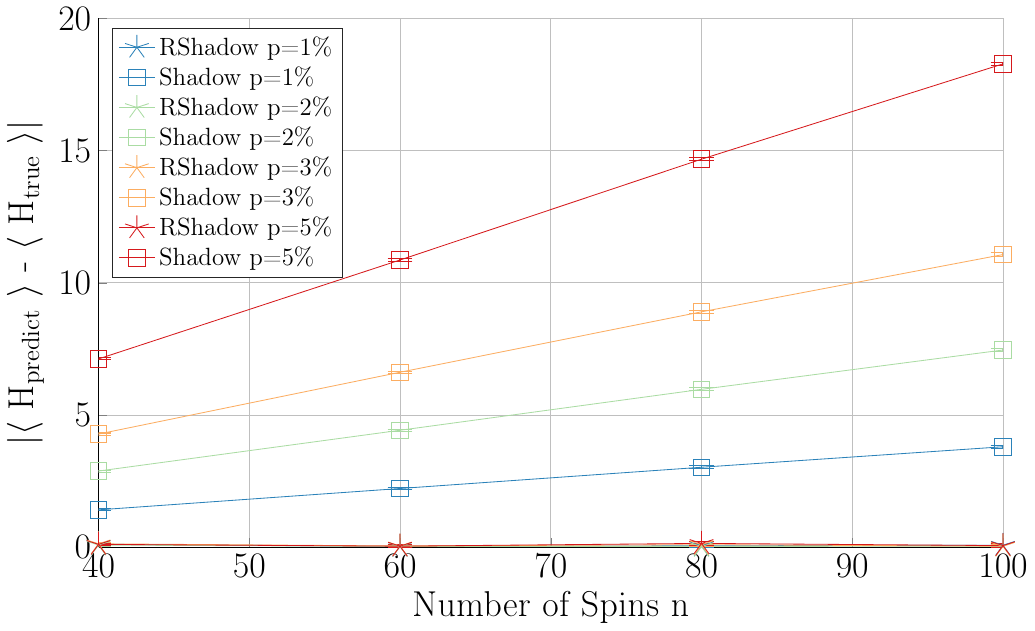}
		\caption{Energy expectation estimation on 1-D TFIM ground state of different spin number with different probabilities of measurement bit flip.}
		\label{fig:TFIM_E_SD}
	\end{figure}

As a remark regarding the computational complexity, we do not calibrate all $f_z$ such that $|z|\le 2$, the number of which scales as $\mc O(n^2)$. 
Instead, we only calibrate the nearest-neighbor terms of $f_z$ for the energy expectation estimation, and the $f_z$ terms such that acts on the first qubit and any other qubit for the correlation function estimation. 
In both case, there are only $\mc O(n)$ parameters to be calibrated. 
Therefore, when the system size gets large, the \textbf{\textrm{RShadow}} protocol remains efficient.
	
	
To demonstrate the noise-resilience of \textbf{\textrm{RShadow}} scheme against $2$-qubit correlated noise, we present more numerical results in Appendix~\ref{sec:app_numer} in the task of estimating the 2-point correlation function of the $n$-qubit GHZ state.
These numerical experiments justify that the \textbf{\textrm{RShadow}} scheme can indeed mitigate the experimental errors and reproduce faithful estimation with a small number of benchmarking trials.

\section{Gate-dependent noise}\label{sec:gatedependent}
	Perhaps the strongest assumption we made is the gate-independence of the noise channel $\Lambda$ with respect to the unitary gate $U$ being sampled. 
	In this section, we present numerical evidence showing that even with an experimentally realistic gate-dependent noise model, \textrm{\textbf{RShadow}} can still greatly reduce noise bias.
	Throughout this section, we focus on \textbf{RShadow} with the local Clifford group, which is experimentally implementable on most near-term platforms.
	The task we consider here is the \emph{electronic structure problem}: decide the ground state energy of a molecule with an unknown electronic structure. 
	This is a important problem in quantum chemistry, and is viewed as one of the most promising applications of near-term quantum algorithms, see \textit{e.g.}~\cite{mcardle2020quantum}. 
	Several recent works have already applied shadow estimation related methods to study this problem~\cite{hadfield2020measurements, huang2021efficient}. 
	
	Specifically, we choose a benchmark molecule and use a certain encoding scheme to map the molecular Hamiltonian into a qubit Hamiltonian. 
	Then, given the ground state of this Hamiltonian, we numerically run the (standard and robust) shadow estimation protocols to estimate its energy, and compare the estimation with the classically computed true value, in the presence of noise. 
	In our setting, we choose $\mr{H}_2$ and apply the Bravyi-Kitaev encoding~\cite{bravyi2002fermionic} to map it to a $4$-qubit Hamiltonian. 
	
	To come up with a realistic gate-dependent noise model, we first need to decide how the local Clifford group is implemented on real experimental platforms. 
	One common approach is to decompose all unitary gates into a small set of generators. 
	Here, we consider the generating set consisting of the following three single-qubit generators
	\begin{equation}
	    \left\{R_P\left(\frac\pi 2\right) = \exp(-i\frac{\pi}{4}P),~P=X,Y,Z\right\}
	\end{equation}
	which can be understood as $\pi/2$ rotations along the X,Y,Z axes respectively. 
	Every single-qubit Clifford gate can be decomposed into two subsequent rotations along two out of these three axes. 
	For example, the Hadamard gate can be implemented by first applying a $\pi/2$ rotation pulse along the Y axis, and then a $\pi$ rotation pulse along the X axis, which is in turn implemented by concatenating two $\pi/2$ X pulses. 
	(See App.~\ref{sec:app_gd} for more details.)
	This generating set is wildly used in real experiments. 
	
	Our numerical simulations will deal with the following two kinds of errors that naturally appear in experiments:
	\begin{enumerate}
	    \item \emph{Pulse mis-calibration}: 
	    The $\pi/2$-pulses have some fixed error due to \textit{e.g.} an uncharacterized constant magnetic field. 
	    These noisy generators would look like
	    $$
	    \widetilde{R}_P = \exp\left(-i\frac12\left(\frac{\pi}{2}P+\Delta_0\right)\right).
	    $$
	    for $P=X,Y,Z$ and some traceless Hermitian operator $\Delta_0$ representing the uncalibrated Hamiltonian.
	    Although $\Delta_0$ is the same for all three generators, the commutator $[P,\Delta_0]$ is in general different for different $P$. 
	    Thus, one can verify that this is a gate-dependent noise model by expanding $\widetilde{R}_P$ using the Baker-Campbell-Hausdorff formula.
	    
	    \item \emph{Random over-rotation}: Due to imperfect pulse control, the actual rotation angle for each generator could be modeled as $\pi/2+\delta$ for some zero-mean Gaussian random variable $\delta$. 
	    The noisy generators look like
	    $$
	    \widetilde{R}_P = \exp\left(-i\frac12\left(\frac{\pi}{2}+\delta\right)P\right).
	    $$
	    We assume that the value of $\delta$ is re-sampled every time a generator is applied. 
	    One can verify that this noise model is equivalent to a dephasing noise on the eigenbasis of Pauli P following the noiseless generator $R_P$, thus it is a gate-dependent noise model. 
	    (See App.~\ref{sec:app_gd}.)
	\end{enumerate}
	
	Our numerical results are presented in Fig.~\ref{fig:gd_main}, where we plot the energy estimation outcome of both standard \textbf{Shadow} and \textbf{RShadow} in the presence of different levels of noise strength. 
	The noise model is \emph{Pulse mis-calibration} for the upper figure and \emph{Random over-rotation} for the lower one. 
	For both noise models, we use $R = 30000~(N=3000,~K=10)$ calibration samples and $R = 10000~(N=1000,~K=10)$ estimation samples for \textbf{RShadow}, and $R = 10000~(N=1000,~K=10)$ samples for standard \textbf{Shadow}~%
	\footnote{One might object that we have taken 4 times as many total samples using \textbf{RShadow} compared with \textbf{Shadow}. While increasing the number of samples in \textbf{Shadow} would indeed improve the precision, it would not impact the \textit{accuracy}, which is where \textbf{RShadow} outperforms \textbf{Shadow} in these numerical simulations}. 
	The data points and the error bars are the average values and the standard deviations over $30$ independent runs~%
	\footnote{Note that, the error bars shown here are obtained in a different manner from those in Sec.~\ref{sec:numer}. 
	Here, we also take into account the deviation for the calibration procedure of \textbf{RShadow}, so the error bars here look longer than those of Sec.~\ref{sec:numer}.}.

	\begin{figure}[!htb]
		\centering
		\includegraphics[width = \columnwidth]{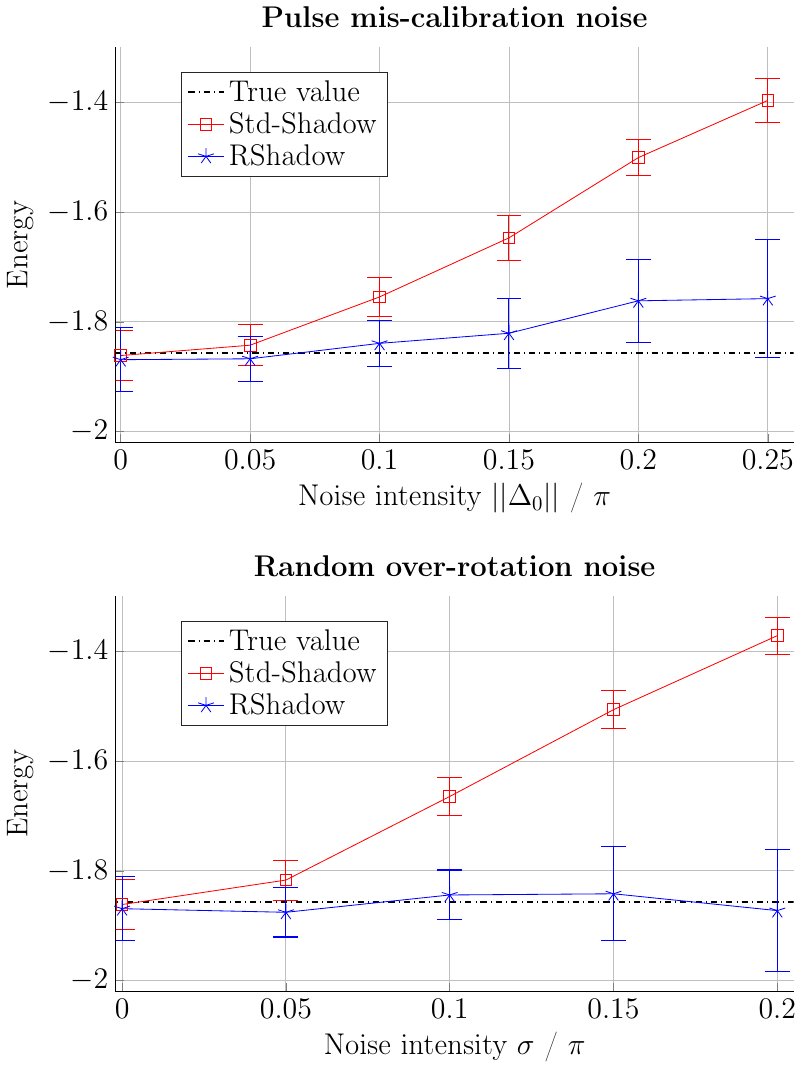}
		\caption{(Upper) Ground-state energy estimation of $\text{H}_2$ with pulse mis-calibration noise. We choose the uncalibrated Hamiltonian $\Delta_0 \coleq \|\Delta_0\| * (aX+bY+cZ)$ where $[a,b,c]= [-0.5500,0.2878,0.7840]$ is a fixed randomly-generated unit vector. See App.~\ref{sec:app_gd} for evidence that this choice of vector is not special; (Lower) Ground-state energy estimation of $\text{H}_2$ with random over-rotation noise. $\sigma$ is the standard deviation of the random over-rotation angle $\delta$.} 
		\label{fig:gd_main}
	\end{figure}
    
    These numerical results provide evidence for the advantage of \textbf{RShadow} over standard shadow estimation even with realistic gate-dependent noise. 
    Specifically, for \emph{Pulse mis-calibration} noises, it seems that \textbf{RShadow} cannot eliminate all biases when the noise strength becomes very large. 
    (Note that, noise intensity $\|\Delta_0\|=0.1\pi$ is already fairly high in practice.) 
    Yet, even in this regime \textbf{RShadow} still significantly outperforms standard shadow estimation and greatly suppresses the bias in the estimated ground state energy. 
    For \emph{Random over-rotation} noise, \textbf{RShadow} completely eliminates all bias even at large noise strengths.
    
    So why does \textbf{RShadow} work for these gate-dependent noise models?  
    One possible explanation is as follows. 
    The key subroutine of \textbf{RShadow} is to learn a Pauli channel. 
    Even though the noise has strong gate-dependence, if the $\widetilde{\mc M}$ channel defined in Eq.~\eqref{eq:Mtilde} is approximately a Pauli channel, and the random unitary gates are ``good enough'' to twirl the input probe state $\lket{\bf 0}$ into an approximate complex projective 2-designs, then we expect \textbf{RShadow} to still perform well.  
    A more rigorous and comprehensive analysis of \textbf{RShadow}'s noise-resilience against general gate-dependent noises is left for future work.

\section{Concluding remarks}
	
We have analyzed the shadow estimation protocol proposed in Ref.~\cite{huang2020predicting} by considering the gate and measurement errors occurring during the process, and have proposed a modified protocol that is robust against such noise. 
We have proven that, in both the global and the local random Clifford group version of the robust shadow protocol, we can efficiently benchmark and suppress the effects caused by the noise. 
On account of the broad application prospects of the shadow estimation protocol in predicting various important properties of quantum states, e.g., entanglement witness, fidelity estimation, correlation functions, \textit{etc.}, we expect that our robust protocol is practical and feasible for current experiments.

While we only focus on estimating linear properties in this work, \textbf{RShadow} can also be used to calibrate the estimation of higher-order properties such as the subsystem \renyi-2 entropy with similar methods shown in~\cite{huang2020predicting}. 
An exploration into the corresponding sample complexity bound is left for future studies. 
It is also interesting to explore how \textbf{\textrm{RShadow}} can be incorporated with other variants of shadow estimation such as the locally-biased classical shadow~\cite{hadfield2020measurements} and derandomized classical shadow~\cite{huang2021efficient}. These two methods can greatly improve the sample efficiency of shadow estimation when one has prior knowledge about which properties are to be predicted, like in the \emph{electronic structure problem}. We believe these techniques can also be applied to \textbf{RShadow} and have been actively developing these methods.


The idea of using additional calibration processes and classical post-processing to eliminate noise effects also appears in the field of error mitigation~\cite{temme2017error,endo2018practical,maciejewski2020mitigationofreadout,bravyi2020mitigating}. 
Among them, it is particularly interesting to compare our work with Ref.~\cite{maciejewski2020mitigationofreadout} and Ref.~\cite{bravyi2020mitigating}, which mitigate the measurement readout error on multi-qubit devices using a measurement calibration (or detector tomography) process. 
The spirit of these works is quite similar to ours, but their assumptions on the noise model are much stronger, and their calibration algorithm is more like a heuristic one without an explicit bound on the sample complexity. 
One reason why we are able to obtain a useful sample complexity bound against a more general noise model is that the random twirling in \textbf{\textrm{RShadow}} greatly simplifies our analysis of the noise estimation.
For future research, it is interesting to explore the relationship between \textbf{\textrm{RShadow}} and other error mitigation schemes, and see if any general results for error mitigation~\cite{takagi2020optimal} can be applied to our scenario. 
Very recently, error mitigation has been shown to be helpful even for fault-tolerant quantum computing~\cite{suzuki2020quantum}. 
We expect \textbf{RShadow} to be a useful protocol in the fault-tolerant regime as well.
	
For our performance guarantee of the robust shadow estimation protocol, the noise in the random gates is allowed to be coherent and highly correlated, but cannot depend on the unitary gate to be implemented. 
This assumption is reasonable in many cases, especially in the protocol with local Clifford gates, where the noise is mainly caused by amplitude damping and decoherence of the system to the environment~\cite{blais2020circuit}. 
Nevertheless, it is important to analyze how gate-dependency and non-Markovianity of the noise can affect the performance of \textbf{RShadow}. 
We have provide some numerical evidences for \textbf{RShadow}'s resilience against gate-dependent noise in Sec.~\ref{sec:gatedependent}, and left more rigorous analysis for future research.

In the PTM representation, the picture of quantum state shadow estimation can be easily extended to the shadow estimation of quantum measurements and quantum channels. 
For example, in order to estimate $\lbra{O_i}\cE\lket{\rho_j}$ for some unknown $n$-qubit quantum channel $\cE$ and a set of given observables $\{O_i\}$ and states $\{\rho_j\}$, one may insert two random measurement channels into the expression,
\begin{equation}
\begin{aligned}
	&\lbra{O_i}\cE\lket{\rho_j} = \lbra{O_i}\mc{M}^{-1}\mc{M}\cE\mc{M}\mc{M}^{-1}\lket{\rho_j} \\&= \mathop{\mathbb{E}}_{U,V\in\mbb G}\sum_{x,y} \lbra{O_i}\mc{M}^{-1}\cU^\dagger\lket{x}\lbra{x}\cU\cE\cV\lket{y}\lbra{y}\cV^\dagger\mc{M}^{-1}\lket{\rho_j}.
\end{aligned}
\end{equation}
In the experiment, one can randomly prepare a computational basis state $\ket{y}$, apply a random unitary $V$ and send to the channel $\cE$, then apply another random unitary $U$ and measure in the computational basis, getting outcome $\ket{x}$. 
Then $2^{-n}\lbra{O_i}\mc{M}^{-1}\cU^\dagger\lket{x}\lbra{y}\cV^\dagger\mc{M}^{-1}\lket{\rho_j}$ is an unbiased estimator of $\lbra{O_i}\cE\lket{\rho_j}$.
This is only the most straightforward way to extend robust shadow estimation to quantum channels; there may exist other schemes that have even better performance. 
We believe a complete analysis of the channel version of shadow estimation will be an interesting direction for further study.

Finally, one can also consider applying (standard or robust) shadow estimation to qudit systems, Boson/Fermion systems and other continuous-variable systems using the techniques developed in this work.

\medskip

\textit{Note added.}~-~ After posting this paper to arXiv, two related but independent works subsequently appeared. 
The independent work by Koh and Grewal~\cite{koh2020shadow} also studies how to mitigate noise in the shadow estimation protocol. 
In their work, the noise channel is assumed to be completely pre-characterized.
The main results Theorem 1.1 and Theorem 1.2 is similar to our Theorem~\ref{th:gl_all} and Theorem~\ref{th:lc_all} if our noise calibration procedure is assumed to be done perfectly. 
The other independent work by Berg, Minev, and Temme~\cite{berg2020model} also contains some similar ideas as presented in our work. 
We thank the authors for communicating their work with us.


\begin{acknowledgments}
	We thank You Zhou for discussions on random unitary schemes,
	Yihong Zhang for suggestions about gate-dependent noise,
	Hsin-Yuan Huang and Jinguo Liu for helpful suggestions on the numerical simulations,
	and John Preskill for discussions on the noise assumptions.
	We also thank Xiongfeng Ma and Richard Kueng for many helpful comments. 
	C.S., W.Y., and P.Z. are supported by the National Natural Science Foundation of China Grants No.~11875173 and No.~11674193, the National Key Research and Development Program of China Grant No.~2019QY0702 and No.~2017YFA0303903, and the Zhongguancun Haihua Institute for Frontier Information Technology.
\end{acknowledgments}

\bibliographystyle{apsrev4-2}

%

\onecolumngrid
\newpage
\begin{appendix}
		
		\section{Preliminaries}  \label{sec:preliminaries}
		In this work, we focus on the $n$-qubit quantum systems with Hilbert space dimension $d = 2^n$. 
		Define $\mc{H}_d$ to be a finite dimensional Hilbert space with the dimension $d$. 
		Define $\mc{L}(\mc{H}_d): \mc{H}_d \to \mc{H}_d$ to be the space of linear operators on $\mc{H}_d$. 
		Define $\mr{Herm}(\mc{H}_d)$ to be the space of Hermitian operator on $\mc{H}_d$, define $\mc{P}((\mc{H}_d))$ to be the set of positive operator on $\mc{H}_d$, and define $\mc{D}(\mc{H}_d)\subset \mc{P}(\mc{H}_d)$ to be the set of quantum states on $\mc{L}(\mc{H}_d)$ which are the positive operators with trace equal to $1$. 
		Sometime we also write $\mc{D}(\mc{H}_d)$ as $\mc{D}(d)$ for simplicity of notations.
		
		\subsection{Groups and representations}\label{App:rep}
		The group representation theory plays an important role in the shadow estimation protocol. 
		Denote a generic group as $\mathbb{G} = \{g_i\}_i$, where $g_i$ is one of the group elements. 
		Denote a unitary representation of $\mathbb{G}$ to be a map
		\begin{equation}
		\phi : \mathbb{G} \to \mc{L}(\mc{H}_d): \mathbb{G} \mapsto \phi(\mathbb{G}),
		\end{equation}
		with the homomorphism
		\begin{equation}
		\phi(g)\phi(h) = \phi(gh), \forall g,h \in \mathbb{G}.
		\end{equation}
		Moreover, we denote all the irreducible representations (irreps.) of the group $\mathbb{G}$ as $R_{\mathbb{G}} = \{\phi_\lambda(\mathbb{G})\}_\lambda$. 
		The Maschke's Lemma~ensures that, every representation of a group can be written as a direct sum of irreps,
		\begin{equation}
		\phi(g) \backsimeq \bigoplus_{\lambda\in R_{\mathbb{G} } } \phi_\lambda(g)^{\otimes m_\lambda}, \forall g\in \mathbb{G},
		\end{equation}
		where $m_\lambda$ is an integer implying the multiplicity of the irrep $\phi_\lambda$. 
		
		In the later discussion, we will frequently come across the \textit{twirling} of an linear operator $O$ on Hilbert space $\mc{H}$ with respect to a group representation $\phi(\mbb{G})$,
		\begin{equation}
		\mc{T}_{\phi}(O) \coleq \frac{1}{|\mbb{G}|} \sum_{g\in\mbb{G}} \phi(g) O \phi(g)^\dag.
		\end{equation}
		As a result of the group structure $\mbb{G}$, the twirling result $\mc{T}_{\phi}(O)$ owns a simple structure, which is related to the irreps in $\phi(\mbb{G})$. 
		The following lemma is a corollary of Schur's lemma.
		
		\begin{lemma} \label{le:schur} (Lemma 1.7 and Prop 1.8 in~\cite{fulton2013representation}, rephrased by~\cite{helsen2019new}) For a finite group $\mbb{G}$ and a representation $\phi$ of $\mbb{G}$ on a complex vector space $\mc{H}$ with decomposition
			\begin{equation}
			\phi(g) \backsimeq \bigoplus_{\lambda\in R_{\mathbb{G} } } \phi_\lambda(g)^{\otimes m_\lambda}, \forall g\in \mathbb{G},
			\end{equation}
			where $\{\phi_\lambda\}$ are the irreps of $\phi(\mbb{G})$, $m_\lambda$ is the multiplicity of $\phi_\lambda$. 
			Then for any linear map $O \in GL(\mc{H})$, the twirling of $O$ with respect to $\phi$ can be written as
			\begin{equation} \label{eq:TwirlDecomp}
			\mc{T}_\phi(O) = \frac{1}{|\mbb{G}|} \sum_{g\in\mbb{G}} \phi(g) O \phi(g)^\dag = \sum_{\lambda\in R_\mbb{G}} \sum_{j_\lambda,j'_\lambda=1}^{m_\lambda} \frac{\tr\left(O \Pi^{j'_\lambda}_{j_\lambda}\right)}{ \tr\left(\Pi^{j'_\lambda}_{j_\lambda}\right) } \Pi^{j'_\lambda}_{j_\lambda},
			\end{equation}
			where $\Pi^{j'_\lambda}_{j_\lambda}$ is a linear map from the support of the $j'_\lambda$-th copy of $\phi_\lambda$ to the support of the $j_\lambda$-th copy of $\phi_\lambda$. 
		\end{lemma}
		
		In this work, we focus on the group representation $\phi$ with no multiplicities,
		\begin{equation}
		\phi(g) \backsimeq \bigoplus_{\lambda\in R_\mbb{G}} \phi_\lambda(g), \quad \forall g\in\mbb{G}.
		\end{equation}
		In this case, Eq.~\eqref{eq:TwirlDecomp} can be simplified as
		\begin{equation} \label{eq:TwirlDecompNoMulti}
		\mc{T}_\phi(O) = \frac{1}{|\mbb{G}|} \sum_{g\in\mbb{G}} \phi(g) O \phi(g)^\dag = \sum_{\lambda\in R_\mbb{G}} \frac{\tr(O \Pi_{\lambda})}{ \tr(\Pi_{\lambda}) } \Pi_{\lambda},
		\end{equation}
		where $\Pi_{\lambda}$ is the projector onto the support of $\phi_\lambda$.
		
		Here, we introduce some common groups that will be frequently used. 
		Note that all the linear operators in $\mc{L}(\mc{H}_d)$ form a Lie group $GL(d,\mbb{C})$. 
		The unitaries in $\mc{L}(\mc{H}_d)$ also form a Lie group $U(d)$.
		
		Denote $\mbb{Z}_2 = \{0,1\}$ to be the 2-element cyclic group. 
		$\mbb{Z}^{n}_2 \coleq (\mbb{Z}_2)^{\otimes n}$ is the $n$-copy tensor of $\mbb{Z}_2$ group. 
		Denote $\mathbb{A} = \langle \{a_i\} \rangle$ with $\{a_i\}$ the generators of the group. 
		In the later discussion, we will also slightly abuse $\mbb{Z}_2^n$ to denote the set of $n$-bit binary string.
		
		For $n$-qubit quantum system, the Pauli group is
		\begin{equation}
		\mathbb{P}^n =  \{ \langle i \rangle \otimes \{I,X,Y,Z\} \}^{\otimes n},
		\end{equation}
		with $I,X,Y,Z$ the qubit Pauli matrices. 
		Denote the quotient of $\mathbb{P}^n$ to be $\mathsf{P}^n = \mathbb{P}^n/\langle i \rangle$, which is an Abelian group and isomorphic to $\mbb{Z}_2^{2n}$. 
		Therefore, we will use a $2n$-bit string to denote the elements in $\mathsf{P}^n$ and choose the elements to be
		\begin{equation}
		P_a = P_{(a_x,a_z)} = i^{a_x\cdot a_z} X^{\otimes a_x} Z^{\otimes a_z}. 
		\end{equation}
		
		The multiplication and commutation of elements in $\mathsf{P}^n$ follows,
		\begin{equation}
		\begin{aligned}
		P_a P_b &= (-i)^{\braket{a,b}} P_{a+b}, \\
		P_a P_b &= (-1)^{\braket{a,b}} P_b P_a,
		\end{aligned}
		\end{equation}
		with
		\begin{equation}
		\braket{a,b} \coleq a_x \cdot b_z - a_z \cdot b_x \quad \text{mod } 4,
		\end{equation}
		a binary symplectic product. 
		This symplectic product owns the following properties
		\begin{equation}
		\begin{aligned}
		\braket{a,b} &= - \braket{b,a}, \\
		(-i)^{\braket{a,b}} &= i^{-\braket{a,b}}, \\
		(-1)^{\braket{a,b}} &= (-1)^{\braket{b,a}}.
		\end{aligned}
		\end{equation}
		
		The $n$-qubit Clifford group ${\sf Cl}(2^n)$ is defined to be
		\begin{equation}
		{\sf Cl}(2^n) = \{ g | g P_a g^{-1} \in \mathbb{P}^n, \forall P_a \in P^n \} / U(1),
		\end{equation}
		where the $U(1)$ represents the global phase. 
		Obviously, $P^n$ is a subgroup of $\mathbb{C}^n$. 
		The single-qubit Clifford group is then ${\sf Cl}_2\coleq {\sf Cl}(2)$. 
		Later we will also come across the tensor-ed $n$-fold single-qubit Clifford group ${\sf Cl}_2^{\otimes n}$.

		\subsection{Random unitaries and t-designs}
		
		The shadow estimation is a direct application of twirling in random unitaries. The ideal ``uniformly distributed'' randomized unitaries over the Lie group $GL(d,\mbb{C})$ is characterized by \textit{Haar measure} $\mu(\mc{H}_d)$~\cite{collins2016random}. The Haar measure is defined to be the unique countably additive, nontrivial measure of the group $U$ such that,
		\begin{equation}
		\int_{\mu(\mc{H}_d)} dU = 1, \quad \int_{\mu(\mc{H}_d)} dU f(U) = \int_{\mu(\mc{H}_d)} dU f(UV) = \int_{\mu(\mc{H}_d)} dU f(VU),
		\end{equation}
		where $f(U)$ is any matrix function of $U$.
		
		In practice, to sample unitaries with respect to Haar measure is challenging due to its continuity. Alternatively, one may choose to sample from a finite subset $\mc{K} = \{U_k\}_{k=1}^{|\mc{K}|}$ over the unitaries in $GL(d,\mbb{C})$.
		
		\begin{definition}
			A finite subset $\mc{K} = \{U_k\}_{k=1}^{|\mc{K}|} \subset \mc{U}(d)$ is a unitary $t$-design if
			\begin{equation}
			\frac{1}{|\mc{K}|} \sum_{k=1}^{|\mc{K}|} f_{(t,t)}(U_k) = \int_{\mu(\mc{H}_d)} dU f_{(t,t)}(U),
			\end{equation}
			for all the polynomial $f_{(t,t)}(U)$ of degree at most $t$ in the matrix elements of $U$ and at most $t$ in the matrix elements of $U^*$.
		\end{definition}
		
		It has been proven that, the Clifford gate set ${\sf Cl}(d)\subset \mc{U}(\mc{H})$ is a unitary $3$-design~\cite{webb2015clifford,zhu2017multiqubit}, while fails to be a unitary $4$-design~\cite{zhu2016clifford}.

		\subsection{Quantum channel and the representations}\label{Liouville}
		
		Quantum channels are the linear maps $\mc{E}:\mc{L}(\mc{H}_d) \to \mc{L}(\mc{H}_d)$ which are completely positive and trace-preserving (CPTP). 
		
		\begin{definition} Let $\mc{E}: \mc{L}(\mc{H}_d) \to \mc{L}(\mc{H}_d)$ be a linear map. We say that
			\begin{enumerate}
				\item $\mc{E}$ is positive if $\mc{E}(\rho) \in \mc{D}(\mc{H}_d)$ for any $\rho \in \mc{D}(\mc{H}_d)$.
				\item $\mc{E}$ is completely positive (CP) if $\mc{I}_{d'} \otimes \mc{E}$ is positive, for all the dimension $d'$.
				\item $\mc{E}$ is trace preserving (TP) if $\tr[\mc{E}(\rho)]=1$ for any $\tr[\rho]=1$.
				\item $\mc{E}$ is a quantum channel if it is both CP and TP.
			\end{enumerate}
		\end{definition}
		
		In this work, we will come across two representations of the quantum channels : Kraus representation and Liouville representation. For a quantum channel $\mc{E}:\mc{L}(\mc{H}_d) \to \mc{L}(\mc{H}_d)$, its action on a linear operator $O\in\mc{L}(\mc{H}_d)$ can be expressed as
		\begin{equation}
		\mc{E}(O) = \sum_{t=1}^{k} K_t O K_t^\dag, 
		\end{equation}
		where $\{K_t\}_{t=1}^k$ are the Kraus operators satisfying $\sum_{t=1}^k K_t^\dag K_t = I$.
		
		To represent the effect of quantum channels in a convenient way, we first introduce the Pauli basis $P^n$ on $\mc{L}(\mc{H}_{d})$ to vectorize the linear operators in $\mc{L}(\mc{H}_{d})$. Define the inner product between two operators to be the Hilbert-Schmidt product
		\begin{equation}
		\braket{Q,W} \coleq \tr(Q W^\dag), \quad \forall Q,W \in GL(\mc{H}_{d}).
		\end{equation}
		In this case, the operators in $P^n$ form an orthogonal basis. We introduce the operators
		\begin{equation}
		\sigma_a = P_a/\sqrt{d},
		\end{equation}
		as the orthonormal basis. To vectorize the linear space spanned by $\{\sigma_a\}$, we introduce the notation $\{\lket{\sigma_a}\}$. For the single-qubit case, we will also use the following notations,
		\begin{equation}
		\begin{aligned}
		&\sigma_I = \sigma_0 = \sigma_{(0,0)}, \quad \sigma_X = \sigma_{(1,0)}, \\
		&\sigma_Z = \sigma_1 = \sigma_{(0,1)}, \quad \sigma_Y = \sigma_{(1,1)}.
		\end{aligned}
		\end{equation}
		
		Then the operators on $\mc{L}(\mc{H}_{d})$ can be vectorized as
		\begin{equation}
		\lket{Q} = \sum_{a\in Z^{2n}_2} \lbraket{Q}{\sigma_a} \lket{\sigma_a}.
		\end{equation}
		
		The quantum channel $\mc{E}$ can then be represented as 
		\begin{equation}
		\mc{E} = \sum_{a,b\in Z^{2n}_2} \lbra{\sigma_a} \mc{E} \lket{\sigma_b} \lket{\sigma_a}\lbra{\sigma_b},
		\end{equation}
		with
		\begin{equation}
		\lbra{\sigma_a}\mc{E} \lket{\sigma_b} \coleq \braket{\sigma_a, \mc{E}(\sigma_b)}.
		\end{equation}
		
		The matrix $\mc{E}$ is the Pauli-transfer matrix (PTM) or Pauli-Liouville representation. In this work, we slightly abuse the notation of a superoperator $\mc{E}$ to represent its PTM. For a unitary matrix $U$, we use the calligraphic $\cU$ to represent its PTM.
		
		For a quantum channel $\mc{E}$ with state $\rho$ input, and POVM measurement $M = \{M_b\}$ with $\sum_{b} M_b = I$, the probability to get the measurement result $b$ is
		\begin{equation}
		p_b = \lbra{M_b} \mc{E} \lket{\rho}.
		\end{equation}
		
		Under the PTM representation, the composition and tensor product of channels $\mc{E}_1$ and $\mc{E}_2$ can be naturally expressed as
		\begin{equation}
		\begin{aligned}
		\lket{\mc{E}_1 \circ \mc{E}_2(\rho)} &= \mc{E}_1 \mc{E}_2 \lket{\rho}, \\
		\lket{\mc{E}_1 \otimes \mc{E}_2(\rho^{\otimes 2})} &= \mc{E}_1 \otimes \mc{E}_2 \lket{\rho^{\otimes 2}}. \\
		\end{aligned}
		\end{equation}
		
		The PTM of the unitaries in $U(d)$ forms a natural group representation of $U(d)$. Denote the PTM of a given unitary $U$ as $\phi^P(U)\coleq \mc{U}$, we have
		\begin{equation}
		\phi^P(U) \phi^P(V) = \phi^P(UV).
		\end{equation}
		
		The PTM representation $\phi^P(U(d))$ can be decomposed to two irreps,
		\begin{equation}
		\phi^P(U) \backsimeq \phi^P_{I}(U) \oplus \phi^P_{\sigma}(U), \quad \forall U \in U(d).
		\end{equation}
		Here,
		\begin{equation}
		\begin{aligned}
		\phi^P_I(U) &= \Pi_I \, \phi^P(U) \, \Pi_I, \\
		\phi^P_I(U) &= \Pi_\sigma \, \phi^P(U) \, \Pi_\sigma, \\
		\end{aligned}
		\end{equation}
		the projectors $\Pi_I$ and $\Pi_\sigma$ are
		\begin{equation}
		\begin{aligned}
		\Pi_I &= \lketbra{ \sigma_{0}^{\otimes n} }{ \sigma_{0}^{\otimes n} }, \\
		\Pi_\sigma &= I - \Pi_I = \sum_{a\in \mbb{Z}_2^{2n},\; a \neq (0,0)^{\otimes n} } \lketbra{ \sigma_{a} }{ \sigma_{a} }.
		\end{aligned}
		\end{equation}
		
		The $n$-qubit Clifford group ${\sf Cl}(2^n)$, as the subset of $n$-qubit unitary group, can also be represented by the PTM matrices. The PTM representation $\phi^P({\sf Cl}(2^n))$ can be decomposed similarly,
		\begin{equation} \label{eq:ClPTMdecomp}
		\phi^P(U) \backsimeq \phi^P_{I}(U) \oplus \phi^P_{\sigma}(U), \quad \forall U \in {\sf Cl}(2^n),
		\end{equation}
		where $\phi^P_{I}$ and $\phi^P_{\sigma}$ are two irreps on the support $\Pi_I$ and $\Pi_\sigma$, respectively.
		
		\subsection{Weingarten Function}\label{sec:Weingarten}
		In this part, we introduce the Weingarten function as a tool to calculate general Haar integrals~\cite{weingarten1978asymptotic, collins2003moments, collins2006integration}.
		The following presentation owes a lot to Section 2 of~\cite{roberts2017chaos}.
		
		\medskip
		
		\noindent For an operator $A$ acting on $\cH_d^{\otimes k}$, define the k-fold Haar twirling of $A$ as
		\begin{equation}
		\Phi_\mr{Haar}^{(k)}(A)\coleq\int_{\mu(\cH_d)}dU(U^{\otimes k})^\dagger A U^{\otimes k}.
		\end{equation}
		Using Schur-Weyl duality, one can show that
		\begin{equation}\label{eq:Weingarten}
		\Phi_\mr{Haar}^{(k)}(A) = \sum_{\pi,\sigma\in S_k}c_{\pi,\sigma}W_\pi\Tr(W_\sigma A).
		\end{equation}
		Here, $S_k$ is the $k$-element permutation group, and $W_\pi$ is the permutation operator defined as follows
		\begin{equation}\label{eq:perm}
		W_\pi \ket{a_1,...,a_k} = \ket{a_{\pi(1)},...,a_{\pi(k)}},\quad\forall\ket{a_1,...,a_k}\in \cH_d^{\otimes k},~\pi\in S_k,
		\end{equation}
		and the coefficients $c_{\pi,\sigma}$ are the Weingarten matrix~\cite{collins2003moments} which can be calculated as
		\begin{equation}\label{eq:Wein_form}
		c_{\pi,\sigma} = (Q^{+})_{\pi,\sigma},\quad Q_{\pi,\sigma}\coleq d^{\#\mr{cycles}(\pi\sigma)},
		\end{equation}
		where Q is called the Gram matrix. $Q^{+}$ stands for the {Moore–Penrose pseudo inverse} of $Q$, which is $Q^{-1}$ when $Q$ is invertible. (Note that, when $Q$ is not invertible, $c$ is not uniquely determined. It is only a conventional choice to take $c = Q^{+}$~\cite{collins2006integration,zinn2010jucys}).
		
		\medskip
		
		\noindent In following sections, we are interested in the case $k=3$. We sort the elements of $S_3$ in the following order 
		\begin{equation}\label{eq:perm_vec}
		\vec W \coleq \begin{bmatrix}W_{()}, & W_{(1,2)}, & W_{(1,3)}, & W_{(2,3)}, & W_{(1,2,3)}, & W_{(1,3,2)} \end{bmatrix}.
		\end{equation}
	    In this basis, the Gram matrix becomes
		\begin{equation}\label{eq:Gram3}
		Q=\begin{bmatrix}
		d^3&d^2&d^2&d^2&d&d \\
		d^2&d^3&d&d&d^2&d^2\\
		d^2&d&d^3&d&d^2&d^2\\
		d^2&d&d&d^3&d^2&d^2\\
		d&d^2&d^2&d^2&d&d^3\\
		d&d^2&d^2&d^2&d^3&{d}
		\end{bmatrix},
		\end{equation}
		For $d\ge 3$, one can show that the Weingarten matrix becomes
		\begin{equation}\label{eq:Wein3}
		c=\frac{1}{d(d^2-1)(d^2-4)}\begin{bmatrix}
		d^2-2&-d&-d&-d&2&2 \\
		-d&d^2-2&2&2&-d&-d\\
		-d&2&d^2-2&2&-d&-d\\
		-d&2&2&d^2-2&-d&-d\\
		2&-d&-d&-d&2&d^2-2\\
		2&-d&-d&-d&d^2-2&{2}
		\end{bmatrix},
		\end{equation}
		while for $d=2$, $Q$ is singular, so we take its pseudo inverse as follows
		\begin{equation}\label{eq:Wein32}
		c=\frac{1}{144}\begin{bmatrix}
		17&1&1&1&-7&-7 \\
		1&17&-7&-7&1&1\\
		1&-7&17&-7&1&1\\
		1&-7&-7&17&1&1\\
		-7&1&1&1&-7&17\\
		-7&1&1&1&17&{-7}
		\end{bmatrix}.
		\end{equation}

		\section{Sample Complexity of \textbf{\textrm{RShadow}} with Global Clifford Group} \label{sec:app_gl}
		
		In this section, we study our robust shadow estimation protocol with $\mbb{G}$ chosen to be the $n$-qubit Clifford group ${\sf Cl}(2^n)$.
		
		\subsection{Calibration Procedure: Global}
		
		 Recall that the channel $\widetilde{\cM}$ can be written on the Pauli basis as
		\begin{equation}\label{eq:Mdef}
		\widetilde\cM = \mathop{\mathbb{E}}_{U \sim {\sf Cl}(2^n)} \cU^\dagger \cM_z\Lambda \cU =    \left[\begin{matrix}
		1      & 0      & \cdots & 0      \\
		0      & f      & \cdots & 0      \\
		\vdots & \vdots & \ddots & \vdots \\
		0      & 0      & \cdots & f      \\
		\end{matrix}\right]
		\end{equation}
		for some $f\in\mathbb R$ depending on $\Lambda$. 
		Note that $f=(d+1)^{-1}$ when the noise channel is trivial, \textit{i.e.} $\Lambda = \id$. 
		We rewrite the \textbf{\textrm{RShadow}} protocol from the main text as below
		\begin{protocol} \label{proto:Global} [{\textbf{\textrm{RShadow}} with ${\sf Cl}(2^n)$}]
			\begin{enumerate}
				\item Prepare $\ket{\bvec 0}\equiv\ket 0^{\otimes n}$. Sample $U$ uniformly form ${\sf Cl}(2^n)$ and apply it to $\ket {\bvec 0}$.
				\item Measure the above state in the computational basis. Denote the outcome state vector as $\ket b$.
				\item Calculate the single-round estimator of $f$ as $\hat f^{(r)}\coleq\cfrac{d\hat F^{(r)}-1}{d-1}$ where $\hat{F}^{(r)}\coleq \left|\bra b U\ket {\bvec 0} \right|^2$.
				\item Repeat step 1-3 $R=NK$ rounds. Then the final estimation of $f$ is given by a median of mean estimator $\hat f$ constructed from the single round estimators $\{\hat f^{(r)}\}_{r=1}^R$ with parameters $N,~K$ (see Eq.~\eqref{eq:meanmedian_estimator}).
				\item After the above steps, apply the standard classical shadow protocol of~\cite{huang2020predicting} on $\rho$ with the inverse channel $\mc M^{-1}$ replaced by $$\widehat{\cM}^{-1} \coleq \left[\begin{matrix}
				1      & 0      & \cdots & 0      \\
				0      & \hat f^{-1}      & \cdots & 0      \\
				\vdots & \vdots & \ddots & \vdots \\
				0      & 0      & \cdots & \hat f^{-1}       \\
				\end{matrix}\right]$$ 
				in the Liouville representation.
			\end{enumerate}
		\end{protocol}
		
		In Protocol~\ref{proto:Global}, the unitary operations and the measurement are assumed to contain gate-independent noise, and the preparation of $\ket {\bvec 0}$ is assumed to be perfect. The next theorem shows that $\hat f^{(r)}$ is an unbiased estimator of $f$ and its variance can be bounded.
		
		\begin{proposition}\label{prop:gl_main}
			The single-round fidelity estimator $\hat F^{(r)}$ given in Protocol~\ref{proto:Global} satisfies
			\begin{equation}\label{eq:main_1}
			\mathbb{E}(\hat F^{(r)}) = F_{\text{avg}}(\widetilde{\mc M}) = \cfrac{F_{Z}(\Lambda) + 1}{d+1},\qquad \mathrm{Var}(\hat F^{(r)}) \le \cfrac2{d^2},
			\end{equation}
			where $F_{\text{avg}}(\widetilde{\mc M})=\int_{\psi\in\text{Haar}} d\psi \lbra\psi\widetilde{\mc M}\lket\psi$ is the average fidelity of $\widetilde{\mc M}$, and $F_{Z}(\Lambda)=\frac1{2^n}\sum_{b\in\{0,1\}^n} \lbra b \Lambda \lket b$ is the Z-basis average fidelity of $\Lambda$. 
			
			Moreover, the single-round estimator $\hat f$ satisfies
			\begin{equation}\label{eq:main_2}
			\mathbb{E}(\hat f^{(r)}) = f= \cfrac{d F_{Z}(\Lambda)-1}{d^2-1},\qquad \mathrm{Var}(\hat f^{(r)}) \le \cfrac2{(d-1)^2}.
			\end{equation}
		\end{proposition}
		
		Before we provide the proof of Proposition~\ref{prop:gl_main}, we first introduce two lemmas. 
		\begin{lemma}(see \textit{e.g.}~\cite[Proposition~4]{zhu2016clifford}) \label{le:tdesign}
			If a group $\mbb{G}\subseteq U(d)$ forms a unitary $t$-design, then
			\begin{equation}
			\mathop\mathbb E_{U\sim \mbb{G}}(U\ketbra {\bvec 0} {\bvec 0} U^\dagger)^{\otimes t} = \cfrac{P_{\text{sym}^t}}{\binom{d+t-1}{t}},
			\end{equation}
			where $P_{\text{sym}^t}$ is the projector onto the $t$-fold symmetric space, or equivalently, $P_{\text{sym}^t} = \frac{1}{|S_t|}\sum_{\pi\in S_t}{W_\pi}$ where $W_\pi$ is the permutation operator defined in Eq.~\eqref{eq:perm}.
		\end{lemma} 
		
		\begin{lemma}\label{le:sym}
			For two operators $A$, $B$ acting on $\mc H(d)$,
			\begin{align}
			\tr(P_{\text{sym}^2}A\otimes B) &= \frac12(\tr A\tr B+\tr (AB))\\
			\tr(P_{\text{sym}^3}A\otimes B\otimes B) &= \frac16(\tr A(\tr B)^2+\tr A\tr (B^2) +2\tr (AB)\tr B+2\tr(AB^2)).
			\end{align}
		\end{lemma}
		
		\begin{proof}[Proof of Lemma~\ref{le:sym}]
        For the first equation,
		\begin{equation}
		    \Tr(P_{\text{sym}^2}A\otimes B) = \frac12\left[\Tr(I(A\otimes B))+\Tr(S(A\otimes B))\right] =\frac12(\tr A\tr B+\tr (AB)),
		\end{equation}
		where $S$ is the swap operator. 
		
		\smallskip
		
		\noindent For the second equation, using the language of tensor network (see \textit{e.g.}~\cite[Sec.~3.1]{gross2015partial}), we can derive,
		\begin{equation}\label{eq:perm_tr}
		\Tr(\vec W (A\otimes B\otimes B)) = \begin{bmatrix}
		\Tr A(\Tr B)^2,& \Tr(AB)\Tr B,&\Tr(AB)\Tr B,&\Tr A\Tr(B^2),&\Tr(AB^2),&\Tr(AB^2)
		\end{bmatrix},
		\end{equation}
        where $\vec W$ is a vectorization of $S_3$ defined in Eq.~\eqref{eq:perm_vec}. Averaging this up gives the second equation.
		\end{proof}
		
		\medskip

		Now we present the proof of Proposition~\ref{prop:gl_main}. 
		
		\begin{proof}[Proof of Proposition~\ref{prop:gl_main}]
			Firstly, from Eq.~\eqref{eq:Mdef} we immediately have
			\begin{equation}
			f = \cfrac{\tr(\widetilde{\mc M})-1}{d^2-1}.
			\end{equation}
			We also have the following relation between the average fidelity of a channel $\widetilde{\mc M}$ and the trace of its Pauli transformer matrix (see \textit{e.g.}~\cite{helsen2019new}),
			\begin{equation}
			F_{\text{avg}}(\widetilde{\mc M}) = \cfrac{d^{-1}\tr(\widetilde{\mc M})+1}{d+1}.
			\end{equation}
			Combining the above two equations, we get
			\begin{equation}
			f = \cfrac{d F_{\text{avg}}(\widetilde{\mc M})-1}{d-1},
			\end{equation}
			hence Eq.~\eqref{eq:main_2} follows directly from Eq.~\eqref{eq:main_1}. We only need to calculate the expectation and variance of $\hat F^{(r)}$.
			
			Denote the Kraus operators of the noise channel $\Lambda$ as $\{K_t\}$. The average fidelity of $\widetilde{\mc M}$ can be explicitly written as follows
			\begin{equation}
			\begin{aligned}
			F_{\text{avg}}(\widetilde{\mc M}) &= \int_{d\psi}\mathop{\mathbb{E}}_{U\sim Cl}\bra\psi U^\dagger \circ M_z \circ \Lambda \circ U \left(\ketbra{\psi}{\psi}\right) \ket\psi\label{gate_independent}\\ 
			&= \int_{d\psi}\mathop{\mathbb{E}}_{U\sim Cl} \sum_{b,t}\bra\psi U^\dagger \ketbra b b K_t U \ketbra \psi \psi U^\dagger K_t^\dagger  \ketbra b b  U  \ket\psi \\
			&= \int_{d\psi}\mathop{\mathbb{E}}_{U\sim Cl} \sum_{b,t}\left|\bra b K_t U \ket\psi\right|^2\left|\bra b  U  \ket\psi\right|^2.
			\end{aligned}
			\end{equation}
			
			On the other hand, the expectation of $\hat F^{(r)}$ can be expressed as
			\begin{equation}\label{eq:EFhat}
			\begin{aligned}
			\mathbb E(\hat F^{(r)}) &= \mathop{\mathbb{E}}_{U\sim Cl} \sum_{b,t} \left|\bra b K_t U\ket {\bvec 0}\right|^2 \left| \bra b U\ket {\bvec 0} \right|^2\\
			&= \mathop{\mathbb{E}}_{V\sim Cl} \mathop{\mathbb{E}}_{W\sim Cl} \sum_{b,t} \left|\bra b K_t V W\ket {\bvec 0}\right|^2 \left| \bra b V W\ket {\bvec 0} \right|^2\\
			&= \int_{d\psi}\mathop{\mathbb{E}}_{V\sim Cl} \sum_{b,t} \left|\bra b K_t V\ket \psi\right|^2 \left| \bra b V\ket \psi \right|^2,
			\end{aligned}
			\end{equation}
			where the first equality is by definition of expectation, the second equality is by the fact that sampling an element $U$ from a group is equivalent to independently sampling two elements $V$, $W$ from the group and taking $U = V\circ W$, and the last equality uses the fact that ${\sf Cl}(2^n)$ is a unitary 2-design. As a result, we've shown that \begin{equation}
			\mathbb E(\hat F^{(r)}) = F_{\text{avg}}(\widetilde{\mc M}).
			\end{equation}
			
			Next, in order to get $\mathrm{Var}(\hat F^{(r)})$, we calculate the value of $\mathbb E(\hat F^{(r)})$ and $\mathbb E(\hat F^{(r)^2})$ explicitly. Based on Lemma~\ref{le:tdesign} and \ref{le:sym}, and recalling the fact that ${\sf Cl}(2^n)$ is a unitary 3-design~\cite{zhu2017multiqubit, webb2015clifford, kueng2015qubit}, we are able to do the following calculations.
			\begin{equation}
			\begin{aligned}
			\mathbb E(\hat F^{(r)}) &= \mathop{\mathbb{E}}_{U\sim Cl} \sum_{b,t} \left|\bra b K_tU\ket {\bvec 0}\right|^2 \left| \bra b U\ket {\bvec 0} \right|^2\\
			&= \sum_{b,t} \tr\left[ \mathop{\mathbb{E}}_{U\sim Cl}\left(U\ketbra {\bvec 0} {\bvec 0} U^\dagger\right)^{\otimes 2} ~ \left(K_t^\dagger\ketbra b b K_t \otimes \ketbra b b \right)\right]\\
			&= \cfrac{2}{(d+1)d} ~\sum_{b,t} \tr\left[ P_{\text{sym}^2} ~ \left(K_t^\dagger\ketbra b b K_t \otimes \ketbra b b \right)\right]\\
			&= \cfrac{2}{(d+1)d} ~\sum_{b,t}\cfrac12 \left( \bra b K_t K_t^\dagger\ket b + \left|\bra b K_t\ket b\right|^2 \right)\\
			&= \cfrac{1}{(d+1)d} ~ \left( d + \sum_{b,t}\left|\bra b K_t\ket b\right|^2 \right)\\
			&= \cfrac{1+ F_{Z}}{d+1}~,
			\end{aligned}
			\end{equation}
			\begin{equation}\label{eq:global_2ndmoment}
			\begin{aligned}
			\mathbb E(\hat F^{(r)^2}) &= \mathop{\mathbb{E}}_{U\sim Cl} \sum_{b,t} \left|\bra b K_tU\ket {\bvec 0}\right|^2 \left| \bra b U\ket {\bvec 0} \right|^4\\
			&=\sum_{b,t} \tr\left[ \mathop{\mathbb{E}}_{U\sim Cl}\left(U\ketbra {\bvec 0}{\bvec 0} U^\dagger\right)^{\otimes 3} ~ \left(K_t^\dagger\ketbra b b K_t \otimes \ketbra b b \otimes \ketbra b b\right)\right]\\
			&=\frac{6}{(d+2)(d+1)d}\sum_{b,t}\tr\left[ P_{\text{sym}^3} ~ \left(K_t^\dagger\ketbra b b K_t \otimes \ketbra b b \otimes \ketbra b b\right)\right]\\
			&= \frac{6}{(d+2)(d+1)d}\sum_{b,t}\frac{1}{3}\left(\bra b K_t K_t^\dagger \ket b+2|\bra b K_t \ket b|^2\right)\\
			&= \frac{2(1+2F_{Z})}{(d+2)(d+1)},
			\end{aligned}
			\end{equation}
			where we write $F_{Z}\equiv F_{Z}(\Lambda)$ as the Z-basis average fidelity of $\Lambda$.

			\medskip
			Now we can bound the variance of $\hat F$ as follows
			\begin{equation}
			\begin{aligned}
			\mathrm{Var}(\hat F^{(r)}) &= \mathbb E(\hat F^{(r)^2}) - (\mathbb E(\hat F^{(r)}))^2\\
			&= \cfrac{-(d+2)F_{Z}^2+2d F_{Z}+d}{(d+2)(d+1)^2}\\
			&\le \frac 2{d^2}.
			\end{aligned}
			\end{equation}
			where we use the fact that $F_{Z}\le 1$. This completes the proof.
		\end{proof}
		
		
		
		Now we analyse the sample complexity of Protocol~\ref{proto:Global} in order to guarantee the protocol to succeed within a given level of precision. Specifically, we consider using the protocol to estimate a linear function of $\rho$, \textit{i.e.} $\lbraket{O}{\rho}$. Given that one makes sufficiently many samples in the estimation procedure, the estimation of this function will be close to $\lbra O \widehat{\cM}^{-1}\widetilde{\cM}\lket \rho$. Hence, we are concerned about the following error
		\begin{equation}\label{eq:derivation_operator}
		\begin{aligned}
		& \left| \lbra O \widehat{\cM}^{-1}\widetilde{\cM} \lket\rho-\lbraket{O}{\rho} \right|\\
		=& \left| \lbra O  \left[\begin{matrix}
		0      & 0      & \cdots & 0      \\
		0      & \hat f^{-1}f - 1      & \cdots & 0      \\
		\vdots & \vdots & \ddots & \vdots \\
		0      & 0      & \cdots & \hat f^{-1}f - 1      \\
		\end{matrix}\right]            \lket \rho   \right|.\\
		=& \left|\lbraket{O_0}{\rho}\right| \cdot \left|\hat f^{-1}f - 1\right|\\
		\le& \|O_0\|_\infty \cdot  \left|\hat f^{-1}f - 1\right|,
		\end{aligned}
		\end{equation}
		where $O_0 = O - \cfrac{\tr(O)}{d} I$ is the traceless part of $O$. 
		Now we want to upper bound $|\hat f^{-1}f-1|$ by some $\varepsilon>0$. 
		Suppose with high probability the estimator in Protocol~\ref{proto:Global} satisfies $|\hat f - f|\le \gamma$ for some $0\le\gamma\le|f|$. Then we have,
		\begin{equation}\label{eq:gamma_epsilon}
		    |\hat f^{-1}f-1| = |\hat f^{-1}|\cdot|\hat f - f|\le  \cfrac{\gamma}{|\hat f|} \le \cfrac{\gamma}{| f|-\gamma}
		\end{equation}
		where the last inequality is by the triangular inequality. Now if we have
		\begin{equation}
		    \cfrac{\gamma}{| f|-\gamma}\le \varepsilon \Longleftrightarrow \gamma \le \cfrac{\varepsilon|f|} {1+\varepsilon},
		\end{equation}
		then we obtain the bound $|\hat f^{-1}f-1|\le \varepsilon$ with high success probability.
		Now is the time to calculate the number of rounds $R$ in order to bound $|\hat f - f|$ as we want with high confidence. As noted before, we will uses the median of means estimator~\cite{jerrum1986random,nemirovsky1983problem} in order to get a preferable scaling with respect to the failing probability. Similar techniques are also applied in~\cite{huang2020predicting}. Specifically, we conduct $R=KN$ rounds of the procedure in Protocol~\ref{proto:Global}, calculate $K$ estimators each of which is the average of $N$ single-round estimators $\hat f$, and take the median of these $K$ estimators as our final estimator $\hat f$. In formula,
		\begin{equation} \label{eq:meanmedian_estimator}
		\begin{aligned}
		\bar{f}^{(k)} &\coleq \cfrac1N\sum_{r=(k-1)N+1}^{kN}\hat f^{(r)},\quad k=1,2,...,K.  \\
		\hat{f} &\coleq \text{median}\left\{ \bar{f}^{(1)},~\bar{f}^{(2)},~...,~\bar{f}^{(K)} \right\}.
		\end{aligned}
		\end{equation}
		The performance of this estimator is given in the following lemma.
		
		\begin{lemma}\label{le:median}(\cite{jerrum1986random,nemirovsky1983problem}, rephrased by~\cite{huang2020predicting})
			For the estimator described by Eq.~\eqref{eq:meanmedian_estimator} where $\hat f^{(r)}$ is identical and independent sample of $f$, if $N=34 \mathrm{Var}(\hat f)/\gamma^2$ for any given $\gamma>0$, then
			\begin{equation}
			\Pr\left(\left|\hat f - \mathbb E \hat f\right|\ge \gamma\right)\le 2\exp(-K/2).
			\end{equation}
			Further, by taking $K = 2\ln(2\delta^{-1})$ for any $\delta>0$, one have
			\begin{equation}
			\Pr\left(\left|\hat f - \mathbb E \hat f\right|\ge \gamma\right)\le \delta.
			\end{equation}
		\end{lemma}
		
		Thanks to this Lemma~and the above discussion, we reach the following theorem which summarize the trade-off between precision and the sample complexity of our main protocol. This theorem is the rigorous version of Theorem~\ref{th:info1} in the main text.
		
		\begin{theorem}\label{th:gl_stat}
			Given $\varepsilon,~\delta>0$, the following number of rounds of calibration in Protocol~\ref{proto:Global}
			\begin{equation}
			R = 136\ln(2\delta^{-1})\cfrac{(1+{\varepsilon}^2)(1+\frac1d)^2}{{\varepsilon}^2(F_{Z}-\frac1d)^2}
			\end{equation}
			is enough for the asymptotic error of the subsequent estimation procedure to satisfy 
			\begin{equation}
			\left| \lbra O \widehat{\cM}^{-1}\widetilde{\cM} \lket\rho-\lbraket{O}{\rho} \right| \le \varepsilon\|O_0\|_\infty,\quad\forall O\in\mr{Herm}(2^n),~\forall\rho\in\cD(2^n).
			\end{equation}
			with a success probability at least $1-\delta$, where 
			$F_{Z}\equiv F_{Z}(\Lambda)$
			is the Z-basis average fidelity of the noise channel $\Lambda$.
		\end{theorem}
				
		\begin{proof}
			Construct the median of means estimator $\hat f$ with $K=2\ln(2\delta^{-1})$ and $N=34\mathrm{Var}(\hat f)/\gamma^2$, where $\gamma = \cfrac{\varepsilon|f|}{1+\varepsilon}$ as Eq.~\eqref{eq:gamma_epsilon} suggests. Use Proposition~\ref{prop:gl_main} to get
			\begin{equation}
			\mathrm{Var}(\hat f)\le \cfrac{2}{(d-1)^2},\qquad |f|= \cfrac{dF_{Z}-1}{d^2-1}.
			\end{equation}
			Then Lemma~\ref{le:median} guarantees
			\begin{equation}
			R=KN=136\ln(2\delta^{-1})\cfrac{(1+{\varepsilon}^2)(d+1)^2}{{\varepsilon}^2(dF_{Z}-1)^2}.
			\end{equation}
		\end{proof}

		Theorem~\ref{th:gl_stat} provides an upper bound on the necessary number of rounds that scales as
		\begin{equation}
		R=O\left(\cfrac{1}{\varepsilon^2(F_{Z}-1/d)^2}\right).
		\end{equation} 
			
			
		
		
		\subsection{Estimation Procedure: Global}
		
		Till now, we have proved the efficiency of the calibration procedure, but have not addressed the efficiency of the estimation procedure. In the noiseless case, the performance of the standard quantum shadow estimation protocol has been characterized in~\cite{huang2020predicting}. Here, we extend their methods to show the performance of the \textbf{\textrm{RShadow}} estimation procedure. 

		For any set of observables $\{O_i\}_{i=1}^M$ and an unknown state $\rho$, the single-round estimation and the final estimation of $o_i\coleq\Tr(O_i\rho)$ is denoted by $\hat o_i^{(r)}$ and $\hat o_i$ respectively, given by Algorithm~\ref{Proto:robust}. The deviation of $\mbb E(\hat o_i^{(r)})$ from $o_i$ has been bounded by Theorem~\ref{th:gl_stat}. Now we want to bound $\mr{Var}(\hat o_i^{(r)})$. We first introduce the following lemma,
		
		\begin{lemma}\label{le:snorm}
			
			For any $O\in\mr{Herm}(2^n)$ and an unknown state $\rho\in\cD(2^n)$, the single-round estimator $\hat o^{(r)}$ given by the \textbf{\textrm{RShadow}} protocol using either ${\sf Cl}(2^n)$ or ${\sf Cl}_2^{\otimes n}$ satisfies
			\begin{equation}
			\mr{Var}(\hat o^{(r)}) \le \|O_0\|^2_{\mr{shadow},\Lambda}
			\end{equation}
			where $O_0\equiv O-\frac{\Tr(O)}{2^n}I$. The function $\|\cdot\|_{\mr{shadow},\Lambda}$ depends on the noise channel and the unitary group being used:
			
			\begin{equation}
			\|O\|_{\mr{shadow},\Lambda}\coleq \max_{\sigma\in\cD(2^n)}\left( \mbb E_{U\sim \mbb{G}} \sum_{b\in\{0,1\}^n}\bra b \Lambda(U\sigma U^\dagger) \ket b\bra b U \widetilde\cM^{-1}(O)U^\dagger\ket b ^2\right)^{1/2}.
			\end{equation} 
		\end{lemma}
		\noindent When $\Lambda = \id$, the function $\|\cdot\|_{\mr{shadow},\Lambda}$ degrades to the norm $\|\cdot\|_\mr{shadow}$ defined in~\cite{huang2020predicting}.

		\begin{proof}
		First observe that the variance of $\hat o^{(r)}$ from Algorithm~\ref{Proto:robust} only depends on the traceless part of $O$:
			\begin{equation}
			\begin{aligned}
			\hat o^{(r)} - \mbb E(\hat o^{(r)}) &=\lbra O \widehat \cM^{-1}{\cU}^\dagger \lket b -  \lbra O \widehat \cM^{-1}\widetilde{\cM} \lket \rho\\&= \lbra {O_0} \widehat \cM^{-1}{\cU}^\dagger \lket b -  \lbra {O_0} \widehat \cM^{-1}\widetilde{\cM} \lket \rho \\&=\hat o_0^{(r)} - \mbb E(\hat o_0^{(r)})
			\end{aligned}
			\end{equation}
			which is because $\widehat \cM$ is diagonal in Pauli transfer matrix representation, and $\widetilde{\cM}$ is a trace-preserving map. Therefore,
			\begin{equation}
			\begin{aligned}
			\mr{Var}(\hat o^{(r)})&=\mbb{E}\left[\left(\hat o^{(r)}-\mbb{E}(\hat o^{(r)})\right)^2\right]\\
			&= \mbb E\left[\left(\lbra {O_0} \widehat \cM^{-1}{\cU}^\dagger \lket b -  \lbra {O_0} \widehat \cM^{-1}\widetilde{\cM} \lket \rho\right)^2\right]\\
			&\le \mbb E~ \lbra {O_0} \widehat \cM^{-1}{\cU}^\dagger \lket b^2\\
			& = \mbb E_{U\sim\mbb G}\sum_{b\in\{0,1\}^n}\lbra b \Lambda {\cU}\lket \rho\lbra b {\cU}\widehat \cM^{-1}\lket {O_0} ^2\\
			&\le \max_{\sigma\in \cD(2^n)}\mbb E_{U\sim\mbb G}\sum_{b\in\{0,1\}^n}\lbra b \Lambda {\cU} \lket \sigma\lbra b {\cU}\widehat \cM^{-1}\lket {O_0} ^2\\
			&= \|O_0\|_{\mr{shadow},\Lambda}^2.
			\end{aligned}
			\end{equation}
		\end{proof}
		
		In the special case that $\mbb G\coleq {\sf Cl}(2^n)$, we can obtain the following bound on the shadow norm $\|\cdot\|_{\mr{shadow},\Lambda}$. 
		\begin{lemma}\label{le:snorm_gl}
			For \textbf{\textrm{RShadow}} using ${\sf Cl}(2^n)$, if the calibration procedure guarantees $\hat f \ge \delta f$ for some $\delta>0$, and we assume $F_{Z}(\Lambda)\ge\frac1d$, then we have
			\begin{equation}
			\|O_0\|_{\mr{shaodw},\Lambda}^2\le\delta^{-2} \left(F_{Z}-\frac1d\right)^{-2}~3\Tr(O_0^2),
			\end{equation}
			for any observable $O$.
		\end{lemma}        
		\begin{proof}
		From the definition of the noisy shadow norm and using the Weingarten functions from Eq.~\eqref{eq:Weingarten} we have
		\begin{equation}
		\begin{aligned}
		\|O_0\|_{\mr{shadow},\Lambda}^2 &= \max_{\sigma\in\cD(2^n)}\mbb E_{U\sim {\sf Cl}(2^n)}\sum_{b\in\{0,1\}^n}\hat f^{-2}\Tr\left[(U\sigma U^\dagger\otimes U O_0 U^\dagger\otimes U O_0 U^\dagger)(\Lambda^\dagger(\ket b\bra b)\otimes\ket b\bra b\otimes\ket b\bra b)\right]\\
		&= \max_{\sigma\in\cD(2^n)}\sum_{b\in\{0,1\}^n}\hat f^{-2}\Tr\left[ \Phi^{(3)}_\mr{Haar} (\sigma\otimes O_0\otimes O_0)(\Lambda^\dagger(\ket b\bra b)\otimes\ket b\bra b\otimes\ket b\bra b)\right]\\
		&= \max_{\sigma\in\cD(2^n)}\sum_{b\in\{0,1\}^n}\hat f^{-2}\sum_{\pi,\xi\in S_3}c_{\pi,\xi}\Tr\left[W_\pi(\sigma\otimes O_0\otimes O_0)\right]\Tr\left[ W_\xi (\Lambda^\dagger(\ket b\bra b)\otimes\ket b\bra b\otimes\ket b\bra b)\right]\,,
		\end{aligned}
		\end{equation}
		where in the last equation we use the Weingarten function to expand the Haar intergral (see Eq.~\eqref{eq:Weingarten}). Now using Eq.~\eqref{eq:perm_tr} to compute the traces appearing above, we have
		\begin{equation}
		    \Tr\left(\vec W(\sigma\otimes O_0\otimes O_0)\right)= \begin{bmatrix}
		0,&0,&0,&\Tr(O_0^2),&\Tr(\sigma O_0^2),&\Tr(\sigma O_0^2)
		\end{bmatrix}.
		\end{equation}
        Recall that $F_{Z}(\Lambda)$ is the Z-basis average fidelity of $\Lambda$ as defined in Prop.~\ref{prop:gl_main}, and we denote it simply as $F_{Z}$ in the following. 
		Then we also have
		\begin{equation}
		\begin{aligned}
        &\sum_{b\in\{0,1\}^n}\Tr\left(\vec W(\Lambda^\dagger(\ket b\bra b)\otimes\ket b\bra b\otimes\ket b\bra b)\right)\\=& \sum_{b\in\{0,1\}^n}\left[\begin{matrix}
		\Tr(\Lambda^\dagger(\ket b\bra b)),& \bra b\Lambda^\dagger(\ket b\bra b)\ket b,&\bra b\Lambda^\dagger(\ket b\bra b)\ket b,&\Tr(\Lambda^\dagger(\ket b\bra b)),&\bra b\Lambda^\dagger(\ket b\bra b)\ket b,&\bra b\Lambda^\dagger(\ket b\bra b)\ket b
		\end{matrix}\right]\\=& ~d* \begin{bmatrix}
		1,&F_{Z}(\Lambda),&F_{Z}(\Lambda),&1,&F_{Z}(\Lambda),&F_{Z}(\Lambda)
		\end{bmatrix}.
		\end{aligned}
		\end{equation}
		Again, $\vec W$ is a vectorization of $S_3$ defined in Eq.~\eqref{eq:perm_vec}, just for the simplicity of notation. 
		Inserting the above two equations and the value of the Weingarten matrix from Eq.~\eqref{eq:Wein3},
		\begin{equation}
		\begin{aligned}
		\|O_0\|_{\mr{shadow},\Lambda}^2 
		&= \max_{\sigma\in\cD(2^n)} \hat f^{-2}\frac{\Tr(O_0^2)(d-2F_{Z}+1)+2\Tr(\sigma O_0^2)(dF_{Z}-1)}{(d+2)(d^2-1)}\\
		& \le \hat {f}^{-2}~ \frac{2dF_{Z}+d-2F_{Z}-1}{(d+2)(d^2-1)} \Tr(O_0^2)\\
		& = \frac{f^2}{\hat{f}^2} \left(\frac{d^2-1}{dF_{Z} -1}\right)^2\frac{2dF_{Z}+d-2F_{Z}-1}{(d+2)(d^2-1)} \Tr(O_0^2) \\
		&\le \frac{f^2}{\hat{f}^2} ~  \left(F_{Z}-\frac1d\right)^{-2}~3 \Tr(O_0^2)
		\end{aligned}
		\end{equation}
		where the first inequality is by the fact that $\Tr(\sigma O^2_0)\le\|O^2_0\|_\infty\le\Tr(O_0^2)$ and the assumption $F_{Z} \geq \frac{1}{d}$, and in the second equality we use the expression of $f$ from Proposition~\ref{prop:gl_main}. 
        \end{proof}

		Compared to Proposition~1 from~\cite{huang2020predicting} which states that $\|O_0\|^2_\mr{shadow}\le 3\Tr(O_0^2)$, we conclude the following. 
		As long as the noise channel $\Lambda$ has a Z-basis fidelity that is not too low and the noise calibration procedure is conducted with sufficiently many rounds, then the estimation procedure of our \textbf{\textrm{RShadow}} protocol using ${\sf Cl}(2^n)$ is as efficient as the noiseless standard quantum shadow estimation protocol~\cite{huang2020predicting} up to a small multiplicative factor. 
		That is to say, the expectation value of any observable $O$ that has small Hilbert-Schmidt norm can be efficiently estimated by \textbf{\textrm{RShadow}}.
		
		To complete the discussion, we give the following theorem as a rigorous version of Theorem~\ref{th:gl_all} in the main text.
		
		\begin{theorem}\label{th:gl_allr}
		For \textbf{\textrm{RShadow}} with ${\sf Cl}(2^n)$, given the noise channel satisfies $F_{Z}(\Lambda)\ge\frac1d$, if the number of calibration samples $R_C$ and the number of estimation samples $R_E$ satisfies
		\begin{equation}
		\begin{aligned}
		R_C & = 136\ln(2\delta_1^{-1})\cfrac{(1+{\varepsilon_1}^2)(1+\frac1d)^2}{{\varepsilon_1}^2( F_{Z}-\frac1d)^2},\\
		R_E &=	\frac{204}{\varepsilon_2^2}\ln(2M/\delta_2) (1+\varepsilon_1)^{2}(F_{Z}-\frac1d)^{-2},
		\end{aligned}
		\end{equation}
		respectively, then the protocol can estimate $M$ arbitrary linear functions $\Tr(O_1\rho),...\Tr(O_M\rho)$ such that $\Tr(O_i^2)\le 1$, up to accuracy $\varepsilon_1+\varepsilon_2$ with success probability at least $1-\delta_1-\delta_2$. 
	\end{theorem}
		
	\begin{proof}
    First, according to Theorem~\ref{th:gl_stat}, for the given number of samples $R_C$ one have
	\begin{equation}\label{eq:eps1_gl}
	\left| \mbb E(\hat o_i^{(r)}) - \Tr(O_i\rho) \right| \le \varepsilon_1.
	\end{equation}
	Meanwhile, from the proof of Theorem~\ref{th:gl_stat} (see Eq.~\eqref{eq:gamma_epsilon}), one also have
	\begin{equation}
	|\hat f^{-1}f-1|\le \varepsilon_1 \Rightarrow \hat f\ge (1+\varepsilon_1)^{-1} f.
	\end{equation}
	Both of the above equations hold simultaneously with probability at least $1-\delta_1$. 
	
	\smallskip
	
	\noindent Now, by Lemma~\ref{le:snorm} and Lemma~\ref{le:snorm_gl}, the single-round estimators in the estimation procedure satisfy:
	\begin{equation}
	\mr{Var}(\hat o_i^{(r)}) \le 3(1+\varepsilon_1)^{2}(F_{Z}-\frac1d)^{-2}.
	\end{equation}
	So we set the median of mean estimators $\hat o_i$ of the estimation procedure with the following parameters:
	\begin{equation}
	N=\frac{34}{\varepsilon_2^2}\cdot 3(1+\varepsilon_1)^{2}(F_{Z}-\frac1d)^{-2},\quad K=2\ln(2M/\delta_2).
	\end{equation}
	Then Lemma~\ref{le:median} combined with the union bound gives that the following holds for all $i$ with probability at least $1-\delta_2$:
	\begin{equation}\label{eq:eps2_gl}
	\left| \hat o_i - \mbb E(\hat o_i^{(r)}) \right| \le \varepsilon_2,
	\end{equation}
    Combining Eq.~\eqref{eq:eps1_gl} and Eq.~\eqref{eq:eps2_gl} using the triangular inequality gives
	\begin{equation}
	|\hat o_i - \Tr(O_i\rho)|\le \varepsilon_1+\varepsilon_2,
	\end{equation}
	which holds with probability at least $1-\delta_1-\delta_2$.
	This completes the proof.
	\end{proof}

\section{Sample Complexity of \textbf{\textrm{RShadow}} with Local Clifford Group} \label{sec:app_lc}
		
The result in App.~\ref{sec:app_gl} is based on the $n$-qubit Clifford group, which is challenging to implement in experiment. In this section, we analyze the protocol using $n$-qubit local Clifford group, denoted as ${\sf Cl}_2^{\otimes n}$, which is the $n$-fold direct product of the single-qubit Clifford group. Such unitaries are all single-qubit operations, thus much easier to implement in the experiment.
		
\subsection{Calibration Procedure: Local}
		
Being twirled by the local Clifford group, the channel $\widetilde{\cM}$ becomes a Pauli channel that is symmetric among the $X, Y, Z$ index, whose Pauli-Liouville representation is~\cite{gambetta2012benchmarking}
		\begin{equation}
		\widetilde{\cM} = \mathop{\mathbb{E}}_{U \sim {\sf Cl}_2^{\otimes n}} \cU^\dagger M_z\Lambda \cU = \sum_{z\in\{0,1\}^n}f_z \Pi_z, 
		\end{equation}
		where $\Pi_z = \bigotimes_{i=1}^n \Pi_{z_i}$, 
		\begin{equation}
		\Pi_{z_i} = 
		\begin{cases}
		&\lketbra{\sigma_0}{\sigma_0},\quad z_i=0, \\ &I-\lketbra{\sigma_0}{\sigma_0},\quad z_i=1,
		\end{cases}
		\end{equation}
		and $f_z$ is the Pauli fidelity. In the noiseless case, one can obtain $f_z = 3^{-|z|}$ where $|z|$ is the number of $1$ in $z$.
		
		\medskip
		
		\noindent\textbf{Notation:} For any string $m\in\{0,1\}^n$ we define $\lket m$ to be the Liouville representation of the computational basis state $\ket m$, while $\lket{\sigma_m}$ stands for the normalized Pauli operator corresponding to $P_m\coleq\bigotimes_{i=1}^{n}P_Z^{m_i}$. On the other hand, the notation of $z$ in this section consistently stands for an $n$-bit string and should not be confused with the Pauli-Z index. 
		
		\medskip
		
		The \textbf{\textrm{RShadow}} protocol using local Clifford group can be written as follows.
		
		\begin{protocol}\label{proto:local} [\textbf{\textrm{RShadow}} with $Cl_{2}^{\otimes n}$] 
			\begin{enumerate}
				\item Prepare $\ket{\bvec 0}\equiv\ket 0^{\otimes n}$. Sample $U$ uniformly form ${\sf Cl}_2^{\otimes n}$ and apply it to $\ket {\bvec 0}$.
				\item Measure the above state in the computational basis. Denote the outcome state vector as $\ket b$.
				\item Calculate the single-round Pauli fidelity estimator 
				$\hat{f}^{(r)}_z = \lbra{b}{\cU}\lket{P_z}$
				for all $z\in\{0,1\}^n$.
				
				\item Repeat step 1--3 for $R=NK$ rounds. Then the final estimation of $f_z$ is given by a median of means estimator $\hat f_z$ constructed from the single round estimators $\{\hat f_z^{(r)}\}_{r=1}^R$ with parameter $N,~K$ (see Eq.~\eqref{eq:meanmedian_estimator}).
				
				\item After the above steps, apply the standard shadow estimation protocol of~\cite{huang2020predicting} on $\rho$, with the inverse channel $\widetilde{\mc M}^{-1}$ replaced by 
				\begin{equation}\label{eq:local_M_expression}
				\widehat{\cM}^{-1} = \sum_{z\in\{0,1\}^n}\hat{f}^{-1}_z \Pi_z, 
				\end{equation}
			\end{enumerate}
		\end{protocol}

		Of course, it is unaffordable in classical computational resource to compute all $\hat f_z^{(r)}$ in a single round. In practice, we only need to compute those $f_z$ of interest. For example, if we only want to predict $k$-local properties, then only $\hat f_z^{(r)}$ such that $|z|\le k$ need to be computed. If we are only interested in nearby qubits, then the number of necessary $\hat f_z^{(r)}$ can be further reduced.
		
		\medskip
		
		Now we show that the single-round estimators $\{\hat{f}^{(r)}_z\}$ are unbiased and  the variance of them are bounded.
		
		
		\begin{proposition}\label{prop:lc_var}
			The single-round Pauli fidelity estimator $\hat{f}^{(r)}_z$ satisfies
			\begin{equation}
			\mathbb{E}(\hat f^{(r)}_z) =f_z =3^{-|z|}~\Gamma_\Lambda(z),\qquad \mathrm{Var}(\hat f^{(r)}_z) \le 3^{-|z|}.\label{eq:lc_mean_estimator}
			\end{equation} 
			where $\Gamma_\Lambda(z)\coleq \cfrac1{2^n}\sum_{x,b\in\{0,1\}^n}(-1)^{z\cdot(x\oplus b)}\lbra b\Lambda\lket x$.
		\end{proposition}
		
				\begin{proof}[Proof of Proposition~\ref{prop:lc_var}]
			
			To begin with, we show that $\hat f^{(r)}_z$ is an unbiased estimator of $f_z$. 
			From the definition of $\hat f^{(r)}_z$ in Protocol~\ref{proto:local} above, we have that the expectation value over the experiments is given by
			\begin{equation}
			\begin{aligned}
			\mathbb{E}(\hat f^{(r)}_z)
			&= \mbb E_{U\sim {\sf Cl}_2^{\otimes n}}\sum_b \lbra{P_z}{\cU}^\dagger\lketbra{b}{b}\Lambda {\cU} \lket {\bvec 0} \\ 
			&= \lbra{P_z}\tilde M \lket {\bvec 0} \\
			&= f_z\lbraket{P_z}{\bvec 0}\\
			&= f_z.
			\end{aligned}
			\end{equation}
			
			\medskip
			\noindent To derive the expression for $\hat f_z$ that depends on the noise channel $\Lambda$, we can alternatively expand the expectation as follows,
			\begin{equation}
			\begin{aligned}
			\mbb E(\hat f^{(r)}_z) &= \mbb E_{U\sim {\sf Cl}_2^{\otimes n}} \sum_b  \bra b \Lambda\left( U\ketbra{\bvec 0}{\bvec 0}U^\dagger  \right)\ket b\tr\left[ U^\dagger\ket b\bra b UP_z\right]\\
			&= \sum_b \tr\left[\mbb E_{U\sim {\sf Cl}_2^{\otimes n}} \left(U\ket {\bvec 0}\bra {\bvec 0} U^\dagger\otimes UP_z U^\dagger\right)\left(\Lambda^\dagger(\ket b\bra b)\otimes\ket b\bra b\right)   \right].
			\end{aligned}
			\end{equation}
			To evaluate this expression, we first consider the single-qubit case. 
			By direct calculation we obtain
			\begin{equation}
			\begin{aligned}
			\mbb E_{U\sim {\sf Cl}_2} \left(U\ket 0\bra 0 U^\dagger \otimes UP_IU^\dagger\right) &= \cfrac12 I,\\
			\mbb E_{U\sim {\sf Cl}_2} \left(U\ket 0\bra 0 U^\dagger \otimes UP_ZU^\dagger\right) &= \cfrac23 P_{\text{sym}^2} - \frac12 I.
			\end{aligned}
			\end{equation}
			Hence, for any $X\in \mathrm{Herm}(2)$ and $b\in\{0,1\}$, by Lemma~\ref{le:sym},
			\begin{equation}
			\begin{aligned}
			\tr\left[ \mbb E_{U\sim {\sf Cl}_2} \left(U\ket 0\bra 0 U^\dagger \otimes UP_IU^\dagger\right)\left(X\otimes\ket b\bra b\right)\right] &= \frac12(\bra b X\ket b+\bra {b\oplus 1}X\ket {b\oplus 1})\\
			\tr\left[ \mbb E_{U\sim {\sf Cl}_2} \left(U\ket 0\bra 0 U^\dagger \otimes UP_ZU^\dagger\right)\left(X\otimes\ket b\bra b\right)\right] &= \frac16(\bra b X\ket b - \bra {b\oplus 1}X\ket {b\oplus 1}).
			\end{aligned}
			\end{equation}
			Applying this to the $n$-qubit case, one can then verify that
			\begin{equation}
			\begin{aligned}
			\mbb E(\hat f^{(r)}_z) = \frac1{3^{|z|}}\frac{1}{2^n}\sum_{x,b}(-1)^{z\cdot(x\oplus b)} \bra x \Lambda^\dagger(\ketbra b b)\ket x = \frac1{3^{|z|}}\Gamma_\Lambda(z).
			\end{aligned}
			\end{equation}

			%

			\medskip
			
			\noindent To compute the variance, we compute
			\begin{equation}\label{eq:local_2ndmoment}
			\begin{aligned}
			\mbb E(\hat f^{(r)^2}_z) &= \mbb E_{U\sim {\sf Cl}_2^{\otimes n}} \sum_b \bra b \Lambda\left( U\ketbra{\bvec 0}{\bvec 0}U^\dagger  \right)\ket b 
			\tr\left[U^\dagger\ket b\bra b UP_z\right]^2\\
			&= \sum_b\tr\left[\mbb E_{U\sim {\sf Cl}_2^{\otimes n}} \left(U\ket {\bvec 0}\bra {\bvec 0} U^\dagger\otimes UP_z U^\dagger \otimes UP_z U^\dagger\right)\left(\Lambda^\dagger(\ket b\bra b)\otimes\ket b\bra b\otimes\ket b\bra b\right) \right].
			\end{aligned}
			\end{equation}
			Again, first consider the single-qubit case. One can verify that
			\begin{equation}
			\begin{aligned}
			\mbb E_{U\sim {\sf Cl}_2} \left(U\ket 0\bra 0 U^\dagger \otimes UP_I U^\dagger \otimes UP_I U^\dagger\right) &= \frac12 I,\\
			\mbb E_{U\sim {\sf Cl}_2} \left(U\ket 0\bra 0 U^\dagger \otimes UP_Z U^\dagger \otimes UP_Z U^\dagger\right) &= \frac12 P_{\text{sym}^3} + \frac13\left( P_{\text{sym}^2}^{(2,3)}-P_{\text{sym}^2}^{(1,2)}-P_{\text{sym}^2}^{(1,3)} \right).
			\end{aligned}
			\end{equation}
			Hence, for any $X\in \mathrm{Herm}(2)$ and $b\in\{0,1\}$, by Lemma~\ref{le:sym},
			\begin{equation}
			\begin{aligned}
			\Tr\left[ \mbb E_{U\sim {\sf Cl}_2} \left(U\ket 0\bra 0 U^\dagger \otimes UP_I U^\dagger \otimes UP_I U^\dagger\right)(X\otimes\ketbra b b\otimes\ketbra b b) \right] = \frac12\tr(X),\\
			\Tr\left[ \mbb E_{U\sim {\sf Cl}_2} \left(U\ket 0\bra 0 U^\dagger \otimes UP_Z U^\dagger \otimes UP_Z U^\dagger\right)(X\otimes\ketbra b b\otimes\ketbra b b) \right] = \frac16\tr(X).\\
			\end{aligned}
			\end{equation}
			One can also verify these equations using the Weingarten matrix. Applying to the $n$-qubit case, one can verify that
			\begin{equation}
			\begin{aligned}
			\mbb E(\hat f^{(r)^2}_z) &= \frac1{2^n}\frac1{3^{|z|}}\sum_b\tr(\Lambda^\dagger(\ketbra b b))\\
			&=\frac1{2^n}\frac1{3^{|z|}}\sum_{x,b}\bra b\Lambda(\ketbra x x)\ket b\\
			&= \frac1{3^{|z|}}.
			\end{aligned}
			\end{equation}
			Since $\mbb E(\hat f_z^{(r)^2})$ serves as an upper bound of $\mathrm{Var}(\hat f_z^{(r)})$, this completes the proof of Proposition~\ref{prop:lc_var}.
		\end{proof}
		
		\medskip
	
		Based on Proposition~\ref{prop:lc_var}, we can now bound the sample complexity of Protocol~\ref{proto:local}. Firstly, we set the median of mean estimator $\hat{f}_z$ according to Lemma~\ref{le:median} as
		\begin{align}
		\bar{f}_{z}^{(t)} &\coleq \cfrac1N\sum_{r=(t-1)N+1}^{tN}\hat{f}_z^{(r)}, \quad t=1,2,...,K,\label{eq:lc_mean} \\
		\hat{f}_z&\coleq\textrm{median}\left\{\bar{f}_{z}^{(1)},\bar{f}_{z}^{(2)},...,\bar{f}_{z}^{(K)}\right\},\label{eq:lc_median}
		\end{align}
		with $N$ and $K$ to be specified. The following theorem gives the performance of Protocol~\ref{proto:local}. 
		
		\begin{theorem}\label{th:lc_stat}
			Given $\varepsilon,~\delta>0$, the number of qubits $n\ge 2$, and an integer $k\le n$,
			the following number of samples for the calibration procedure
			\begin{equation}
			R = \mc O\left( \cfrac{3^{k}(k\ln n+\ln\delta^{-1})}{\varepsilon^{2}\min_{|z|\le k}\Gamma_\Lambda^{2}(z)} \right)
			\end{equation}
			is enough for the subsequent shadow estimation procedure to estimate any $k$-local observable for any state to the following precision
			\begin{equation}
			\left| \lbra O \widehat{\cM}^{-1}\widetilde{\cM} \lket\rho-\lbraket{O}{\rho} \right| \le \varepsilon 2^k \|O\|_\infty,\quad\forall~\text{k-local}~O\in\mr{Herm}(2^n),~\forall\rho\in\cD(2^n).
			\end{equation}
			with a success probability at least $1-\delta$. 
		\end{theorem}
		
		
		Here, An operator $O$ is called $k$-local if it only non-trivially acts on a $k$-qubit subspace, \textit{i.e.} $O = \tilde O_{S} \otimes I_{[n]\backslash S}$ for some index set $S\subset [n]$ and $|S| = k$.

		\begin{proof}
			We first notice that
			\begin{equation}\label{eq:derivation_hilbert}
			\begin{aligned}
			\left| \lbra O \widehat{\cM}^{-1}\widetilde{\cM} \lket\rho-\lbraket{O}{\rho} \right| &= \left|\sum_{a\in\mbb Z^{2n}_2}(\hat f_{z(a)}^{-1}f_{z(a)}-1)\lbraket{O}{\sigma_a}\lbraket{\sigma_a}{\rho}\right|\\
			&\le \max_{|z|\le k}\left|\hat f^{-1}_zf_z-1\right|\cdot \sum_{a\in\mbb Z_2^{2n}}|\lbraket{O}{\sigma_a}|\cdot|\lbraket{\sigma_a}{\rho}|\\
			&\le \max_{|z|\le k}\left|\hat f^{-1}_zf_z-1\right|\cdot \sum_{a\in\mbb Z_2^{2n}}\frac{1}{2^n}|\lbraket{O}{P_a}|
			\end{aligned}
			\end{equation}
			where the first equality is by expanding the Pauli transfer basis and we define the mapping $z$ as
			\begin{equation}
			z: \mbb Z_2^{2n}\to \{0,1\}^n,~z(p)_i =\left\{\begin{aligned}
				&0,\quad (P_p)_i=I,\\
				&1,\quad (P_p)_i\ne I,
				\end{aligned} \right.
			\end{equation}
			and the first inequality uses the fact that $O$ is $k$-local. Now we bound the second factor of the above equation. Without loss of generality, suppose $O$ acts non-trivially on the first $k$ qubits: 
			$O=\tilde O \otimes I_{2^{n-k}}$, and that $\tilde O$ can be decomposed as
			\begin{equation}
			\tilde O = \sum_{\tilde a\in \mbb Z_2^{2k}} \alpha_{\tilde a} P_{\tilde a}.
			\end{equation}
			Then we naturally have
			\begin{equation}
			O = \tilde O \otimes I_{2^{n-k}} = \sum_{\tilde a\in \mbb Z_2^{2k}} \alpha_{\tilde a} P_{\tilde a}\otimes P_I^{\otimes(n-k)}.
			\end{equation}
			So,
			\begin{equation}
			\sum_{a\in\mbb Z_2^{2n}}\frac{1}{2^n}|\lbraket{O}{P_a}| = \sum_{\tilde a\in\mbb Z_2^{2k}} |\alpha_{\tilde a}| \le \sqrt{4^k}\sqrt{\sum_{\tilde a\in\mbb Z_2^{2k}}{\alpha}^2_{\tilde a}} = 2^k \sqrt{\frac{\Tr(\tilde O^2)}{2^k}} \le 2^k \|\tilde O\|_\infty = 2^k \|O\|_\infty,
			\end{equation}
			where the first inequality is by Cauchy-Schwarz inequality. Combining the above results, we have
			\begin{equation}\label{eq:c23}
			\left| \lbra O \widehat{\cM}^{-1}\widetilde{\cM} \lket\rho-\lbraket{O}{\rho} \right| \le \max_{|z|\le k}\left|\hat f^{-1}_zf_z-1\right|\cdot 2^k \|O\|_\infty 
			\end{equation}
			
%
%
			
			\medskip
			
			\noindent For any $z\in\{0,1\}^n$, suppose $|\tilde{f_z}-f_z|\le \gamma_z$, and then we have
			\begin{equation}
			\left| 1- \hat{f}_z^{-1}f_z \right|\le \cfrac{|\hat f_z - f_z|}{|\hat f_z|} \le \cfrac{\gamma_z}{|f_z|-\gamma_z}.
			\end{equation}
			By setting $\gamma_z = \cfrac{\varepsilon}{1+\varepsilon}|f_z|$, the above equation is upper bounded by $\varepsilon$.
			Therefore, if we set $$N = 34\mathrm{Var}(\hat f_z)/\gamma_z^2,\quad K = 2\ln(2\delta^{-1})$$ for the median of mean estimator in Eq.~\eqref{eq:lc_mean} and Eq.~\eqref{eq:lc_median}, by Lemma~\ref{le:median} we have $|1-\hat f_z^{-1}f_z|\le \varepsilon$ with a success probability at least $1-\delta$. Now we want all $z\in\{0,1\}^n$ such that $|z|\le k$ to statisfy this inequality. The number of such strings is no larger than $n^k$, so we set
			\begin{align}
			N &= \max_{|z|\le k} 34\mathrm{Var}(\hat f_z)/\gamma_z^2 \le 34\cdot 3^{k}\cfrac{(1+\varepsilon)^2}{\varepsilon^2}\max_{|z|\le k}\Gamma_\Lambda^{-2}(z), \\
			K &= 2\ln(2(\delta/n^k)^{-1}),
			\end{align}
			and apply the union bound. Now we have $|1-\hat f_z^{-1}f_z|\le\varepsilon$ for all $|z|\le k$ with probability at least $1-\delta$. Our final upper bound of the sample complexity is
			\begin{equation}
			R=NK\le 68\cdot 3^{k}\cfrac{(1+\varepsilon)^2}{\varepsilon^2}\left(k\ln n+\ln2\delta^{-1}\right)\max_{|z|\le k}\Gamma_\Lambda^{-2}(z),
			\end{equation}
			which completes the proof.
		\end{proof}
		
		The quantity $\Gamma_\Lambda(z)$ can be lower bounded when $\Lambda$ is close to an identity channel, as shown by the following lemma. 
		
		\begin{lemma}\label{le:lc_weak}
			if the Z-basis average fidelity of $\Lambda$ satisfies $F_{Z}(\Lambda)\ge 1-c$ for some $0\le c\le 1$, then $\Gamma_\Lambda(z)\ge 1-2c$ for all $z\in\{0,1\}^n$.
		\end{lemma}
		\begin{proof}
			\begin{equation}
			\begin{aligned}
			\Gamma_\Lambda(z) &= \cfrac1{2^n}\sum_{x,\delta\in\{0,1\}^n}(-1)^{z\cdot\delta}\lbra {x\oplus\delta}\Lambda\lket x \\
			&\ge\cfrac1{2^n}\sum_{x\in\{0,1\}^n}\left(\lbra x \Lambda\lket x - \sum_{\delta\in\{0,1\}^n,|\delta|\ne 0}\lbra {x\oplus\delta}\Lambda\lket x\right)\\
			&= \cfrac1{2^n}\sum_{x\in\{0,1\}^n} (2\lbra x\Lambda\lket x - 1)\\
			&= 2F_{Z}(\Lambda) - 1\\
			&\ge 1-2c.
			\end{aligned}
			\end{equation}
			where the second equality is by the fact that $\Lambda$ is trace-preserving, and hence $\sum_{b\in\{0,1\}^n}\lbra b\Lambda\lket x = 1$.
		\end{proof}

		Specifically, if we substitute the bound for $\Gamma_\Lambda(z)$ from Lemma~\ref{le:lc_weak} into the above theorem, we get Theorem~\ref{th:info2} in the main text. We conclude that our Protocol~\ref{proto:local} can mitigate the noise in the computation of the expectation of any $k$-local observable efficiently, given that $k$ is small and the noise is weak.
		
		\medskip
		
		
        \subsection{Estimation Procedure: Local}
		
		Now we consider the \textbf{\textrm{RShadow}} estimation procedure using ${\sf Cl}_2^{\otimes n}$. 
		Thanks to Lemma~\ref{le:snorm}, we only need to characterize $\|\cdot\|_{\mr{shadow},\Lambda}^2$.
		Due to technical difficulties, we are currently not able to bound $\|\cdot\|^2_{\mr{shadow},\Lambda}$ for the most general noise channel $\Lambda$, but we do have results for local noise channel $\Lambda$ (hence also for any separable $\Lambda$ by linearity). Suppose
		$\Lambda \equiv \bigotimes_{i=1}^n \Lambda_i$,
		and denote the Z-basis fidelity of the qubit channels $\Lambda_i$ as $F_{{Z},i}$.  
		Further assume $O$ is $k$-local, which means it is non-trivially supported on only $k$ qubits.
		We have
		\begin{equation}\label{eq:local_shadow_norm}
		\begin{aligned}
		\|O\|^2_{\mr{shadow},\Lambda} &= \max_{\sigma\in\cD(2^n)}\mbb E_{U\sim {\sf Cl}_2^{\otimes n}}\sum_{b\in\{0,1\}^n}\Tr\left[\left(\sigma\otimes\widehat\cM^{-1}(O)\otimes\widehat\cM^{-1}(O)\right)U^{\dagger\otimes 3}\left(\Lambda^\dagger(\ket b\bra b)\otimes\ket b\bra b\otimes\ket b\bra b\right)U^{\otimes 3}\right]\\
		\end{aligned}
		\end{equation}
		Consider the single-qubit case, one have
		\begin{equation}
		\begin{aligned}\label{eq:local_Phi}
		\Phi_i\coleq &\mbb E_{U\sim {\sf Cl}_2}\sum_{b={0,1}}  U^{\dagger\otimes 3}\left(\Lambda_i^\dagger(\ket b\bra b)\otimes\ket b\bra b\otimes\ket b\bra b\right)U^{\otimes 3}\\
		=&\sum_{b=0,1} \Phi^{(3)}_{\mr{Haar}}\left(\Lambda_i^\dagger(\ket b\bra b)\otimes\ket b\bra b\otimes \ket b\bra b\right)\\
		=&\sum_{b=0,1} \sum_{\pi,\xi\in S_3}c_{\pi,\xi}W_\pi\Tr\left(W_\xi(\Lambda_i^\dagger(\ket b\bra b)\otimes\ket b\bra b\otimes \ket b\bra b)\right)\\
		=& \frac1{12}\left[ (3-2F_{{Z},i})(W_{()}+W_{(2,3)}) + (2F_{{Z},i}-1)(W_{(1,2)}+W_{(1,3)}+W_{(1,2,3)}+W_{(1,3,2)})   \right]
		\end{aligned}
		\end{equation}
		where we use the Weingarten function to expand the Haar integral, see Eq.~\eqref{eq:Weingarten}, and the value of the Weingarten matrix is from Eq.~\eqref{eq:Wein32}.
		
		For any $X\in \mr{Herm}(2^n)$ and single-qubit Pauli operators $P_p, P_q$, we want to calculate the following quantity $\Tr\left[ (X\otimes P_p\otimes P_q) \Phi_i \right]$. By direct calculation using Eq.~\eqref{eq:local_Phi}, one can verify that there are following four different cases
		\begin{equation}\label{eq:local_1qubit}
		\Tr\left[ (X\otimes P_p\otimes P_q) \Phi_i \right] = \Tr(XP_pP_q)\cdot\left\{\begin{aligned}
		1,&\quad P_p=P_q=I,\\
		\frac13,&\quad P_p=P_q\ne I,\\
		\frac{2F_{{Z},i}-1}{3},&\quad (P_p=I, P_q\ne I)\text{ or }(P_p\ne I, P_q= I),\\
		0,& \quad \text{otherwise}. 
		\end{aligned}\right.
		\end{equation}
		This indicates that, the value $\Tr\left[ (X\otimes P_p\otimes P_q) \Phi_i \right]$ is non-zero if and only if the two single-qubit Pauli operators $P_p$ and $P_q$ commute. 
		
		Now we return to the evaluation of Eq.~\eqref{eq:local_shadow_norm}. Our strategy is similar to~\cite{huang2020predicting}. We first decompose $O$ into the Pauli operator basis (Note that, we use un-normalized Pauli operators here)
		\begin{equation}
		O \equiv \sum_{p\in\mbb Z_2^{2n}}\alpha_p P_p,\quad \text{for~}\alpha_p\in\mbb R.
		\end{equation}
		Since $O$ is $k$-local, one have $\alpha_p = 0$ for all $|p|> k$, where for any $p\in \mbb Z_2^{2n}$ we denote the Pauli weight of $P_p$ as $|p|$. 
		Also recall from Eq.~\eqref{eq:local_M_expression} that
		\begin{equation}
		\widehat \cM = \sum_{p\in\mbb Z_2^{2n}}\hat f_{z(p)}\lket{\sigma_p}\lbra{\sigma_p},
		\end{equation}
		where we define $z$ as the following mapping
		\begin{equation}
		z: \mbb Z_2^{2n}\to \{0,1\}^n,~z(p)_i = 0 \text{~iff~} (P_p)_i=I,\quad\forall i\in[n].
		\end{equation}
		The intuition is that after twirling over the local Clifford group the Pauli X, Y, Z indexes are symmetrized.
		
		\smallskip
		
		Now we can calculate Eq.~\eqref{eq:local_shadow_norm} as follows
		\begin{equation}\label{eq:local_shadow_norm2}
		\begin{aligned}
		\|O\|^2_{\mr{shadow},\Lambda} &= \max_{\sigma\in\cD(2^n)} \sum_{p,q\in\mbb Z_2^{2n}} \hat f_{z(p)}^{-1}\hat f_{z(q)}^{-1}\alpha_p\alpha_q \Tr\left[(\sigma\otimes P_p\otimes P_q)(\otimes_{i=1}^n\Phi_i)\right]\\
		&= \max_{\sigma\in\cD(2^n)} \sum_{p,q\in\mbb Z_2^{2n}} \hat f_{z(p)}^{-1}\hat f_{z(q)}^{-1}\alpha_p\alpha_q \delta(p,q) \Tr(\sigma P_p P_q)\frac{ \prod_{i\in[n]:(P_{p,i}=I, P_{q,i}\ne I)\lor (P_{p,i}\ne I, P_{q,i}= I)} (2F_{{Z},i}-1)}{3^{|p\lor q|}}\\
		&= \left\|\sum_{p,q\in\mbb Z_2^{2n}} \hat f_{z(p)}^{-1}\hat f_{z(q)}^{-1}\alpha_p\alpha_q \delta(p,q) P_p P_q\frac{ \prod_{i\in[n]:(P_{p,i}=I, P_{q,i}\ne I)\lor (P_{p,i}\ne I, P_{q,i}= I)} (2F_{{Z},i}-1)}{3^{|p\lor q|}}\right\|_\infty\\
		&\le \sum_{p,q\in\mbb Z_{2}^{2n}} \left| \hat f_{z(p)}^{-1}\hat f_{z(q)}^{-1}\alpha_p\alpha_q \delta(p,q)\frac{ \prod_{i\in[n]:(P_{p,i}=I, P_{q,i}\ne I)\lor (P_{p,i}\ne I, P_{q,i}= I)} (2F_{{Z},i}-1)}{3^{|p\lor q|}}\right|\\
		&\le \sum_{p,q\in\mbb Z_{2}^{2n}} \delta(p,q) 3^{|p\land q|}|\alpha_p||\alpha_q|  \frac{|\hat f_{z(p)}^{-1}\hat f_{z(q)}^{-1}|}{3^{|p|}3^{|q|}}\\
		&\le \left(\sum_{p,q\in\mbb Z_{2}^{2n}} \delta(p,q) 3^{|p\land q|}|\alpha_p||\alpha_q|\right)\cdot\left(\max_{z\in\{0,1\}^n:|z|\le k}\frac{\hat f_z^{-2}}{3^{2|z|}}\right).
		\end{aligned}
		\end{equation}
		Here, for the second equality, we apply the single-qubit result from Eq.~\eqref{eq:local_1qubit}. The functional $\delta(p,q)$ equals to $1$ if $P_{p_i}$ commutes with $P_{q_i}$ for all $i\in[n]$ and equals to $0$ otherwise, and we have the following definitions
		\begin{equation}
		\begin{aligned}
		|p\lor q|\coleq & ~\#\{i\in[n]:P_{p,i}\ne I \text{ or } P_{q,i}\ne I\}.\\
		|p\land q|\coleq &~\#\{i\in[n]:P_{p,i}\ne I \text{ and } P_{q,i}\ne I\}.
		\end{aligned}
		\end{equation} 
		The third equality is by the dual characterization of the operator norm. The first inequality is by the fact that the operator norm of a Pauli operator is $1$. The second inequality is by relaxing $F_{{Z},i}$ to $1$ and noticing that $|p\land q| = |p\lor q| - |p| - |q|$. The last inequality uses the $k$-local property of $O$.
		
		\medskip
		
		\noindent The first factor of Eq.~\eqref{eq:local_shadow_norm2} can be bounded using the same method as in~\cite{huang2020predicting}. We reproduce their proof here for the convenience of the reader.
		{Without loss of generality, suppose $O$ is supported on the first $k$ qubits, and hence can be written as $O = \tilde O \otimes I_{2^{n-k}}$.} The decomposition of $\tilde O$ is denoted as 
		\begin{equation}
		    \tilde O = \sum_{p\in\mbb Z_2^{2k}}\tilde\alpha_p P_p. 
		\end{equation}
		For any two $q,s\in\mbb Z_2^{2n}$ we write $q\triangleright s$ if one can obtain $P_q$ from $P_s$ by replacing some single-qubit Paulis of $P_s$ with $I$. 
		Then,
		\begin{equation}\label{eq:local_1term}
		\begin{aligned}
		\sum_{p,q\in\mbb Z_{2}^{2n}} \delta(p,q) 3^{|p\land q|}|\alpha_p||\alpha_q| &= \sum_{p,q\in\mbb Z_{2}^{2k}} \delta(p,q) 3^{|p\land q|}|\tilde\alpha_p||\tilde\alpha_q|\\
		&= \frac{1}{3^k}\sum_{P_s\in\{P_X,P_Y,P_Z\}^{\otimes k}}\left(\sum_{q:q\triangleright s}3^{|q|}|\tilde\alpha_q|\right)^2\\
		&\le \frac{1}{3^k}\sum_{P_s\in\{P_X,P_Y,P_Z\}^{\otimes k}} \left(\sum_{q:q\triangleright s}3^{|q|}\right)\left(\sum_{q:q\triangleright s}3^{|q|}|\tilde \alpha_q|^2\right)\\
		&= 4^k \sum_{P_s\in\{P_X,P_Y,P_Z\}^{\otimes k}}\sum_{q:q\triangleright s}3^{|q|-k}|\tilde \alpha_q|^2\\
		&= 4^k \sum_{q\in\mbb Z_2^{2k}}|\tilde \alpha_q|^2\\
		&= 2^k \Tr(\tilde O^2)\le 4^k \|\tilde O\|_\infty^2= 4^k \|O\|_\infty^2
		\end{aligned}
		\end{equation}
		where in the first equality we restrict our attention to the first $k$ qubits, the second equality can be verified by checking the coefficients of every $|\tilde\alpha_p||\tilde\alpha_q|$, the first inequality is by Cauchy-Schwarz inequality, the third and fourth equality is by simple combinatoric arguments.
		For the last line, the first equation follows from the definition of $\tilde O$, the inequality follows from the relationship between the Hilbert-Schmidt norm and the operator norm, and the last equality is by the fact that the largest eigen value of $O$ equals to that of $\tilde O$.

		
		On the other hand, suppose the preceding calibration procedure guarantees $\hat f_z\ge\delta f_z$ for all $|z|\le k$ for some postive number $\delta$ close to $1$. Then the second term of Eq.~\eqref{eq:local_shadow_norm2} can be bounded as follows by Proposition~\ref{prop:lc_var},
		\begin{equation}\label{eq:E23}
		\begin{aligned}
		\max_{|z|\le k}\frac{\hat f_z^{-2}}{3^{2|z|}} \le \delta^{-2}\max_{|z|\le k}\frac{ f_z^{-2}}{3^{2|z|}} = \delta^{-2}\max_{|z|\le k} \Gamma_\Lambda(z)^{-2}
		\end{aligned}
		\end{equation}
		Since $\Lambda$ is assumed to be local noise, we have the following bound for $\Gamma_\Lambda(z)$, which could be better than Lemma~\ref{le:lc_weak},
		\begin{lemma}\label{le:lc_weak_local}
		Suppose $\Lambda\coleq\bigotimes_{i=1}^n\Lambda_i$ and satisfies $F_{Z}(\Lambda_i)\ge 1-\xi$ for all $i\in[n]$ and some $0\le\xi< \frac12$, then
	    \begin{equation}
	    \Gamma_\Lambda(z)\ge (1-2\xi)^{|z|},\quad \forall z\in\{0,1\}^n.    
	    \end{equation}
	    \end{lemma}
		\begin{proof}
			\begin{equation}
			\begin{aligned}
			\Gamma_\Lambda(z) &= \cfrac1{2^n}\sum_{x,\delta\in\{0,1\}^n}(-1)^{z\cdot\delta}\lbra {x\oplus\delta}\Lambda\lket x \\
			& = \cfrac1{2^n} \prod_{i=1}^n\sum_{x,\delta\in\{0,1\}}(-1)^{z_i\cdot \delta}\lbra{x\oplus\delta}\Lambda_i\lket{x}\\
			& = \cfrac1{2^{|z|}} \prod_{i:z_i=1}\sum_{x,\delta\in\{0,1\}}(-1)^\delta\lbra{x\oplus\delta}\Lambda_i\lket{x}\\
			&= \prod_{i:z_i=1} (\sum_{x\in\{0,1\}}\lbra x\Lambda\lket x - 1)\\
			&= (2F_{Z}(\Lambda_i) - 1)^{|z|}\\
			&\ge (1-2\xi)^{|z|}.
			\end{aligned}
			\end{equation}
			where the third equality is by the fact that $\Lambda_i$ is trace-preserving, and hence $\sum_{x,\delta\in\{0,1\}}\lbra {x\oplus\delta}\Lambda\lket x = 2$, so we can eliminate those indexes $i$ such that $z_i = 0$.		
		\end{proof}
		
		Combine Lemma~\ref{le:lc_weak_local} with Eq.~\eqref{eq:E23}, we get the following lemma: (Note that we substitute $O$ with its traceless part $O_0$ in order to use Lemma~\ref{le:snorm} later.)


		\begin{lemma}\label{le:snorm_lc}
		For \textbf{\textrm{RShadow}} using ${\sf Cl}_2^{\otimes n}$, suppose the noise is local, \textit{i.e.} $\Lambda\coleq\bigotimes_{i=1}^n\Lambda_i$, and satisfies $F_{Z}(\Lambda_i)\ge 1-\xi$ for all $i\in[n]$ and some $0\le\xi< \frac12$. Then, if the calibration procedure guarantees $\hat f_z \ge \delta f_z$ for all $|z|\le k$ and some $\delta>0$, we have
		\begin{equation}
		\|O_0\|_{\mr{shadow},\Lambda}^2 \le \delta^{-2}(1-2\xi)^{-2k}~4^k \|O\|_\infty^2,
		\end{equation}
		for any $k$-local observable $O$.
	\end{lemma}
		Compared to Proposition~2 from~\cite{huang2020predicting} that $\|O_0\|_\mr{shadow}^2\le 4^k \|O\|_\infty^2$, we conclude that, when the separable noise channel $\Lambda$ has not too low Z-basis fidelity per qubit and the noise calibration procedure is conducted sufficiently many rounds, the estimation procedure of our \textbf{\textrm{RShadow}} protocol using ${\sf Cl}_2^{\otimes n}$ is as efficient as the noiseless standard quantum shadow estimation protocol~\cite{huang2020predicting} up to a small multiplicative factor. That is to say, expectation value of any observable $O$ {located on a $k$-qubit subsystem} can be efficiently estimated.
		
		To complete the discussion, we give the following theorem as a rigorous version of Theorem~\ref{th:lc_all} in the main text.
		
		\begin{theorem}\label{th:lc_allr}
			For \textbf{\textrm{RShadow}} with ${\sf Cl}_2^{\otimes n}$, suppose the noise is local, \textit{i.e.} $\Lambda\coleq\bigotimes_{i=1}^n\Lambda_i$, and satisfies $F_{Z}(\Lambda_i)\ge 1-\xi$ for all $i\in[n]$ and some $0\le\xi< \frac12$. Then, if the number of calibration samples $R_C$ and the number of estimation samples $R_E$ satisfies
			\begin{equation}
			\begin{aligned}
			R_C & = 68\cdot 3^{k}\left(1+\frac{2^k}{\varepsilon_1}\right)^2\left(k\ln n+\ln2\delta^{-1}\right)   (1-2\xi)^{-2k},\\
			R_E &= \frac{34}{\varepsilon_2^2}\cdot 4^k \ln(2M/\delta_2) {(1+\varepsilon_1)^2}(1-2\xi)^{-2k},
			\end{aligned}
			\end{equation}
			respectively, then the protocol can estimate $M$ arbitrary linear functions $\Tr(O_1\rho),...\Tr(O_M\rho)$ such that $\|O_i\|_\infty\le 1$ and that $O_i$ is $k$-local, up to accuracy $\varepsilon_1+\varepsilon_2$ with success probability at least $1-\delta_1-\delta_2$. 
		\end{theorem}

		\begin{proof}
			First, according to Theorem~\ref{th:lc_stat}, for the given number of samples $R_C$ one have
			\begin{equation}\label{eq:eps1_lc}
			\left| \mbb E(\hat o_i^{(r)}) - \Tr(O_i\rho) \right| \le \varepsilon_1.
			\end{equation}
			Note that we apply the bound for $\Gamma_\Lambda(z)$ from Lemma~\ref{le:lc_weak_local}. 
			
			\noindent Meanwhile, from the proof of Theorem~\ref{th:lc_stat} (see Eq.~\eqref{eq:gamma_epsilon}), one also have
			\begin{equation}
			|\hat f_z^{-1}f_z-1|\le \varepsilon_1 \Rightarrow \hat f_z\ge (1+\varepsilon_1)^{-1} f_z,\quad \forall |z|\le k.
			\end{equation}
			Both equations hold simultaneously with probability at least $1-\delta_1$. 
			
			\smallskip
			
			\noindent Now, by Lemma~\ref{le:snorm} and Lemma~\ref{le:snorm_lc}, the single-round estimators in the estimation procedure satisfy
			\begin{equation}
			\mr{Var}(\hat o_i^{(r)}) \le4^k {(1+\varepsilon_1)^2}(1-2\xi)^{-2k}.
			\end{equation}
			So we set the median of mean estimators $\hat o_i$ of the estimation procedure with the following parameters:
			\begin{equation}
			N=\frac{34}{\varepsilon_2^2}\cdot 4^k {(1+\varepsilon_1)^2}(1-2\xi)^{-2k},\quad K=2\ln(2M/\delta_2).
			\end{equation}
			Then Lemma~\ref{le:median} combined with the union bound gives that the following holds for all $i$ with probability at least $1-\delta_2$:
			\begin{equation}\label{eq:eps2_lc}
			\left| \hat o_i - \mbb E(\hat o_i^{(r)}) \right| \le \varepsilon_2.
			\end{equation}
			Combining Eq.~\eqref{eq:eps1_lc} and Eq.~\eqref{eq:eps2_lc} using the triangular inequality gives
			\begin{equation}
			|\hat o_i - \Tr(O_i\rho)|\le \varepsilon_1+\varepsilon_2,
			\end{equation}
			which holds with probability at least $1-\delta_1-\delta_2$.
			This completes the proof.
		\end{proof}
		
		Specifically, if $\xi \ll \frac12$ then $(1-2\xi)^{-2k} = \left((1-2\xi)^{-\frac1{2\xi}}\right)^{4k\xi} \approx e^{4k\xi}$. That is how we get the bound in Theorem~\ref{th:lc_all}.
		
		\section{The effect of state preparation noise} \label{sec:app_spn}

		In this section, we will prove Theorem~\ref{th:sp_g} and Theorem~\ref{th:sp_l} in the main text establishing the robustness of \textbf{\textrm{RShadow}} against state preparation noise in the calibration procedure. Let's first fix the notations: We assume $\ket{\bm0}$ is experimentally prepared as some other state $\rho_{\bm0}$ which is fixed over time, and we will use a subscript ``SP'' to denote the state-preparation noisy version of our estimators. For example, $\widehat \cM_\mr{SP} = \sum_{\lambda\in R_{\mathbb G}}\hat f_{\lambda,\mr{SP}}\Pi_{\lambda}$ is our estimation for the physical channel $\widetilde\cM \coleq \sum_{\lambda\in R_{\mathbb G}}f_\lambda\Pi_{\lambda}$ when the calibration process suffers from state preparation error. 
		
		\subsection{Robustness of RShadow with Global Clifford Group}
		
		\begin{lemma}\label{le:sp_g}
		For \textbf{\textrm{RShadow}} using ${\sf Cl}(2^n)$, if the state-preparation fidelity satisfies
		\begin{equation}
		F(\ketbra{\bvec 0}{\bvec 0},\rho_{\bvec 0}) \ge 1-\varepsilon_\mr{SP},
		\end{equation}
		then the SP-noisy single-round estimator $\hat f_{SP}^{(r)}$ satisfies
		\begin{equation}
		\begin{aligned}
		&f\ge \mbb E(\hat f_\mr{SP}^{(r)}) \ge (1-2\varepsilon_\mr{SP})f,\\
		&\mr{Var}(\hat f_\mr{SP}^{(r)}) \le \frac{6d}{(d-1)^3}.
		\end{aligned}
		\end{equation}
		\end{lemma}
		\begin{proof}
		    According to the calibration procedure described in Algorithm~\ref{Proto:robust} or Protocol~\ref{proto:Global} of App.~\ref{sec:app_gl} , we have
			\begin{equation}
			\begin{aligned}
			\mbb E(\hat F_\mr{SP}^{(r)}) &= \mbb E_{U\sim {\sf Cl}(2^n)}\sum_{b} \lbra{\bvec 0}{\cU}^\dagger\lket b\lbra b \Lambda {\cU}\lket{\rho_0}\\
			& = \lbra{\bvec 0} ~\left[\lket{\sigma_{\bvec 0}}\lbra{\sigma_{\bvec 0}}+f(I-\lket{\sigma_{\bvec 0}}\lbra{\sigma_{\bvec 0}})\right]~\lket{\rho_0}\\
			& = \frac1d+f(\bra{\bvec 0}\rho_{\bm0}\ket{\bvec 0}-\frac1d).\\
			\mbb E(\hat f_\mr{SP}^{(r)}) &= \frac{d\mbb E(\hat F_\mr{SP}^{(r)})-1}{d-1}\\
			&= \frac{d\bra{\bvec 0}\rho_{\bm0}\ket{\bvec 0}-1}{d-1}f\\
			&\ge (1-\varepsilon_\mr{SP}\frac{d}{d-1})f.
			\end{aligned}
			\end{equation}
			One can immediately conclude that $f\ge\mbb E(\hat f^{(r)}_{SP})\ge(1-2\varepsilon_\mr{SP})f$.  
			
			\medskip
			
			\noindent The second moment of $\hat F^{(r)}_{SP}$ can be written as (see Eq.~\eqref{eq:global_2ndmoment})
			\begin{equation}
			\begin{aligned}
			\mbb E(\hat F_\mr{SP}^{(r)^2}) &= \sum_{b\in\{0,1\}^n} \Tr\left[ \mbb E_{U\sim {\sf Cl}(2^n)}\left(U\rho_0U^\dagger\otimes U\ket{\bvec 0}\bra{\bvec 0}U^\dagger\otimes U\ket{\bvec 0}\bra{\bvec 0}U^\dagger\right) ~\left(\Lambda^\dagger(\ket b\bra b)\otimes \ket b\bra b\otimes \ket b\bra b\right)\right]\\
			&= \sum_{b\in\{0,1\}^n} \Tr\left[ \Phi_\mr{Haar}^{(3)}(\rho_0\otimes \ketbra{\bvec 0}{\bvec 0}\otimes \ketbra{\bvec 0}{\bvec 0}) ~\left(\Lambda^\dagger(\ket b\bra b)\otimes \ket b\bra b\otimes \ket b\bra b\right)\right]\\
			&= \sum_{b\in\{0,1\}^n}\sum_{\pi,\sigma\in S_3}c_{\pi,\sigma}\Tr\left[W_\pi  (\rho_0\otimes \ketbra{\bvec 0}{\bvec 0}\otimes \ketbra{\bvec 0}{\bvec 0})\right]\Tr\left[W_\sigma \left(\Lambda^\dagger(\ket b\bra b)\otimes \ket b\bra b\otimes \ket b\bra b\right) \right]\\
			&= \frac{2(d-2F_{Z} -2F_0+2dF_{Z} F_0)}{(d^2-1)(d+2)}\\
			&\le \frac{6d}{(d^2-1)(d+2)},
			\end{aligned}
			\end{equation}
			where we define $F_0\coleq\bra{\bvec 0}\rho_0\ket{\bvec 0}$ and $F_{Z}\coleq F_{Z}(\Lambda)$. Therefore,
			\begin{equation}
			\mr{Var}(\hat f_\mr{SP}^{(r)}) = \frac{d^2}{(d-1)^2}\mr{Var}(\hat F_\mr{SP}^{(r)})\le \frac{d^2}{(d-1)^2}\mbb{E}(\hat F_\mr{SP}^{(r)^2})\le \frac{6d}{(d-1)^3}.
			\end{equation}
		\end{proof}
		
	The following theorem is a more detailed formalisation of Theorem~\ref{th:sp_g} in the main text.
		
    \begin{theorem}
	For \textbf{\textrm{RShadow}} using $\sf{Cl}(2^n)$, if the state-preparation fidelity satisfies
	\begin{equation}
	    F(\ket{\bm 0}\bra{\bm 0},\rho_{\bm 0})\ge 1-\varepsilon_{\mr{SP}},
	\end{equation}
	then with $R=\tilde\cO(\varepsilon^{-2}F_Z^{-2})$ calibration samples, the subsequent estimation procedure with high probability satisfies
	\begin{equation}
    \left| \mbb E(\hat o^{(r)}) - \Tr(O\rho)  \right| \le (\varepsilon+2\varepsilon_\mr{SP}) \|O\|_\infty.
	\end{equation}
	up to the first order of $\varepsilon$ and $\varepsilon_\mr{SP}$ for any observable $O$. We have assumed $F_Z\coleq F_Z(\Lambda)\gg 1/d$.
	\end{theorem}
	
	\begin{proof}
	    First notice that the target function can be upper bounded as
	    \begin{equation}
	   \begin{aligned}
	   \left| \mbb E(\hat o^{(r)}) - \Tr(O\rho)  \right| &= \left| \lbra{O}\widehat \cM_\mr{SP}^{-1}\widetilde\cM-1\lket{\rho}\right| \\
	   &= \left| \lbra{O_0}\widehat \cM_\mr{SP}^{-1}\widetilde\cM-1\lket{\rho}\right|\\
	   & \le \left|\lbraket{O_0}{\rho}\right|\cdot \left| \hat f_\mr{SP}^{-1}f-1\right|\\
	   & \le \left\|O\right\|_\infty \cdot \left| \hat f_\mr{SP}^{-1}f-1\right|.
	   \end{aligned}
	    \end{equation}
	    According to Lemma~\ref{le:median}, by taking the parameters of the median of mean estimators as
	    \begin{equation}
	    \begin{aligned}
	    N &= 34\mr{Var}(\hat f^{(r)}_\mr{SP})\varepsilon^{-2}f^{-2},\\
	    K &= 2\ln(2\delta^{-1}),
	    \end{aligned}
	    \end{equation}
	    the following holds with probability at least $1-\delta$,
	    \begin{equation}
	        \left| \hat f_\mr{SP} - \mbb E(\hat f_\mr{SP}^{(r)}) \right| \le \varepsilon f.
	    \end{equation}
	    We also have, from Lemma~\ref{le:sp_g}, that
	    \begin{equation}
	        \left| \mbb E(\hat f_\mr{SP}^{(r)}) - f \right| \le 2\varepsilon_\mr{SP}f.
	    \end{equation}
	    Therefore, our final bound is as claimed
	    \begin{equation}
	    \begin{aligned}
	    \left| \mbb E(\hat o^{(r)}) - \Tr(O\rho)  \right| &\le \|O\|_\infty\cdot \frac{|f-\hat f_\mr{SP}|}{|\hat f_\mr{SP}|}\\
	    &\le \|O\|_\infty\cdot\frac{\varepsilon+2\varepsilon_\mr{SP}}{1-\varepsilon-2\varepsilon_\mr{SP}}\\
	    &=\|O\|_\infty\cdot(\varepsilon+2\varepsilon_\mr{SP}+ o(\varepsilon+2\varepsilon_\mr{SP})).
	    \end{aligned}
	    \end{equation}
	    The sample complexity is
	    \begin{equation}
	        R = NK \le 2\ln(2\delta^{-1})\cdot 204 \varepsilon^{-2}(F_Z-1/d)^{-2}\frac{(d+1)^2}{d(d-1)} = \tilde\cO(\varepsilon^{-2}F_Z^{-2}),
	    \end{equation}
	    for $F_Z\coleq F_Z(\Lambda)\gg 1/d$. Here we have used Lemma~\ref{le:sp_g} and Proposition~\ref{prop:gl_main} to bound $\mr{Var}(\hat f^{(r)}_\mr{SP})$ and $f$, respectively.
	\end{proof}

\subsection{Robustness of RShadow with Local Clifford Group}
		
Note that, we consider a local state-preparation noise model  for the results in this section, \textit{i.e.}, no cross-talk between qubits.
		
\begin{lemma}\label{le:sp_l}
For \textbf{\textrm{RShadow}} using ${\sf Cl}_2^{\otimes n}$, if the prepared state is in a product form, \textit{i.e.}, $\rho_{\bm 0} = \bigotimes_{i=1}^n\rho_{0,i}$, and the single-qubit state-preparation fidelity satisfies
\begin{equation}
	F(\ketbra{0}{0},\rho_{0,i}) \ge 1-\xi_\mr{SP},\quad \forall~i\in[n],
\end{equation}
for some $\xi_\mr{SP}<1/2$, then the SP-noisy single-round estimator $\hat f^{(r)}_{z,\mr{SP}}$ satisfies
\begin{equation}
\begin{aligned}
&f_z\ge \mbb E(\hat f^{(r)}_{z,\mr{SP}}) \ge (1-2 \xi_\mr{SP}|z|)f_z,\\
&\mr{Var}(\hat f^{(r)}_{z,\mr{SP}})\le 3^{-|z|},\quad\forall z\in\{0,1\}^{n}.
\end{aligned}
\end{equation}
\end{lemma}
		
	\begin{proof}
	    According to the calibration procedure described in Algorithm~\ref{Proto:robust} or Protocol~\ref{proto:local} of App.~\ref{sec:app_lc} , we have
		\begin{equation}\label{eq:d6}
		\begin{aligned}
		\mbb E(\hat f^{(r)}_{z,\mr{SP}}) &= \mbb E_{U\sim {\sf Cl}_2^{\otimes n}}\sum_{b} \lbra{P_z}{\cU}^\dagger\lket b\lbra b \Lambda {\cU}\lket{\rho_{\bm0}}\\
		&= \lbra{P_z} \sum_{m\in\{0,1\}^n}f_m\Pi_m \lket{\rho_{\bm0}}\\
		&= f_z\lbraket{P_z}{\rho_{\bm0}}\\
		&= f_z \prod_{i:z_i=1}(2\bra{0}\rho_{0,i}\ket{0}-1)\\
		&\ge (1-2|z|\,\xi_\mr{SP})f_z.
		\end{aligned}
		\end{equation}
		One can immediately conclude that $f_z\ge\mbb E(\hat f^{(r)}_{z,\mr{SP}})\ge(1-2|z|\,\xi_\mr{SP})f_z$.  
		
		\medskip
		
		\noindent To calculate the second moment, 
		\begin{equation}
		\mbb E(\hat f^{(r)^2}_{z,\mr{SP}}) = \sum_b\tr\left[\mbb E_{U\sim {\sf Cl}_2^{\otimes n}} \left(U\rho_{\bm0} U^\dagger\otimes UP_z U^\dagger \otimes UP_z U^\dagger\right)\left(\Lambda^\dagger(\ket b\bra b)\otimes\ket b\bra b\otimes\ket b\bra b\right) \right],
		\end{equation}
		we can first investigate the single-qubit case:
		\begin{equation}\label{eq:spn_local_1qubit}
		\begin{aligned}
		\mbb E_{U\sim {\sf Cl}_2}\left(U\rho_{0,i} U^\dagger\otimes U P_I U^\dagger \otimes U P_I U^\dagger\right)&= \frac12 I_2^{\otimes 3},\\  
		\mbb E_{U\sim {\sf Cl}_2}\left(U\rho_{0,i} U^\dagger\otimes UP_Z U^\dagger \otimes UP_Z U^\dagger\right)&=\Phi^{(3)}_\mr{Haar}(\rho_{0,i}\otimes P_Z\otimes P_Z),\\
		\end{aligned}
		\end{equation}
		To further simplify the second expressions, one can verify that
		\begin{equation}
		\Tr(\vec W (\rho_{0,i}\otimes P_Z\otimes P_Z)) = \begin{bmatrix}
		0 & 0 & 0 & 2 & 1 & 1
		\end{bmatrix},
		\end{equation}
		where $\vec{W}$ is defined in Eq.~\eqref{eq:perm_vec}. 
		Calculating the Haar integral using Eq.~\eqref{eq:Weingarten}, one immediately notice that the form of $\rho_{0,i}$ has nothing to do with the result. So we can safely replace all $\rho_{0,i}$ with $\ket 0\bra 0$ and retrieve the result with no state preparation error: $\mbb E(\hat f^{(r)^2}_{z, \mr{SP}}) = \mbb E(\hat f_z^{(r)^2})=3^{-|z|}$ and hence $\mr{Var}(\hat f^{(r)}_{z,\mr{SP}})\le 3^{-|z|}$. 
	\end{proof}

	The following theorem is a more detailed formalisation of Theorem~\ref{th:sp_l} in the main text.

    \begin{theorem}
	For \textbf{\textrm{RShadow}} using $\sf{Cl}_2^{\otimes n}$, if the state is prepared as some {product state} $\rho_{\bm 0}=\bigotimes_{i=1}^n\rho_{0,i}$ and the single-qubit state-preparation fidelity satisfies
	\begin{equation}
	    F(\ket{0}\bra{0},\rho_{0,i})\ge 1-\xi_\mr{SP},\quad\forall i\in[n],
	\end{equation}
	then with $R = \tilde \cO(3^k\varepsilon^{-2}F_Z^{-2})$ calibration samples, the subsequent estimation procedure with high probability satisfies
	\begin{equation}
    \left| \mbb E(\hat o^{(r)}) - \Tr(O\rho)  \right| \le (\varepsilon+2 k\xi_{SP})2^k \|O\|_\infty.
	\end{equation}
	up to the first order of $\varepsilon$ and $k\xi_{SP}$, for any $k$-local observable $O$.
	\end{theorem}
	
	\begin{proof}
	Suppose $O$ is a $k$-local observable for some $k$.
	Following exactly the same procedure as in the proof of Theorem~\ref{th:lc_stat} (see Eq.~\eqref{eq:c23}), we
	can bound our target function as follows,
	\begin{equation}
	\begin{aligned}
	\left| \mbb E(\hat o^{(r)}) - \Tr(O\rho)  \right| &\le 2^k\|O\|_\infty\cdot\max_{|z|\le k}\left|\hat f^{-1}_{z,\mr{SP}}f_z-1\right|.
	\end{aligned}
	\end{equation}
	
    \noindent According to Lemma~\ref{le:median}, by taking the parameters of the median of mean estimators as
    \begin{equation}
    \begin{aligned}
    N &= \max_{|z|\le k} 34\mr{Var}(\hat f^{(r)}_{z,\mr{SP}})\varepsilon^{-2}f_z^{-2},\\
    K &= 2\ln(2(\delta/n^k)^{-1}),
    \end{aligned}
    \end{equation}
    the following holds with probability at least $1-\delta/n^k$ for any $z$ whose weight is no larger than $k$, hence simultaneously holds for all such $z$ with probability at least $1-\delta$ by the union bound:
    \begin{equation}
        \left| \hat f_{z,\mr{SP}} - \mbb E(\hat f^{(r)}_{z,\mr{SP}}) \right| \le \varepsilon f_z,\quad \forall~z\in\{0,1\}^n:|z|\le k.
    \end{equation}
    We also have, from Lemma~\ref{le:sp_l}, that
    \begin{equation}
        \left| \mbb E(\hat f_{z,\mr{SP}}^{(r)}) - f_z \right| \le 2\xi_\mr{SP}|z|f_z,\quad \forall~z\in\{0,1\}^n.
    \end{equation}
    Therefore, our final bound is as claimed
    \begin{equation}
    \begin{aligned}
    \left| \mbb E(\hat o^{(r)}) - \Tr(O\rho)  \right| &\le 2^k\|O\|_\infty\cdot \max_{|z|\le k} \frac{|f_z-\hat f_\mr{z,SP}|}{|\hat f_\mr{z,SP}|}\\
    &\le 2^k\|O\|_\infty\cdot\frac{\varepsilon+2k\xi_\mr{SP}}{1-\varepsilon-2k\xi_\mr{SP}}\\
    &=2^k\|O\|_\infty\cdot(\varepsilon+2k\xi_\mr{SP}+ o(\varepsilon+2k\xi_\mr{SP})).
    \end{aligned}
    \end{equation}
    The sample complexity is
    \begin{equation}
    \begin{aligned}
        R = NK &\le 2\ln(2\delta^{-1}n^k) \cdot 34\cdot3^k \varepsilon^{-2} \Gamma_z(\Lambda)^{-2} \\
        &\le 2\ln(2\delta^{-1}n^k) \cdot 34\cdot3^k \varepsilon^{-2} F_Z(\Lambda)^{-2} \\
        &= \tilde\cO(3^k\varepsilon^{-2}F_Z^{-2}),
    \end{aligned}
    \end{equation}
    where we have used Lemma~\ref{le:sp_l} and Proposition~\ref{prop:lc_var} to bound $\mr{Var}(\hat f^{(r)}_{z,\mr{SP}})$ and $f_z$, respectively. The second inequality is by Lemma~\ref{le:lc_weak}. We remark that one can alternatively use a stronger bound given in Lemma~\ref{le:lc_weak_local} when the noise model is assumed to be local. 
	\end{proof}

	\section{More numerical results} \label{sec:app_numer}
		
	\subsection{Coherent and correlated noise with the local Clifford group}
		
	Here, we present more numerical results to show the performance of \textbf{\textrm{RShadow}} in the task of estimating the (average) 2-point ZZ correlation function of the GHZ state using the local Clifford group. More precisely, the quantity we want to estimate can be expressed as \begin{equation}\label{eq:Zcorr}
	    \frac{1}{n-1}\sum_{i=1}^{n-1}\bra{\mr{GHZ}_n}Z_iZ_{i+1}\ket{\mr{GHZ}_n},
	\end{equation}
	the true value of which is obviously $1$.
	The aim of this appendix is to close one gap between the theory we derived in the main text and practical needs: whether \textrm{\textbf{RShadow}} using the local Clifford group is still sample-efficient against qubitwise-correlated noises. Our numerical results give an affirmative answer.
	Here, we work with $5$-qubit GHZ state. We perform the estimation task under two different noise models: single-qubit $X$-axis rotation and two-qubit $XX$ cross-talk noises. When the single-qubit $X$-axis rotation error happens, each qubit will experience a coherent rotation after the implementation of the random unitary gate $R_X(\theta) = e^{-i\theta X}$,
	where $\theta = \frac{k\pi}{40}$, $k=0,1,2,3,4,5$. When the two-qubit $XX$ cross-talk noise happens, each two adjacent qubits will experience a coherent rotation after the implementation of the random unitary gate $R_{XX}(\theta) = e^{-i\theta XX}$, where $\theta = \frac{3k\pi}{100}$, $k=0,1,2,3,4,5$. 
	For clarity, we estimate the 2-point correlation function of the $5$-th qubit with every other qubit and plot the average values. 

\begin{figure}[!htbp]
\centering
\includegraphics[width = 0.5\columnwidth]{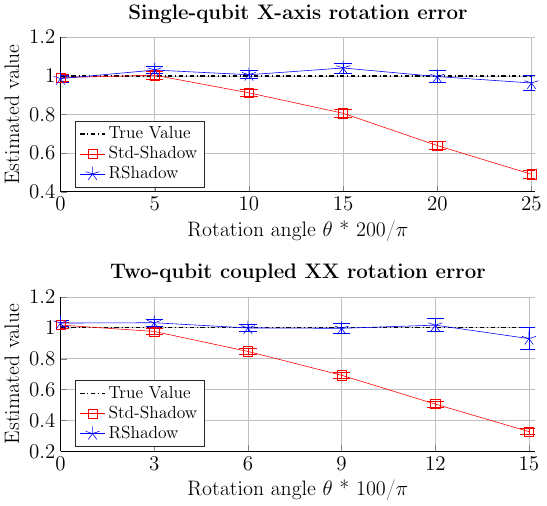}
\caption{$5$-qubit GHZ 2-point correlation function estimation under the coherent noise models, including single-qubit $X$-axis rotation and two-qubit $XX$-coupling noise.}
\label{fig:local_Xrot}
\end{figure}

\subsection{Gate-dependent noise: more details}\label{sec:app_gd}

Here, we present more details about the gate-dependent noise simulations in Sec.~\ref{sec:gatedependent} of the main text.

\medskip

\emph{Gate decomposition}: In our simulations, we decompose each single-qubit Clifford gate using the following three generators
\begin{equation}
    \left\{R_P\left(\frac\pi 2\right) = \exp(-i\frac{\pi}{4}P),~P=X,Y,Z\right\}.
\end{equation}
To see how this works, note that each single-qubit Clifford gate can be uniquely specified by its conjugation on Pauli Z and Pauli X. In notations, any single-qubit Clifford gate $C$ can be equivalently described by
\begin{equation}
    \{st[C] \coleq CZC^\dagger,\quad de[C] \coleq CXC^\dagger\},
\end{equation}
where $st[C]$ and $de[C]$ are both one of $\{\pm X,\pm Y,\pm Z\}$. They are usually called the stabilizer and de-stabilizer of $C$, respectively. 
Another way to say this is that $C$ rotates the 3d Cartesian coordinates so that the $+X$ ($+Z$) direction is now at the $st[C]$ ($de[C]$) direction.

Now, we can decompose $C$ into two subsequent rotations: First rotate the $+Z$ direction into the direction specified by $st[C]$ along either the $X$ or $Y$ axis. Then rotate the \emph{new} $+X$ direction into $de[C]$ along the $st[C]$ axis. Each rotation is of angle $\{0,\frac\pi2,\pi,\frac{3\pi}{2}\}$, so it can be implemented by concatenating up to three generators.

As an example mentioned in the main text, the Hadamard gate has the following stabilizer and de-stabilizer
$$
\{st[H] = +X,\quad de[H] = +Z\},
$$
so it can be decomposed as $H = R_X^2(\frac\pi2)R_Y(\frac{\pi}{2})$, as shown in Fig.~\ref{fig:hadamard}.

\begin{figure}[!htp]
    \centering
    \includegraphics[width=0.6\columnwidth]{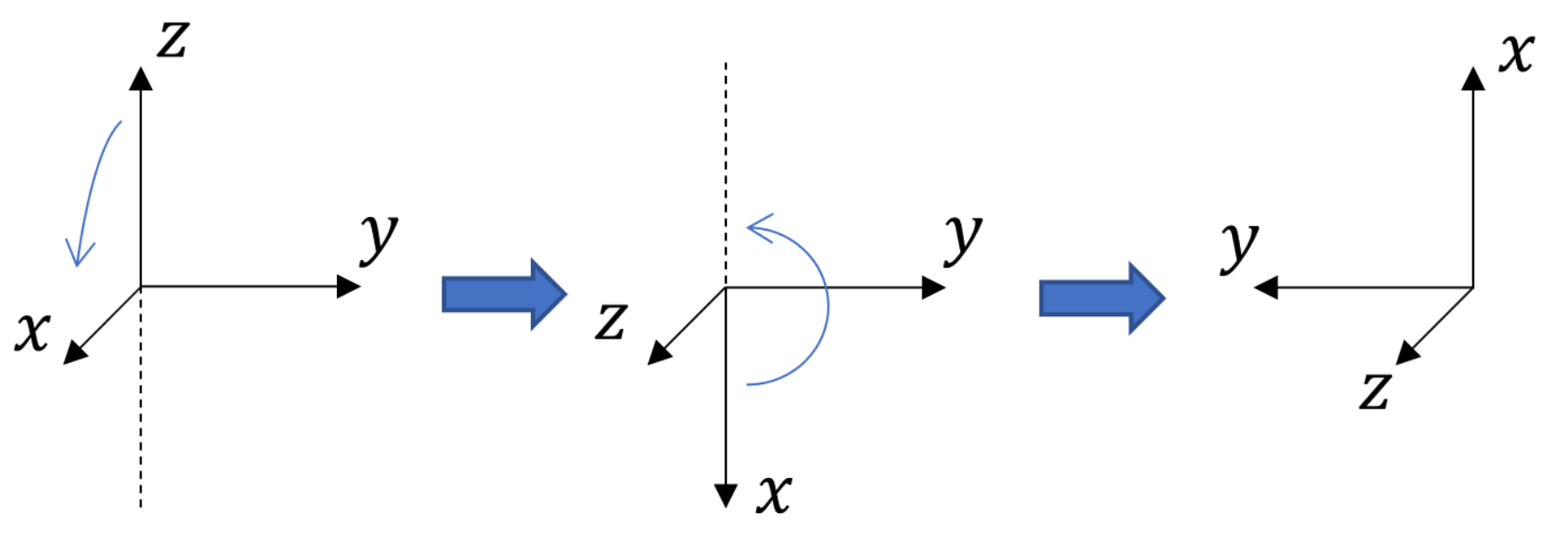}
    \caption{Rotational decomposition of the Hadamard gate.}
    \label{fig:hadamard}
\end{figure}

\medskip

\emph{Gate-dependence of the noise model}:
We claim in the main text that both \emph{pulse mis-calibration} and \emph{random over-rotation} are gate-dependent noise model. Here we explain in more details. The noisy generators for \emph{pulse mis-calibration} error are
$$
\{\widetilde{R}_P(\frac\pi2) = \exp(-i\frac12(\frac\pi2 P+\Delta_0)),~P=X.Y.Z\}.
$$
By the Baker-Campbell-Hausdorff formula, this can be expanded as
$$
\widetilde{R}_P(\frac\pi2) = \exp(-i\frac{\pi}{4}P)\exp(-i\frac12\Delta_0)\exp(\frac{\pi}{16}[{P},{\Delta_0}])\cdots.
$$
As long as $[P,\Delta_0]$ is different for different $P$, the noise channel is different for different generators. We also notice that the gate-dependent effect only appears in higher-order terms.

\medskip

For \emph{random over-rotation} error, the noisy generators are
$$
\widetilde{R}_P = \exp(-i\frac12(\frac\pi2+\delta)P).
$$
For a generator $R_P$, the noise channel for this noise model can be written as
\begin{equation}
    \begin{aligned}
        \Lambda_P(\rho) &= \int d\delta~ p(\delta;\sigma)R_P(\delta)\rho R_P^\dagger(\delta)\\
        &= \int d\delta~ p(\delta;\sigma)e^{-\frac{i}{2}\delta P}\rho e^{\frac{i}{2}\delta P}
    \end{aligned}
\end{equation}
where $p(\delta;\sigma) = \frac{1}{\sigma\sqrt{2\pi}}\exp(-\frac{\delta^2}{2\sigma^2})$ is the Gaussian distribution. Let $\ket{P\pm}$ denote the $\pm 1$ eigen-vectors of $P$ and carry out the Gaussian integral, we get
\begin{equation}
    \begin{aligned}
    \Lambda_P(\rho) &= \ket{P+}\bra{P+}\bra{P+}\rho\ket{P+} + \ket{P-}\bra{P-}\bra{P-}\rho\ket{P-} \\&+ e^{-\frac12\sigma^2}\ket{P+}\bra{P-}\bra{P+}\rho\ket{P-} + e^{-\frac12\sigma^2}\ket{P-}\bra{P+}\bra{P-}\rho\ket{P+}
    \end{aligned}
\end{equation}
This is a dephasing channel on the $\ket{P\pm}$ basis. Since this basis is different for different Pauli $P$, this noise model is gate-dependent by definition.

\medskip

\emph{Different directions of $\Delta_0$}:
In the numerical results of Fig.~\ref{fig:gd_main}, we fixed a random direction for $\Delta_0$ and only varied its magnitude. Here we show that there is nothing special about the direction we picked. We uniformly sample $10$ unit vectors $\bm v$ and set $\Delta_0$ to be $0.1\pi \times (\bm v_1X+\bm v_2Y+\bm v_3Z)$. For each of these noise settings, we use $R = 30000~(N=3000,~K=10)$ calibration samples and $R = 10000~(N=1000,~K=10)$ estimation samples for \textbf{RShadow}, and $R = 10000~(N=1000,~K=10)$ samples for standard \textbf{Shadow}. The performance is shown in Fig.~\ref{fig:sphere} in which the average and standard deviation is calculated over $10$ independent runs. For all cases, \textbf{RShadow} gives a more accurate estimate than standard \textbf{Shadow}.

\begin{table}[!htp]
    \centering
    \begin{tabular}{|c|c|}
    \hline
    Label of samples & $\bm v$ (direction of $\Delta_0$)\\
    \hline
    1&$(-0.550, +0.288, +0.784)$ \\
    2&$(-0.086, +0.987, -0.136)$ \\
    3&$(+0.652, -0.396, +0.647)$ \\
    4&$(+0.589, +0.183, +0.787)$ \\
    5&$(+0.543, -0.781, -0.309)$ \\
    6&$(-0.666, -0.720, -0.196)$ \\
    7&$(+0.174, -0.512, -0.841)$ \\
    8&$(+0.908, +0.417, +0.034)$ \\
    9&$(-0.183, +0.935, -0.304)$ \\
    10&$(+0.599, -0.704, -0.382)$ \\
    \hline
    \end{tabular}
    \caption{Directions of $\Delta_0$ that are uniformly randomly sampled from a sphere}
    \label{tab:my_label}
\end{table}

\begin{figure}[!htbp]
	\centering
    \includegraphics[width = 0.5\columnwidth]{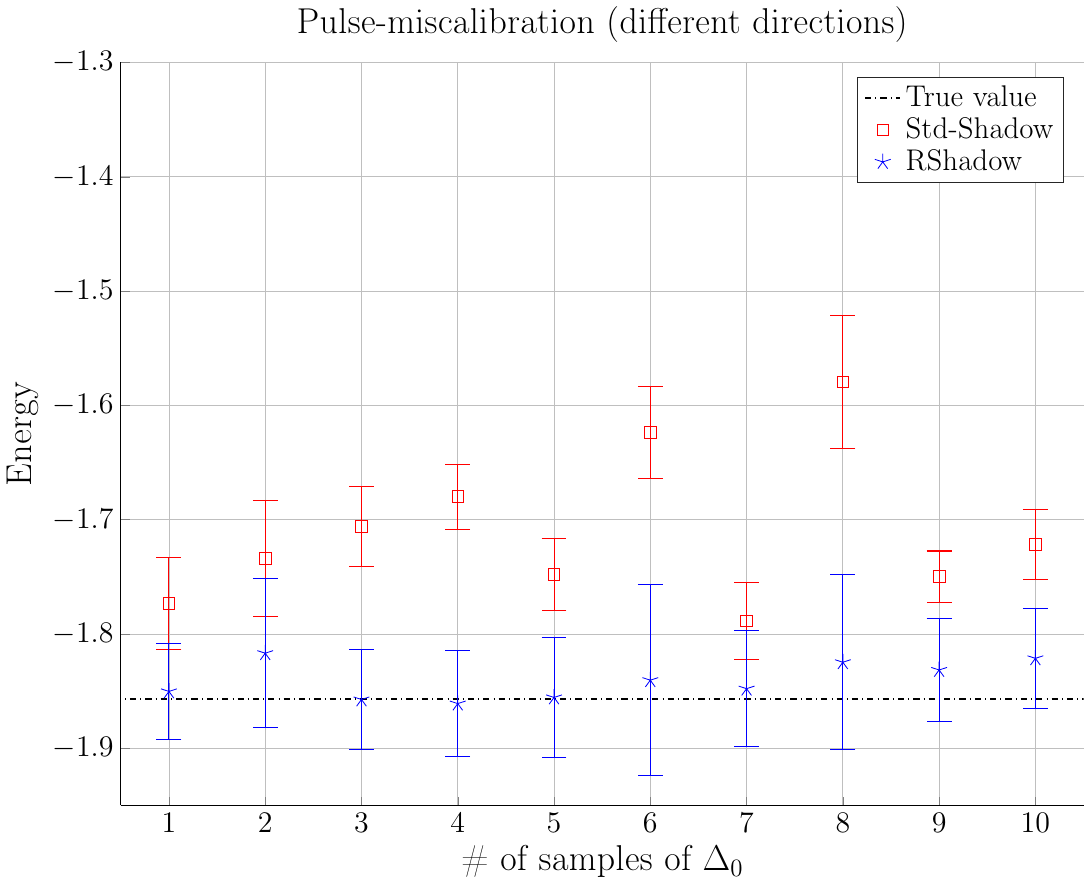}
    \caption{Ground-state energy estimation of $\mathrm{H}_2$ with pulse mis-calibration noise. Here we choose $\Delta_0 = 0.1\pi \times (\bm v_1X+\bm v_2Y+\bm v_3Z)$, where each different label of x axis corresponds to a different sample of $\bm v$ given in Table~\ref{tab:my_label}}.
	\label{fig:sphere}
\end{figure}

\end{appendix}
	
\end{document}